\documentclass[reqno]{amsart}

\usepackage[a4paper]{geometry}
\usepackage[T2A]{fontenc}
\usepackage[utf8]{inputenc}
\usepackage[ukrainian,russian,english]{babel}
\usepackage{comment}

\usepackage{amssymb}

\usepackage[mathscr]{eucal}
\usepackage{tikz}
\usetikzlibrary{positioning,calc,decorations.markings}
\usetikzlibrary{patterns}
\usetikzlibrary{calc}
\usepackage{mathtools}
\usepackage{enumerate}
\usepackage{accents}
\usepackage{dsfont}
\usepackage{multicol}

\usepackage{color}
\usepackage[colorlinks=true]{hyperref}
\hypersetup{urlcolor=blue,citecolor=red,linkcolor=blue}
\usepackage[initials,nobysame]{amsrefs}
\definecolor{darkgreen}{rgb}{0.13, 0.55, 0.13}

\numberwithin{equation}{section}
\theoremstyle{plain}
\newtheorem{theorem}{Theorem}[section]
\newtheorem{proposition}[theorem]{Proposition}
\newtheorem{corollary}[theorem]{Corollary}
\newtheorem{lemma}[theorem]{Lemma}
\theoremstyle{definition}
\newtheorem*{rh-pb*}{Basic RH problem}
\newtheorem*{rh-I*}{RH problem normalized as $I$ at $\infty$}
\newtheorem*{sol-rh-pb*}{Soliton RH problem}
\newtheorem*{data*}{Data of this RH problem associated with $\BS{u_0(x)}$}
\theoremstyle{remark}
\newtheorem{remark}[theorem]{Remark}
\newtheorem*{notations*}{Notations}
\newtheorem*{assumptions*}{Assumptions}
\providecommand{\BS}[1]{\boldsymbol{#1}}  
\newcommand{\R}{\mathbb{R}}

\newcommand{\dd}{\mathrm{d}}
\newcommand{\eul}{\mathrm{e}}
\newcommand{\ii}{\mathrm{i}}
\newlength{\dhatheight}
\newcommand{\doublehat}[1]{%
    \settoheight{\dhatheight}{\ensuremath{\hat{#1}}}%
    \addtolength{\dhatheight}{-0.3ex}%
    \hat{\vphantom{\rule{1pt}{\dhatheight}}%
    \smash{\hat{#1}}}}

\renewcommand{\Im}{\operatorname{Im}}

\newcommand{\ord}{\mathrm{O}}
\DeclareMathOperator{\Res}{Res}

\newif\ifshort
\shorttrue

\title[A RH approach to the two-component mCH system with step-like boundary conditions]{A Riemann--Hilbert approach to the two-component modified Camassa--Holm system with step-like boundary conditions}

\author[Iryna Karpenko]{Iryna Karpenko}
\address{IK: Faculty of Mathematics\\ University of Vienna\\
Oskar-Morgenstern-Platz 1\\ 1090 Wien\\ Austria\\ and B. Verkin Institute for Low Temperature Physics and Engineering\\ 47, Nauky ave\\ 61103 Kharkiv\\ Ukraine}
\email{\href{mailto:iryna.karpenko@univie.ac.at}{iryna.karpenko@univie.ac.at}}

\author[Dmitry Shepelsky]{Dmitry Shepelsky}
\address{DS: B.Verkin Institute for Low Temperature Physics and Engineering, Kharkiv, Ukraine\\ and V.N.Karazin Kharkiv National University, Ukraine}
\email{\href{mailto:shepelsky@yahoo.com }{shepelsky@yahoo.com}}

\author[Gerald Teschl]{Gerald Teschl}
\address{GT: Faculty of Mathematics\\ University of Vienna\\
Oskar-Morgenstern-Platz 1\\ 1090 Wien\\ Austria}
\email{\href{mailto:gerald.teschl@univie.ac.at}{gerald.teschl@univie.ac.at}}


\begin{document}

\begin{abstract}
We develop a Riemann--Hilbert formalism for the Cauchy problem for the
two-component modified Camassa--Holm (2-mCH) system on the real line
under non-zero step-like boundary conditions. Starting from the
corresponding Lax pair, we carry out the direct spectral analysis and
establish the analytic, symmetry, and asymptotic properties of the
eigenfunctions and spectral data determined by the initial conditions.
These results are then used to formulate right and left matrix Riemann--Hilbert problems and derive
parametric formulas for reconstructing both components. These formulas
also use a transported profile determined by the initial ratio of the
two momenta. Under suitable regularity and positivity assumptions, we
show that a solution of the right Riemann--Hilbert problem, together
with this profile, yields a local solution of the 2-mCH system.
For a common constant background, we construct a family of smooth
reflectionless solutions associated with a single pair of discrete
eigenvalues. The scalar soliton, obtained when the profile vanishes,
is the only travelling wave in this family.
\end{abstract}

\maketitle

\section{Introduction}

The modified Camassa--Holm (mCH) equation
\begin{equation}\label{eq:mch}
 m_t+\bigl((u^2-u_x^2)m\bigr)_x=0,
 \qquad
 m:=u-u_{xx},
\end{equation}
is an integrable equation with cubic nonlinearity and may be viewed as
a modification of the Camassa--Holm equation \cite{CH93}. It arose in
several equivalent integrability constructions developed by
Fokas, Fuchssteiner, and Olver and Rosenau
\cites{Fokas95,Fuchssteiner96,OR96}. A Lax representation and explicit
cuspon and W/M-shaped solutions were obtained by Qiao in
\cite{Qiao06}. For this reason, \eqref{eq:mch} is also referred to as the
Fokas--Olver--Rosenau--Qiao (FORQ) equation.

The Riemann--Hilbert approach to the mCH equation with a nonzero
constant background was developed in \cites{BKS20,Karpenko2022}. The case of
different nonzero asymptotic states at the two spatial ends was
subsequently studied in \cite{KST22}, where the direct spectral theory
and the associated Riemann--Hilbert problem were formulated in the
presence of inner spectral cuts.

In this paper, we study the two-component integrable modified
Camassa--Holm (2-mCH) equation, introduced by Song, Qu, and Qiao in
\cite{SQQ11} as a natural extension of \eqref{eq:mch} to two interacting
fields:
\begin{subequations}\label{eq:2mch}
\begin{alignat}{4}           
& m_t+\left((u-u_x)(v+v_x)m\right)_x = 0,&\quad&m:= u-u_{xx},&\quad&t>0,&\;&-\infty<x<+\infty,\\
& n_t+\left((u-u_x)(v+v_x)n\right)_x = 0,&\quad&n:= v-v_{xx},&\quad&t>0,&\;&-\infty<x<+\infty.
\end{alignat}
\end{subequations}

It was shown in \cite{SQQ11} that the 2-mCH system describes
pseudospherical surfaces.   This system of equations is known to possess infinitely many conservation laws as well as a Lax formulation:
\begin{subequations}
    \label{eq:lax}
    \begin{equation}
        \Psi_x(x,t,\lambda)=\mathsf{U}(x,t,\lambda)\Psi(x,t,\lambda),\qquad\Psi_t(x,t,\lambda)=\mathsf{V}(x,t,\lambda)\Psi(x,t,\lambda),
    \end{equation}
    where the coefficient matrices $\mathsf{U}$ and $\mathsf{V}$ are defined in terms of a solution of the 2-mCH equation:
    \begin{equation}
        \mathsf{U}=\frac{1}{2}\begin{pmatrix} -1 & \lambda  m \\
-\lambda n & 1\end{pmatrix},~
\mathsf{V}= \begin{pmatrix}
  \displaystyle \lambda^{-2}
    +\frac{(u-u_x)(v+v_x)}{2}
  &\displaystyle -\lambda^{-1}(u-u_x)
    -\frac{\lambda}{2}(u-u_x)(v+v_x)m\\[6pt]
  \displaystyle \lambda^{-1}(v+v_x)
    +\frac{\lambda}{2}(u-u_x)(v+v_x)n
  &\displaystyle -\lambda^{-2}
    -\frac{(u-u_x)(v+v_x)}{2}
 \end{pmatrix}.
    \end{equation}
\end{subequations}
A notable feature of \eqref{eq:2mch} is that passing from the scalar
mCH equation to its two-component extension does not increase the
dimension of the auxiliary spectral problem. For comparison, the
$N$-component vector NLS equation is associated with an
$(N+1)\times(N+1)$ Lax pair; in particular, the two-component
Manakov system \cite{Manakov74} leads to a $3\times3$ spectral and
RH formulation \cites{PAB06,PBT11}. Similarly, the standard Lax representation of the matrix NLS equation
with an $N\times N$ matrix-valued potential involves a
$2N\times2N$ spectral problem \cite{PDLH18}. By contrast, both the Lax pair
\eqref{eq:lax} and the RH problems for the 2-mCH system remain
$2\times2$. The coupling between the components is encoded in the two
off-diagonal momenta $m$ and $n$, rather than through an enlargement
of the matrix spectral problem.

A geometric formulation in terms of pseudospherical surfaces and
explicit cuspon and W/M-shaped solutions were also presented in
\cite{SQQ11}. The local well-posedness for the associated Cauchy problem in the Besov spaces, explicit expressions of its single peakon and two peakon solutions, and blow-up
scenario were studied in \cite{YanQiaoZhang2015}. Interlacing multipeakons were constructed by
inverse spectral methods in \cite{CHS16}, where their global existence
and long-time asymptotics were also analyzed.

We consider the Cauchy problem for \eqref{eq:2mch} subject to the initial conditions
\begin{equation}\label{eq:initial-data}
 u(x,0)=u_0(x),\qquad v(x,0)=v_0(x),\qquad x\in\R,
\end{equation}
assuming that
\begin{equation}\label{eq:step-data}
 \begin{aligned}
  (u_0(x),v_0(x))& \to  (A_1,A_1),&&x\to-\infty,\\
  (u_0(x),v_0(x))& \to  (A_2,A_2),&&x\to+\infty,
 \end{aligned}
 \qquad 0<A_1<A_2.
\end{equation}
and that the solution $(u(x,t),v(x,t))$ preserves this behavior for all fixed $t>0$: 
\begin{equation}\label{eq:step-solution}
 \begin{aligned}
  (u(x,t),v(x,t))& \to  (A_1,A_1),&&x\to-\infty,\\
  (u(x,t),v(x,t))& \to  (A_2,A_2),&&x\to+\infty.
 \end{aligned}
\end{equation}
The case with different
background values is briefly discussed in Section~\ref{sec:different-component-backgrounds}.

We also impose the positivity condition
\begin{equation}\label{eq:positivity}
 m_0(x):=u_0(x)-u_{0,xx}(x)>0,
 \qquad
 n_0(x):=v_0(x)-v_{0,xx}(x)>0;
\end{equation}
then it can be shown that  $m(x,t)>0$ and $n(x,t)>0$ for all $t$ (see Appendix \ref{app:A}, for the case of the mCH equation, see \cites{KST22}).

The paper aims at extending the ``step-like Riemann--Hilbert
formalism'' of \cite{KST22} to the two-component system
\eqref{eq:2mch}.
We focus on establishing representation results, assuming the existence of a solution to the Cauchy problem.

For this purpose, we introduce a balancing gauge for the two momenta,
develop the direct spectral theory, and analyze the Jost solutions and
scattering data, including their symmetries, discrete spectrum, and
behavior at the branch points and at $\lambda=0$. We then formulate
the associated right and left Riemann--Hilbert problems.

A key distinction from the scalar case appears in the reconstruction
of the solution. Setting
\[
 r=\frac14\log\frac{m}{n},
\]
the expansion of the Riemann--Hilbert solutions at $\lambda=0$
determines the combination $x+2r$, rather than the inverse variable
$x$ alone. The function $r$ is transported by the common
characteristic flow, so its initial profile, or equivalently the
initial ratio $m_0/n_0$, must be retained in order to separate $x$
and $r$ and recover both components of the solution. We derive
parametric recovery formulas that incorporate this additional
two-component information.

The outline of the paper is as follows.
Section~\ref{sec:2} develops the gauge transformations and direct
spectral analysis. The right and left Riemann--Hilbert problems and
the relation between them are presented in
Section~\ref{sec:ic-rh-problems}. Section~\ref{sec:ic-two-component-reconstruction}
derives the recovery formulas from both formulations. Section~\ref{sec:different-component-backgrounds} describes the
modifications required when the two components of the solution $(u(x,t),v(x,t))$ have different
background values at the same spatial end.
Appendix \ref{app:A} discusses a sign-preserving property of $m$ and $n$.
Appendix~\ref{sec:ic-rh-to-local-2mch} constructs a local solution
of \eqref{eq:2mch} from a sufficiently regular solution of the right
RH problem and the initial profile of $r$.
Appendix~\ref{app:rh-local-reflectionless} constructs reflectionless
solutions on a common constant background for a single pair of
discrete points. The scalar soliton, obtained when the initial profile of $r$
vanishes, is the only travelling wave in this family.

\section{Notations and Assumptions}
\begin{notations*}
Throughout the paper, $\sigma_1,\sigma_2,\sigma_3$ stand for the standard
Pauli matrices,
\[
\sigma_1:=
\begin{pmatrix}0&1\\1&0\end{pmatrix},
\qquad
\sigma_2:=
\begin{pmatrix}0&-\ii\\ \ii&0\end{pmatrix},
\qquad
\sigma_3:=
\begin{pmatrix}1&0\\0&-1\end{pmatrix}.
\]
For any $2\times2$ matrices $A$ and $B$, we denote their commutator by
$[A,B]:=AB-BA$ and, for $\alpha\in\mathbb C$, we write
\[
e^{\alpha\hat\sigma_3}A:= e^{\alpha\sigma_3}Ae^{-\alpha\sigma_3}.
\]
We denote the upper and lower half-planes by
\[
\mathbb C^\pm:=\{\lambda\in\mathbb C:\ \pm\Im \lambda>0\},
\]
and set
\[
 \mathbb R_1:=(-\infty,0),
 \qquad
 \mathbb R_2:=(0,+\infty).
\]
For $j=1,2$, we introduce the intervals
\[
\Sigma_j:=
\left(-\infty,-\frac{1}{A_j}\right]
\cup
\left[\frac{1}{A_j},\infty\right),
\qquad
\dot{\Sigma}_j:=
\left(-\infty,-\frac{1}{A_j}\right)
\cup
\left(\frac{1}{A_j},\infty\right),
\]
\[
\Sigma_0:=
\left[-\frac{1}{A_1},-\frac{1}{A_2}\right]
\cup
\left[\frac{1}{A_2},\frac{1}{A_1}\right],
\qquad
\dot{\Sigma}_0:=
\left(-\frac{1}{A_1},-\frac{1}{A_2}\right)
\cup
\left(\frac{1}{A_2},\frac{1}{A_1}\right),
\]
and
\[
\Gamma_j:=
\left[-\frac{1}{A_j},\frac{1}{A_j}\right],
\qquad
\dot{\Gamma}_j:=
\left(-\frac{1}{A_j},\frac{1}{A_j}\right).
\]
The functions
\[
l_j(\lambda):=\sqrt{\lambda^2-\frac{1}{A_j^2}},
\qquad j=1,2,
\]
are defined with branch cuts along $\Gamma_j$ and normalized by
$l_j(\lambda)\sim\lambda$ as $\lambda\to\infty$.

For $\lambda$ on a real cut, $\lambda_+$ and $\lambda_-$ denote the
points on its upper and lower sides, respectively. If a scalar-, vector-,
or matrix-valued function $F$ admits the corresponding boundary values,
we write
\[
 F_\pm(\lambda)
 \equiv F(\lambda_\pm)
 :=
 \lim_{\varepsilon\downarrow0}
 F(\lambda\pm\ii\varepsilon).
\]

For a function $f$, its Schwarz conjugate is denoted by
\[
 f^*(\lambda):=\overline{f(\bar\lambda)}.
\]
For a vector or matrix $A$, we denote its Hermitian transpose by
\[
 A^\dagger:=\overline{A}^{\,\mathsf T}.
\]

For a scalar-, vector-, or matrix-valued function $F$ with a simple pole at a
point $\lambda_k$ on a cut, we denote its residues from the two sides by
\begin{equation}\label{ic:one-sided-residues}
 \Res_{\lambda=\lambda_{k\pm}}F
 :=
 \lim_{\substack{\lambda\to\lambda_k\\
                  \pm\Im\lambda>0}}
 (\lambda-\lambda_k)F(\lambda).
\end{equation}

Finally, for an interval
$\Omega\subseteq\mathbb R$, we denote the Sobolev space $W^{k,p}(\Omega)$ by
\[
W^{k,p}(\Omega)
=
\left\{
f\in L^p(\Omega):
D^j f\in L^p(\Omega),\quad j=1,\ldots,k
\right\}.
\]

\end{notations*}

\begin{assumptions*}

Throughout the article, $(u,v)$ denotes a given real-valued classical
solution of \eqref{eq:2mch} on $\mathbb R\times[0,T)$, for some
$T\leq \infty$, satisfying
\eqref{eq:initial-data}--\eqref{eq:positivity}. Set
\begin{equation}\label{ic:Q}
Q:=(u-u_x)(v+v_x).    
\end{equation}
We assume that
\[
 m,n,Q\in C\bigl([0,T);C_b^1(\mathbb R)\bigr).
\]
By Proposition~\ref{prop:positivity}, condition \eqref{eq:positivity}
implies
\[
 m(x,t)>0,
 \qquad
 n(x,t)>0,
 \qquad
 (x,t)\in\mathbb R\times[0,T).
\]
We may therefore define
\begin{equation}\label{ic:def-rhoRr}
 \rho:=\sqrt{mn},
 \qquad
  \nu :=\frac{m}{n},
 \qquad
 r:=\frac14\log\frac{m}{n}.
\end{equation}

We assume that, uniformly for $t$ in compact subintervals of $[0,T)$,
\begin{equation}\label{ic:rhoR-limits}
 \begin{aligned}
  \rho(x,t)& \to  A_1,
  &\qquad
   \nu (x,t)& \to 1,
  &\qquad
  r(x,t),\,r_x(x,t)& \to 0,
  &&x\to-\infty,\\
  \rho(x,t)& \to  A_2,
  &
   \nu (x,t)& \to 1,
  &
  r(x,t),\,r_x(x,t)& \to 0,
  &&x\to+\infty.
 \end{aligned}
\end{equation}
In particular, $m$ and $n$ are uniformly bounded above and away from zero
on $\mathbb R\times[0,T_0]$ for every $T_0<T$.

Set
\begin{equation}\label{ic:delta-def}
 \delta_j
 :=
 \frac{\rho-A_j}{2A_j}-r_x,
 \qquad j=1,2.
\end{equation}
For every $T_0<T$, we assume that
\begin{equation}\label{ic:standing-direct-assumptions}
\begin{aligned}
 \sup_{0\leq t\leq T_0}
 \sum_{j=1}^2
 \Bigg(
 &\bigl\|(1+|\cdot|)^2
       \bigl(\rho(\cdot,t)-A_j\bigr)
  \bigr\|_{L^1(\mathbb R_j)}
 +\bigl\|(1+|\cdot|)^2r_x(\cdot,t)
  \bigr\|_{L^1(\mathbb R_j)}
 +
 \left\|
  \frac{\delta_j(\cdot,t)}{\rho(\cdot,t)}
 \right\|_{W^{1,1}(\mathbb R_j)}
 \Bigg)<\infty.
\end{aligned}
\end{equation}

\end{assumptions*}

\section{Lax pairs and eigenfunctions}\label{sec:2}

Introduce
\begin{equation}\label{ic:def-pqQ}
 p:=u-u_x,\qquad q:=v+v_x
\end{equation}
so that
\begin{equation}\label{eq:identities}
 m=p+p_x,\qquad n=q-q_x,\qquad Q=pq, \qquad Q_x=mq-np.
\end{equation}

\begin{lemma}\label{lem:rho-r-conservation}
The functions $\rho$ and $r$ satisfy
\begin{equation}\label{ic:rho-r-conservation}
 \rho_t+(Q\rho)_x=0,
 \qquad
 r_t+Qr_x=0.
\end{equation}
\end{lemma}

\begin{proof}
Expanding the conservation laws in \eqref{eq:2mch}, we obtain
\[
 m_t+Qm_x+Q_xm=0,
 \qquad
 n_t+Qn_x+Q_xn=0.
\]
Since $m,n>0$, division by $m$ and $n$, respectively, yields
\[
 (\log m)_t+Q(\log m)_x+Q_x=0,
 \qquad
 (\log n)_t+Q(\log n)_x+Q_x=0.
\]
Adding these identities and using
$\log\rho=\frac12(\log m+\log n)$, we obtain
 $(\log\rho)_t+Q(\log\rho)_x+Q_x=0$,
which is equivalent to the first equation in
\eqref{ic:rho-r-conservation}, while subtracting the two identities and using
$r=\frac14(\log m-\log n)$ yields the second equation.
\end{proof}

The RH formulation is built from solutions of the Lax pair with controlled
analytic behavior in the spectral parameter. Since $\mathsf U$ and
$\mathsf V$ are traceless, the determinant of any matrix solution of
\eqref{eq:lax} is independent of $x$ and $t$. The points $\lambda=0$ and
$\lambda=\infty$ require separate treatment: $\mathsf U$ is singular at
infinity, whereas
\[
 \mathsf U(x,t,0)=-\frac12\sigma_3
\]
is independent of the solution. This structure motivates the use of the
expansion at $\lambda=0$ in
the recovery formulas for the 2-mCH equation (cf. \cites{BS06,KST22}).

\subsection{Gauge transformations}

To control the large-$\lambda$ behavior of the eigenfunctions, we use a
two-stage gauge procedure. First, a common diagonal gauge balances the
momenta $m$ and $n$, replacing them by their geometric mean
$\rho=\sqrt{mn}$ in the leading spatial coefficient. The second stage
uses two different gauges: $H_1$ is adapted to the background $A_1$ at
$x=-\infty$, while $H_2$ is adapted to the background $A_2$ at
$x=+\infty$. Each gauge diagonalizes the corresponding limiting spatial
coefficient. After the resulting diagonal phase is extracted, the
transformed coefficients remain bounded as $\lambda\to\infty$.

We first set
\begin{equation}\label{ic:G-def}
 G:=\operatorname{diag}( \nu ^{-1/4}, \nu ^{1/4})
   =\eul^{-r\sigma_3},
 \qquad
 \Psi^{\mathrm b}:=G\Psi.
\end{equation}
Direct substitution into \eqref{eq:lax} gives the following transformed
Lax pair.

\begin{proposition}
\label{prop:balanced-lax-pair}
The function $\Psi^{\mathrm b}$ satisfies
\begin{subequations}\label{ic:balanced-lax}
    \begin{equation}\label{ic:balanced-lax-pair}
 \Psi^{\mathrm b}_x=U^{\mathrm b}\Psi^{\mathrm b},
 \qquad
 \Psi^{\mathrm b}_t=V^{\mathrm b}\Psi^{\mathrm b},
\end{equation}
where
\begin{equation}\label{ic:balanced-U}
 U^{\mathrm b}
 =
 \frac12
 \begin{pmatrix}
  -1&\lambda\rho\\
  -\lambda\rho&1
 \end{pmatrix}
 -r_x\sigma_3
\end{equation}
and
\begin{equation}\label{ic:balanced-V}
 V^{\mathrm b}
 =
 \begin{pmatrix}
  \lambda^{-2}+\dfrac Q2+Qr_x
  &
  -\lambda^{-1}p \nu ^{-1/2}
  -\dfrac{\lambda}{2}Q\rho
  \\[2mm]
  \lambda^{-1}q \nu ^{1/2}
  +\dfrac{\lambda}{2}Q\rho
  &
  -\lambda^{-2}-\dfrac Q2-Qr_x
 \end{pmatrix},
\end{equation}
where $\rho$, $r$, and $ \nu $ are defined in \eqref{ic:def-rhoRr}, and $Q$ is defined in \eqref{ic:Q}.
\end{subequations}

\end{proposition}

The second stage consists of two transformations, one for each asymptotic
background.

For $j=1,2$, introduce $H_j$, which diagonalizes the limiting
spatial coefficient associated with $A_j$:
\begin{equation}\label{ic:H-def}
 H_j(\lambda)
 :=
 \sqrt{\frac12}
 \sqrt{\frac1{\ii A_jl_j(\lambda)}-1}
 \begin{pmatrix}
  \dfrac{\lambda A_j}{1-\ii A_jl_j(\lambda)}&-1\\[3mm]
  -1&\dfrac{\lambda A_j}{1-\ii A_jl_j(\lambda)}
 \end{pmatrix},
\end{equation}
with inverse
\begin{equation}\label{ic:H-inverse}
 H_j^{-1}(\lambda)
 :=
 \sqrt{\frac12}
 \sqrt{\frac1{\ii A_jl_j(\lambda)}-1}
 \begin{pmatrix}
  \dfrac{\lambda A_j}{1-\ii A_jl_j(\lambda)}&1\\[3mm]
  1&\dfrac{\lambda A_j}{1-\ii A_jl_j(\lambda)}
 \end{pmatrix}.
\end{equation}
Here, the square root is taken on $\mathbb C\setminus[0,\infty)$ with
$\sqrt{-1}=\ii$. This normalization
implies $\det H_j=1$. Moreover,
\begin{equation}\label{ic:H-identities}
\begin{aligned}
 H_j(\lambda)\sigma_1H_j^{-1}(\lambda)=\sigma_1,\quad
 H_j(\lambda)\sigma_3H_j^{-1}(\lambda)
 =\frac{\lambda}{l_j(\lambda)}\sigma_2
   +\frac1{\ii A_jl_j(\lambda)}\sigma_3,\quad
 H_j(\lambda)\sigma_2H_j^{-1}(\lambda)
 =\frac1{\ii A_jl_j(\lambda)}\sigma_2
   -\frac{\lambda}{l_j(\lambda)}\sigma_3,
\end{aligned}
\end{equation} 
and
\begin{equation}\label{ic:H-diagonalization}
 H_j(\lambda)
 \begin{pmatrix}
  -1&\lambda A_j\\
  -\lambda A_j&1
 \end{pmatrix}
 H_j^{-1}(\lambda)
 =
 -\ii A_jl_j(\lambda)\sigma_3.
\end{equation}

\begin{remark}\label{rem:ic-H-at-zero}
The denominator $1-\ii A_jl_j(\lambda)$ in
\eqref{ic:H-inverse} vanishes at $\lambda=0$ when $\lambda$ approaches
zero from $\mathbb C^-$. Nevertheless, $H_j^{-1}$ remains finite. Indeed,
\[
 l_j(\lambda)=
 \begin{cases}
  \displaystyle
  \frac{\ii}{A_j}-\frac{\ii A_j}{2}\lambda^2
  +\ord(\lambda^4),
  &\lambda\in\mathbb C^+,\\[2mm]
  \displaystyle
  -\frac{\ii}{A_j}+\frac{\ii A_j}{2}\lambda^2
  +\ord(\lambda^4),
  &\lambda\in\mathbb C^-,
 \end{cases}
\]
and therefore
\[
 1-\ii A_jl_j(\lambda)=
 \begin{cases}
  \displaystyle
  2-\frac{A_j^2}{2}\lambda^2+\ord(\lambda^4),
  &\lambda\in\mathbb C^+,\\[2mm]
  \displaystyle
  \frac{A_j^2}{2}\lambda^2+\ord(\lambda^4),
  &\lambda\in\mathbb C^-.
 \end{cases}
\]
At the same time,
\[
 \sqrt{\frac1{\ii A_jl_j(\lambda)}-1}
 =
 \begin{cases}
  \sqrt{2}\,\ii+\ord(\lambda^2),
  &\lambda\in\mathbb C^+,\\[1mm]
  \displaystyle
  -\frac{A_j}{\sqrt{2}}\lambda+\ord(\lambda^3),
  &\lambda\in\mathbb C^-.
 \end{cases}
\]
Thus, in $\mathbb C^-$, the zero of the scalar factor cancels the pole
arising from $1-\ii A_jl_j(\lambda)$. Substitution into
\eqref{ic:H-inverse} yields
\begin{equation}\label{H_at_0}
 H_j^{-1}(\lambda)=
 \begin{cases}
  \displaystyle
  \ii\sigma_1+\frac{\ii A_j}{2}\lambda I
  +\ord(\lambda^2),
  &\lambda\to0,\quad\lambda\in\mathbb C^+,\\[3mm]
  \displaystyle
  -I-\frac{A_j}{2}\lambda\sigma_1
  +\ord(\lambda^2),
  &\lambda\to0,\quad\lambda\in\mathbb C^-.
 \end{cases}
\end{equation}
Hence $H_j^{-1}$ admits finite boundary values at $\lambda=0$ from both
half-planes, and these boundary values are distinct.
\end{remark}

To extract the resulting diagonal term, define
\begin{equation}\label{ic:f-def}
 f_j(x,t,\lambda)
 :=
 \ii A_jl_j(\lambda)
 \left[
  \frac1{2A_j}
  \int_{(-1)^j\infty}^{x}
  \bigl(\rho(\xi,t)-A_j\bigr)\dd\xi
  +\frac x2
  -t\left(
    \frac1{\lambda^2}+\frac{A_j^2}{2}
   \right)
 \right].
\end{equation}
Then Lemma \ref{lem:rho-r-conservation} implies
\begin{equation}\label{ic:f-derivatives}
 f_{j,x}(x,t,\lambda)=\frac{\ii l_j(\lambda)\rho(x,t)}{2},
 \qquad
 f_{j,t}(x,t,\lambda)
 =-\ii l_j(\lambda)
 \left(
  \frac{Q(x,t)\rho(x,t)}{2}+\frac{A_j}{\lambda^2}
 \right).
\end{equation}

\begin{remark}
The limits in \eqref{ic:rhoR-limits} imply that
$m,n\to A_j$ as $x\to(-1)^j\infty$. Since
\[
 p(x,t)=\int_{-\infty}^{x}\eul^{y-x}m(y,t)\,\dd y,
 \qquad
 q(x,t)=\int_x^\infty\eul^{x-y}n(y,t)\,\dd y,
\]
it follows that
\[
 p,q \to  A_j,\qquad
 Q \to  A_j^2,\qquad
 Q\rho \to  A_j^3,
 \qquad x\to(-1)^j\infty.
\]
\end{remark}

Combining \eqref{ic:H-diagonalization}, \eqref{ic:H-identities}, and \eqref{ic:f-derivatives}, we obtain the
following transformed Lax pairs.

\begin{proposition}
\label{prop:inner-cut-transformed-lax}
Let
\begin{equation}\label{ic:hatPsi-def}
 \widehat\Psi_j(x,t,\lambda):=H_j(\lambda)\Psi^{\mathrm b}(x,t,\lambda).
\end{equation}
Then
\begin{subequations}\label{ic:transformed-lax}
\begin{align}
 \widehat\Psi_{jx}
 +f_{j,x}\sigma_3\widehat\Psi_j
 &=\widehat W_j\widehat\Psi_j,\\
 \widehat\Psi_{jt}
 +f_{j,t}\sigma_3\widehat\Psi_j
 &=\widehat V_j\widehat\Psi_j,
\end{align}
where
\begin{equation}\label{ic:W-def}
 \widehat W_j
 =
 \delta_j
 \left(
  \frac{\lambda}{l_j}\sigma_2
  +\frac1{\ii A_jl_j}\sigma_3
 \right),
\end{equation}
and
\begin{equation}\label{ic:hatV-def}
\begin{aligned}
 \widehat V_j
 &=
 -\frac1{A_jl_j}
 \left[
  \frac{
   p \nu ^{-1/2}+q \nu ^{1/2}-2A_j
  }{2\lambda}
  +A_j\lambda Q\delta_j
 \right]\sigma_2
 +\frac{
  q \nu ^{1/2}-p \nu ^{-1/2}
 }{2\lambda}\sigma_1
 \\
 &\quad
 -\frac1{\ii A_jl_j}
 \left[
  \frac{A_j}{2}
  \bigl(
   p \nu ^{-1/2}+q \nu ^{1/2}-2A_j
  \bigr)
  +Q\delta_j
 \right]\sigma_3.
\end{aligned}
\end{equation}
\end{subequations}
In particular,
\[
 \widehat W_j(x,t,\lambda),~
 \widehat V_j(x,t,\lambda)
 \to 0,
 \qquad
 x\to(-1)^j\infty,
\]
and as $\lambda\to\infty$, both matrices remain bounded and their diagonal
parts are of order $\ord(\lambda^{-1})$.
\end{proposition}

\begin{remark}
\label{rem:scalar-reduction}
Under the reduction $u=v$, one has
\[
 m=n,\qquad
  \nu =1,\qquad
 r=0,\qquad
 \rho=m,
\]
and
\[
 \delta_j=\frac{m-A_j}{2A_j}.
\]
Substituting these identities into \eqref{ic:W-def} and
\eqref{ic:hatV-def} recovers the transformed coefficients for the scalar
mCH equation obtained in \cite{KST22}.
\end{remark}

\subsection{Eigenfunctions and scattering functions}

The Lax pair \eqref{ic:transformed-lax} can be used to construct, via the
corresponding integral equations, particular solutions whose
behavior for large values of the spectral parameter $\lambda$
is nicely controlled. Namely, we introduce
\begin{equation}\label{ic:tildePsi-def}
 \widetilde\Psi_j
 :=\widehat\Psi_j\eul^{R_j},\qquad R_j(x,t,\lambda):=f_j(x,t,\lambda)\sigma_3.
\end{equation}
In terms of $\widetilde\Psi_j$, the Lax pair \eqref{ic:transformed-lax} takes the form
\begin{equation}\label{ic:commutator-system}
\begin{cases}
 \widetilde\Psi_{jx}+[R_{jx},\widetilde\Psi_j]
   =\widehat W_j\widetilde\Psi_j,\\
 \widetilde\Psi_{jt}+[R_{jt},\widetilde\Psi_j]
   =\widehat V_j\widetilde\Psi_j.
\end{cases}
\end{equation}
We then define the Jost solutions
$\widetilde\Psi_j(x,t,\lambda)$,
$j=1,2$, of \eqref{ic:commutator-system} through the Volterra integral equations
\begin{equation}\label{ic:volterra-tildePsi}
\begin{aligned}
 \widetilde\Psi_j(x,t,\lambda)=I
 &+\int_{(-1)^j\infty}^{x}
 \eul^{\frac{\ii l_j(\lambda)}2      \int_x^\xi\rho(\tau,t)\dd\tau\,\sigma_3}
 \widehat W_j(\xi,t,\lambda)\widetilde\Psi_j(\xi,t,\lambda)
 \eul^{-\frac{\ii l_j(\lambda)}2      \int_x^\xi\rho(\tau,t)\dd\tau\,\sigma_3}
 \dd\xi.
\end{aligned}
\end{equation}
Since $\widehat V_j(x,t,\lambda)$ as $x\to(-1)^j\infty$, it can be shown that in view of the compatibility of the two equations in~\eqref{ic:transformed-lax}, the solutions of~\eqref{ic:volterra-tildePsi} also satisfy the $t$-equation in~\eqref{ic:commutator-system}. 

An equivalent Volterra representation is obtained by returning to the
balanced gauge. Define
\begin{equation}\label{ic:Psi-b-j-def}
 \Psi_j^{\mathrm b}
 :=
 H_j^{-1}\widetilde\Psi_j\eul^{-R_j},
 \qquad
 \Psi_{0,j}^{\mathrm b}
 :=
 H_j^{-1}\eul^{-R_j}.
\end{equation}
Since
$ \widehat W_j
 =
 \delta_jH_j\sigma_3H_j^{-1}$,
equation \eqref{ic:volterra-tildePsi} is equivalent to
\begin{equation}\label{ic:volterra-Psi-b}
\begin{aligned}
 \Psi_j^{\mathrm b}(x,t,\lambda)
 &=
 \Psi_{0,j}^{\mathrm b}(x,t,\lambda)
 +\int_{(-1)^j\infty}^{x}
 P_j(x,\xi,t,\lambda)
 \delta_j(\xi,t)\sigma_3
 \Psi_j^{\mathrm b}(\xi,t,\lambda)
 \dd\xi,
\end{aligned}
\end{equation}
where
\begin{equation}\label{ic:Pj-def}
 P_j(x,y,t,\lambda)
 :=
 \Psi_{0,j}^{\mathrm b}(x,t,\lambda)
 \bigl(\Psi_{0,j}^{\mathrm b}(y,t,\lambda)\bigr)^{-1}.
\end{equation}
The corresponding Jost solutions of the original Lax pair are recovered
from
\begin{equation}\label{ic:Psi-j-def}
 \Psi_j=G^{-1}\Psi_j^{\mathrm b}.
\end{equation}

\begin{lemma}
\label{lem:entire-background-propagator}
For fixed $x,y,t$, the matrix
$P_j(x,y,t,\lambda)$ is an entire function of $\lambda$.
\end{lemma}

\begin{proof}
Set
\[
 L_\rho(x,y,t):=\int_y^x\rho(\xi,t)\dd\xi,
 \qquad
 \theta_j:=\frac{l_j(\lambda)}2L_\rho(x,y,t).
\]
Since
$ f_j(x,t,\lambda)-f_j(y,t,\lambda)=\ii\theta_j,$
direct multiplication yields
\begin{equation}\label{ic:background-propagator-formula}
 P_j(x,y,t,\lambda)
 =
 \begin{pmatrix}
  \displaystyle
  \cos\theta_j-\frac{\sin\theta_j}{A_jl_j}
  &
  \displaystyle
  \frac{\lambda\sin\theta_j}{l_j}\\[3mm]
  \displaystyle
  -\frac{\lambda\sin\theta_j}{l_j}
  &
  \displaystyle
  \cos\theta_j+\frac{\sin\theta_j}{A_jl_j}
 \end{pmatrix}.
\end{equation}
The expansions
\begin{align*}
 \cos\theta_j
 &=
 \sum_{k=0}^{\infty}
 \frac{(-1)^k}{(2k)!}
 \left(\frac{L_\rho}{2}\right)^{2k}
 \left(\lambda^2-A_j^{-2}\right)^k,\\
 \frac{\sin\theta_j}{l_j}
 &=
 \sum_{k=0}^{\infty}
 \frac{(-1)^k}{(2k+1)!}
 \left(\frac{L_\rho}{2}\right)^{2k+1}
 \left(\lambda^2-A_j^{-2}\right)^k
\end{align*}
show that every entry of $P_j$ is entire in $\lambda$.
\end{proof}

Standard Volterra arguments applied to
\eqref{ic:volterra-tildePsi}, using the decay of the exponential kernels
in the corresponding half-planes, yield the following properties.

\begin{proposition}
\label{prop:inner-cut-column-analyticity}  For each fixed $(x,t)$:
\begin{enumerate}[\textbullet]
 \item $\widetilde\Psi_1^{(1)}$ and
       $\widetilde\Psi_2^{(2)}$ are analytic in $\mathbb C^+$;
 \item $\widetilde\Psi_1^{(2)}$ and
       $\widetilde\Psi_2^{(1)}$ are analytic in $\mathbb C^-$;
 \item every column admits a continuous boundary value on
       $\dot\Sigma_j$ from its half-plane of analyticity;
       $\widetilde\Psi_1^{(1)}$ and
       $\widetilde\Psi_2^{(2)}$ extend to the upper sides of
       $\dot\Gamma_1$ and $\dot\Gamma_2$, respectively, whereas
       $\widetilde\Psi_1^{(2)}$ and
       $\widetilde\Psi_2^{(1)}$ extend to the corresponding lower
       sides;
 \item $\det\widetilde\Psi_j\equiv1$ wherever both columns are defined.

\item 
\begin{subequations}\label{ic:tildePsi-large-lambda}
\begin{align}
 \left(\begin{smallmatrix}
  \widetilde\Psi_1^{(1)}&\widetilde\Psi_2^{(2)}
 \end{smallmatrix}\right)&=I+\ord(\lambda^{-1}),
 &&\lambda\to\infty,\quad\lambda\in\overline{\mathbb C^+},\\
 \left(\begin{smallmatrix}
  \widetilde\Psi_2^{(1)}&\widetilde\Psi_1^{(2)}
 \end{smallmatrix}\right)&=I+\ord(\lambda^{-1}),
 &&\lambda\to\infty,\quad\lambda\in\overline{\mathbb C^-}.
\end{align}
\end{subequations}
\end{enumerate}
\end{proposition}

Since the Jost solutions $\Psi_1$ and $\Psi_2$ satisfy the same Lax pair
\eqref{eq:lax}, they are related by the
scattering relation
\begin{equation}\label{ic:scattering-relation}
 \Psi_1(x,t,\lambda)
 =
 \Psi_2(x,t,\lambda)c(\lambda),
 \qquad \lambda\in\dot\Sigma_1.
\end{equation}
Here the coefficient matrix $c(\lambda)$ is called the scattering matrix. For
the normalized eigenfunctions $\widetilde \Psi_j$, $j=1,2$, the
scattering relation \eqref{ic:scattering-relation} takes the
form
\begin{equation}\label{ic:normalized-scattering-relation}
 H_1^{-1}\widetilde\Psi_1
 =
 H_2^{-1}\widetilde\Psi_2\,\widetilde c(x,t,\lambda),
 \qquad \lambda\in\dot\Sigma_1,
\end{equation}
where
\[
 \widetilde c(x,t,\lambda)
 :=
 \eul^{-R_2}c(\lambda)\eul^{R_1}.
\]

\begin{proposition}\label{prop:scattering-matrix}
The scattering matrix $c(\lambda)$ is independent of $x$ and $t$ and satisfies
\[
 \det c(\lambda)=1.
\]
For $\lambda\in\dot\Sigma_1$, its entries admit the determinant representations
\begin{subequations}\label{ic:scattering-determinants}
\begin{align}
 c_{11}
 &=\det\bigl(\Psi_1^{(1)},\Psi_2^{(2)}\bigr),
 &
 c_{12}
 &=\det\bigl(\Psi_1^{(2)},\Psi_2^{(2)}\bigr),\\
 c_{21}
 &=\det\bigl(\Psi_2^{(1)},\Psi_1^{(1)}\bigr),
 &
 c_{22}
 &=\det\bigl(\Psi_2^{(1)},\Psi_1^{(2)}\bigr).
\end{align}
\end{subequations}

In particular, $c_{11}$ extends analytically to $\mathbb C^+$ and
continuously to the upper side of $\dot\Gamma_1\setminus\{\pm A_2^{-1}\}$, whereas $c_{22}$ extends analytically to $\mathbb C^-$ and continuously
to the lower side of $\dot\Gamma_1\setminus\{\pm A_2^{-1}\}$. Moreover, $c_{21}$ and $c_{12}$
admit continuous extensions to the upper and lower sides of
$\dot\Sigma_0$, respectively. On $\dot\Sigma_0$, the corresponding
column relations are
\begin{subequations}\label{ic:continued-column-scattering}
\begin{align}
 \Psi_{1,+}^{(1)}
 &=
 c_{11,+}\Psi_{2,-}^{(1)}
 +c_{21,+}\Psi_{2,+}^{(2)},
 &&\lambda\in\dot\Sigma_0,\\
 \Psi_{1,-}^{(2)}
 &=
 c_{12,-}\Psi_{2,-}^{(1)}
 +c_{22,-}\Psi_{2,+}^{(2)},
 &&\lambda\in\dot\Sigma_0.
\end{align}
\end{subequations}
Here each column of the Jost solutions is evaluated from its half-plane of analyticity:
$\Psi_2^{(1)}$ from $\mathbb C^-$ and $\Psi_2^{(2)}$ from
$\mathbb C^+$. The subscripts on the scattering coefficients indicate
their continuations to the corresponding sides of $\Gamma_1$.
\end{proposition}

\begin{proof}
First notice that, for $\lambda\in\dot\Sigma_1$,
$ c(\lambda)=\Psi_2^{-1}(x,t,\lambda)\Psi_1(x,t,\lambda)$.
Differentiating this identity and using the Lax pair \eqref{eq:lax}, we conclude that
 $c$ is independent of $x$ and $t$. Since the coefficient
matrices in \eqref{eq:lax} are traceless and the Jost solutions have
unit-determinant normalizations, we have $\det \Psi_j=1$, and thus $\det c=1$.
The determinant representations
\eqref{ic:scattering-determinants} then follow from Cramer's rule.

The analyticity and boundary properties of the scattering coefficients
follow from Proposition~\ref{prop:inner-cut-column-analyticity} and
\eqref{ic:Psi-tildePsi-relation}. On $\dot\Sigma_0$, the columns
$\Psi_2^{(1)}(\lambda_-)$ and $\Psi_2^{(2)}(\lambda_+)$ form a
fundamental matrix. Expanding the corresponding columns of $\Psi_1$
in this basis gives \eqref{ic:continued-column-scattering}.
\end{proof}

\subsection{Symmetries}
\label{sec:ic-symmetries} 
The symmetries of the Jost solutions and the scattering data are
inherited from those of the functions $l_j$, the matrices $H_j$, and
the coefficient matrices $\mathsf U$ and $\mathsf V$ in the Lax pair
\eqref{eq:lax}. For brevity, we suppress the dependence on $x$ and $t$ in what follows.

We begin with the symmetries satisfied by $l_j(\lambda)$.

\begin{lemma}\label{lem:ell-symmetries}
For $j=1,2$, the following relations hold:
\begin{subequations}\label{ic:l-symmetries}
\begin{alignat}{2}
 l_j(-\lambda)
 &=-l_j(\lambda),
 &\qquad&
 \lambda\in\mathbb C\setminus\Gamma_j,
 \label{ic:l-symmetry-a}\\
 \overline{l_j(\overline\lambda)}
 &=l_j(\lambda),
 &&
 \lambda\in\mathbb C\setminus\Gamma_j,
 \label{ic:l-symmetry-c}\\
 l_j(\lambda_-)
 &=-l_j(\lambda_+),
 &&
 \lambda\in\dot\Gamma_j,
 \label{ic:l-symmetry-e}\\
 \overline{l_j(\lambda_\pm)}
 &=-l_j(\lambda_\pm),
 &&
 \lambda\in\dot\Gamma_j,
 \label{ic:l-symmetry-d}\\
 l_j((-\lambda)_\pm)
 &=l_j(\lambda_\pm),
 &&
 \lambda\in\dot\Gamma_j.
 \label{ic:l-symmetry-b}
\end{alignat}
\end{subequations}
\end{lemma}

\begin{proof}
The functions on the two sides of \eqref{ic:l-symmetry-a} and
\eqref{ic:l-symmetry-c} have the same square,
$\lambda^2-A_j^{-2}$. Their behavior at infinity fixes the sign and
proves these two identities.

Crossing the cut changes the sign of the square root, and hence
\eqref{ic:l-symmetry-e} holds. Taking boundary values in
\eqref{ic:l-symmetry-c} gives
\begin{equation}
 l_j(\lambda_+)=\overline{l_j(\lambda_-)},
 \qquad \lambda\in\Gamma_j.
\end{equation}
Using \eqref{ic:l-symmetry-e}, we obtain
\begin{equation}
 \overline{l_j(\lambda_+)}
 =-\overline{l_j(\lambda_-)}
 =-l_j(\lambda_+),  \qquad \lambda\in\Gamma_j
\end{equation}
which proves \eqref{ic:l-symmetry-d}.

Finally, taking the upper boundary value in
\eqref{ic:l-symmetry-a} gives
\begin{equation}
 l_j((-\lambda)_-)=-l_j(\lambda_+),
\end{equation}
since $-\lambda$ is approached from the lower half-plane when
$\lambda$ is approached from the upper half-plane. By
\eqref{ic:l-symmetry-e},
\begin{equation}
 l_j((-\lambda)_-)=-l_j((-\lambda)_+),
\end{equation}
and therefore
\begin{equation}
 l_j(\lambda_+)=l_j((-\lambda)_+),
\end{equation}
which is \eqref{ic:l-symmetry-b}.
\end{proof}

Now, recall that the
original and normalized Jost solutions are related by
\begin{equation}\label{ic:Psi-tildePsi-relation}
 H_j^{-1}(\lambda)\widetilde\Psi_j(\lambda)
 =
 G\Psi_j(\lambda)\eul^{R_j(\lambda)}.
\end{equation}
Since $G$ is real and diagonal, it satisfies
\begin{equation}\label{ic:G-symmetry-properties}
 \overline G=G,
 \qquad
 G=\sigma_3G\sigma_3.
\end{equation}
The corresponding identities for $H_j$ and $R_j$ are collected in the
following lemma.

\begin{lemma}\label{lem:ic-HR-symmetries}
For $j=1,2$ and $\lambda\in\dot\Sigma_j$,
\begin{subequations}\label{ic:HR-outer-symmetries}
\begin{align}
 H_j^{-1}(\lambda)
 &=\sigma_3H_j^{-1}(-\lambda)\sigma_2,
 \label{ic:H-outer-minus}\\
 H_j^{-1}(\lambda)
 &=-\ii\overline{H_j^{-1}(\lambda)}\sigma_1,
 \label{ic:H-outer-conjugation}\\
 R_j(-\lambda)&=-R_j(\lambda),
 \label{ic:R-outer-minus}\\
 R_j(\lambda)&=-\overline{R_j(\lambda)}.
 \label{ic:R-outer-conjugation}
\end{align}
\end{subequations}
For $\lambda\in\dot\Gamma_j\setminus\{0\}$,
\begin{subequations}\label{ic:HR-cut-symmetries}
\begin{align}
 R_j((-\lambda)_\pm)&=R_j(\lambda_\pm),
 \label{ic:R-cut-minus}\\
 \overline{R_j(\lambda_\pm)}&=R_j(\lambda_\pm),
 \label{ic:R-cut-real}\\
 R_j(\lambda_-)&=-R_j(\lambda_+),
 \label{ic:R-cut-jump}\\
 \bigl(H_j^{-1}(\lambda_\pm)\bigr)^{(1)}
 &=\mp\sigma_3
   \bigl(H_j^{-1}((-\lambda)_\pm)\bigr)^{(1)},
 \label{ic:H-cut-minus-first}\\
 \bigl(H_j^{-1}(\lambda_\pm)\bigr)^{(2)}
 &=\pm\sigma_3
   \bigl(H_j^{-1}((-\lambda)_\pm)\bigr)^{(2)},
 \label{ic:H-cut-minus-second}\\
 H_j^{-1}(\lambda_+)
 &=-\overline{H_j^{-1}(\lambda_+)},
 \label{ic:H-cut-upper-conjugation}\\
 H_j^{-1}(\lambda_-)
 &=\overline{H_j^{-1}(\lambda_-)}.
 \label{ic:H-cut-lower-conjugation}
\end{align}
\end{subequations}
\end{lemma}

\begin{proof} First, notice that the definition of $f_j$, together with
\eqref{ic:l-symmetry-a} and \eqref{ic:l-symmetry-c}, yields
\[
 R_j(-\lambda)=-R_j(\lambda),
 \qquad
 R_j(\lambda)=-\overline{R_j(\overline\lambda)}.
\]
Restricting these identities to $\dot\Sigma_j$ gives
\eqref{ic:R-outer-minus} and
\eqref{ic:R-outer-conjugation}.
Similarly,
for $\lambda\in\dot\Gamma_j\setminus\{0\}$, substitution of
\eqref{ic:l-symmetry-b}, \eqref{ic:l-symmetry-d}, and \eqref{ic:l-symmetry-e}  into the definition of $f_j$ yields
\eqref{ic:R-cut-minus},
\eqref{ic:R-cut-real}, and \eqref{ic:R-cut-jump}, respectively.

To establish the identities for $H_j^{-1}$, set
\[
 h_j(\lambda)
 :=
 \frac{\lambda A_j}{1-\ii A_jl_j(\lambda)},
 \qquad
 \gamma_j(\lambda)
 :=
 \sqrt{\frac12}
 \sqrt{\frac1{\ii A_jl_j(\lambda)}-1},
\]
so that
\[
 H_j^{-1}(\lambda)
 =
 \gamma_j(\lambda)
 \begin{pmatrix}
  h_j(\lambda)&1\\
  1&h_j(\lambda)
 \end{pmatrix}.
\]
Recall that the square root entering the definition of $\gamma_j$ is
defined on $\mathbb C\setminus[0,\infty)$ and normalized by
$\sqrt{-1}=\ii$.

For $\lambda\in\dot\Sigma_j$, substituting
\eqref{ic:l-symmetry-a} and \eqref{ic:l-symmetry-c} into this
representation gives \eqref{ic:H-outer-minus} and
\eqref{ic:H-outer-conjugation}, respectively.

For $\lambda\in\dot\Gamma_j\setminus\{0\}$,
\eqref{ic:l-symmetry-b} implies
\[
 h_j((-\lambda)_\pm)=-h_j(\lambda_\pm),
\]
whereas the boundary values of $\gamma_j$ satisfy
\[
 \gamma_j((-\lambda)_+)=\gamma_j(\lambda_+),
 \qquad
 \gamma_j((-\lambda)_-)=-\gamma_j(\lambda_-).
\]
Indeed, on the upper side of $\Gamma_j$, the argument of the square
root in $\gamma_j$ lies on the negative real axis and is unchanged
under $\lambda\mapsto-\lambda$, which proves the first equality. On
the lower side, this argument lies on the branch cut $[0,\infty)$.
Although its limiting value is unchanged, the transformation
$\lambda\mapsto-\lambda$ interchanges the sides from which this cut is
approached. The corresponding boundary values of the square root
differ by a sign, which proves the second equality.

Taking the two columns separately proves
\eqref{ic:H-cut-minus-first} and
\eqref{ic:H-cut-minus-second}. Moreover,
\eqref{ic:l-symmetry-d} shows that $h_j$ is real on both sides of the
cut. The factor $\gamma_j$ is purely imaginary on the upper side and
real on the lower side. Consequently,
\eqref{ic:H-cut-upper-conjugation} and
\eqref{ic:H-cut-lower-conjugation} follow.
\end{proof}

We will also use the following boundary-uniqueness principle, which
follows from the Schwarz reflection principle and the identity theorem;
see, for example,
\cite{Ahlfors79}.

\begin{lemma}\label{lem:ic-boundary-uniqueness}
Let $F$ be a scalar- or matrix-valued function analytic in
$\mathbb C^\pm$ and continuous up to a nonempty open interval
$I\subset\mathbb R$. If $ F(\lambda)=0$ for $\lambda\in I$,
then $F\equiv0$ in $\mathbb C^\pm$.
\end{lemma}

\begin{proof}
It suffices to argue componentwise. Since the boundary values of $F$
vanish on $I$, Schwarz reflection extends $F$ analytically across
$I$. The extension vanishes on $I$, and hence the identity theorem
implies that $F\equiv0$ in $\mathbb C^\pm$.
\end{proof}

Now we are at the position to formulate the symmetry properties of the Jost solutions and the associated scattering coefficients.

\subsubsection*{First symmetry: $\lambda\longleftrightarrow-\lambda$}

\begin{proposition}\label{prop:ic-Psi-minus}
For $j=1,2$, the  Jost solutions $\Psi_j$ satisfy
\begin{equation}\label{ic:Psi-minus}
 \Psi_j(\lambda)
 =
 \sigma_3\Psi_j(-\lambda)\sigma_2,
 \qquad
 \lambda\in\dot\Sigma_j.
\end{equation}
\end{proposition}

\begin{proof}
We first observe that the coefficient matrices of the Lax pair
\eqref{eq:lax} possess the symmetries
\begin{equation}\label{ic:UV-minus-symmetry}
 \mathsf U(-\lambda)
 =
 \sigma_3\mathsf U(\lambda)\sigma_3,
 \qquad
 \mathsf V(-\lambda)
 =
 \sigma_3\mathsf V(\lambda)\sigma_3.
\end{equation}
Consequently, whenever $\Psi_j(\lambda)$ is a solution of
\eqref{eq:lax}, the function $\sigma_3\Psi_j(-\lambda)$
is also a solution. Moreover, since $G\to I$ and $\widetilde\Psi_j\to I$
as $x\to(-1)^j\infty$, the Jost solutions have the asymptotic behavior
\[
 \Psi_j(\lambda)
 \sim
 H_j^{-1}(\lambda)\eul^{-R_j(\lambda)},\qquad x\to(-1)^j\infty.
\]
Furthermore, the identities
\eqref{ic:H-outer-minus}, \eqref{ic:R-outer-minus},
together with
$\eul^{R_j(\lambda)}=\sigma_2\eul^{-R_j(\lambda)}\sigma_2$,
yield
\begin{align*}
 \sigma_3H_j^{-1}(-\lambda)
 \eul^{-R_j(-\lambda)}\sigma_2
 =
 H_j^{-1}(\lambda)\eul^{-R_j(\lambda)}.
\end{align*}
Thus, $\Psi_j(\lambda)$ and
$\sigma_3\Psi_j(-\lambda)\sigma_2$ satisfy the same
Lax pair and have the same asymptotic normalization.
The uniqueness of the solution to the corresponding Volterra
integral equation therefore yields \eqref{ic:Psi-minus}.
\end{proof}

\begin{corollary}\label{cor:ic-minus-symmetries}
The following relations hold.
\begin{enumerate}
 \item For $\lambda\in\dot\Sigma_1$,
 \begin{equation}\label{ic:c-minus-matrix}
  c(\lambda)
  =
  \sigma_2c(-\lambda)\sigma_2.
 \end{equation}
 Equivalently,
 \begin{equation}\label{ic:c-minus-entries}
  c_{11}(\lambda)=c_{22}(-\lambda),
  \qquad
  c_{21}(\lambda)=-c_{12}(-\lambda).
 \end{equation}

 \item For $\lambda\in\mathbb C^+$,
 \begin{equation}\label{ic:c11-c22-minus-analytic}
  c_{11}(\lambda)
  =
  c_{22}(-\lambda).
 \end{equation}
 Moreover,
 \begin{equation}\label{ic:c11-c22-minus-cut}
  c_{11}(\lambda_+)
  =
  c_{22}((-\lambda)_-),
  \qquad
  \lambda\in
  \dot\Gamma_1\setminus\{\pm A_2^{-1}\}.
 \end{equation}

 \item For $\lambda\in\mathbb C^+$,
 \begin{subequations}\label{ic:Psi-minus-columns}
 \begin{align}
  \Psi_1^{(1)}(\lambda)
  &=
  \ii\sigma_3\Psi_1^{(2)}(-\lambda),\\
  \Psi_2^{(2)}(\lambda)
  &=
  -\ii\sigma_3\Psi_2^{(1)}(-\lambda).
 \end{align}
 \end{subequations}

 \item For $j=1,2$ and $\lambda\in\dot\Sigma_j$,
 \begin{equation}\label{ic:HtildePsi-minus}
  \bigl(H_j^{-1}\widetilde\Psi_j\bigr)(\lambda)
  =
  \sigma_3
  \bigl(H_j^{-1}\widetilde\Psi_j\bigr)(-\lambda)
  \sigma_2.
 \end{equation}
 \item  For $\lambda\in\mathbb C^+$,
 \begin{subequations}\label{ic:HtildePsi-minus-columns}
 \begin{align}
  \bigl(H_1^{-1}\widetilde\Psi_1\bigr)^{(1)}(\lambda)
  &=
  \ii\sigma_3
  \bigl(H_1^{-1}\widetilde\Psi_1\bigr)^{(2)}(-\lambda),\\
  \bigl(H_2^{-1}\widetilde\Psi_2\bigr)^{(2)}(\lambda)
  &=
  -\ii\sigma_3
  \bigl(H_2^{-1}\widetilde\Psi_2\bigr)^{(1)}(-\lambda).
 \end{align}
 \end{subequations}

\end{enumerate}
\end{corollary}

\begin{proof}
Substituting \eqref{ic:Psi-minus} into
\eqref{ic:scattering-relation} implies \eqref{ic:c-minus-matrix},
which is equivalent to  \eqref{ic:c-minus-entries}.

The functions $c_{11}(\lambda)$ and $c_{22}(-\lambda)$ are analytic
in $\mathbb C^+$ and, by \eqref{ic:c-minus-entries}, have the same
boundary values on $\dot\Sigma_1$. Applying
Lemma~\ref{lem:ic-boundary-uniqueness} to their difference proves
\eqref{ic:c11-c22-minus-analytic}. Passing to the upper boundary of
the cut then gives \eqref{ic:c11-c22-minus-cut}, since
$-\lambda_+=(-\lambda)_-$.

Taking the first column of \eqref{ic:Psi-minus} for $j=1$ and the
second column for $j=2$ gives \eqref{ic:Psi-minus-columns} on
$\dot\Sigma_1$ and $\dot\Sigma_2$, respectively. Since both sides of
these identities are analytic in $\mathbb C^+$,
Lemma~\ref{lem:ic-boundary-uniqueness}, applied to their differences,
yields \eqref{ic:Psi-minus-columns} for all $\lambda\in\mathbb C^+$.

Finally, using \eqref{ic:Psi-tildePsi-relation},
\eqref{ic:G-symmetry-properties}, \eqref{ic:Psi-minus}, and
\eqref{ic:R-outer-minus}, together with $\eul^{R_j(\lambda)}=\sigma_2\eul^{-R_j(\lambda)}\sigma_2$,
we obtain \eqref{ic:HtildePsi-minus}. Applying
Lemma~\ref{lem:ic-boundary-uniqueness} to the relevant columns then
yields \eqref{ic:HtildePsi-minus-columns}.
\end{proof}

\subsubsection*{Second symmetry:
$\lambda\longleftrightarrow\overline\lambda$}

\begin{proposition}\label{prop:ic-Psi-conjugation}
For $j=1,2$, the  Jost solutions $\Psi_j$ satisfy
\begin{equation}\label{ic:Psi-conjugation}
 \Psi_j(\lambda)
 =
 -\ii\overline{\Psi_j(\overline\lambda)}\sigma_1,
 \qquad
 \lambda\in\dot\Sigma_j.
\end{equation}
\end{proposition}

\begin{proof}
Since the solution $(u,v)$ is real-valued, the coefficient matrices
of the Lax pair \eqref{eq:lax} satisfy
\begin{equation}\label{ic:UV-conjugation-symmetry}
 \overline{\mathsf U(\overline\lambda)}
 =
 \mathsf U(\lambda),
 \qquad
 \overline{\mathsf V(\overline\lambda)}
 =
 \mathsf V(\lambda).
\end{equation}
Therefore, whenever $\Psi_j(\lambda)$ solves
\eqref{eq:lax}, the function $\overline{\Psi_j(\overline\lambda)}$
is also a solution.

To compare their asymptotic behavior, recall that
\[
\Psi_j(\lambda)
\sim
H_j^{-1}(\lambda)\eul^{-R_j(\lambda)},
\qquad x\to(-1)^j\infty.
\]
For $\lambda\in\dot\Sigma_j$, the identities
\eqref{ic:H-outer-conjugation} and
\eqref{ic:R-outer-conjugation}, together with
$\eul^{R_j(\lambda)}
=\sigma_1\eul^{-R_j(\lambda)}\sigma_1$,
yield
\begin{align*}
 -\ii\overline{
 H_j^{-1}(\bar\lambda)\eul^{-R_j(\bar\lambda)}}\sigma_1 =
 H_j^{-1}(\lambda)\eul^{-R_j(\lambda)}.
\end{align*}
Consequently, $\Psi_j(\lambda)$ and
$-\ii\overline{\Psi_j(\overline\lambda)}\sigma_1$
satisfy the same Lax pair with the same asymptotic normalization.
By uniqueness of the solution to the corresponding Volterra
integral equation, the symmetry relation
\eqref{ic:Psi-conjugation} follows.
\end{proof}

\begin{corollary}\label{cor:ic-conjugation-symmetries}
The following relations hold.
\begin{enumerate}
 \item For $\lambda\in\dot\Sigma_1$,
 \begin{equation}\label{ic:c-conjugation-matrix}
  c(\lambda)
  =
  \sigma_1\overline{c(\lambda)}\sigma_1.
 \end{equation}
 Equivalently,
 \begin{equation}\label{ic:c-conjugation-entries}
  c_{11}(\lambda)=\overline{c_{22}(\lambda)},
  \qquad
  c_{21}(\lambda)=\overline{c_{12}(\lambda)}.
 \end{equation}

 \item For $\lambda\in\dot\Sigma_1$,
\begin{equation}\label{ic:c-pseudo-unitarity}
 |c_{11}(\lambda)|^2-|c_{21}(\lambda)|^2=1.
\end{equation}

 \item For $\lambda\in\mathbb C^+$,
 \begin{equation}\label{ic:c11-c22-conjugation-analytic}
  c_{11}(\lambda)
  =
  \overline{c_{22}(\overline\lambda)}.
 \end{equation}
 Moreover,
 \begin{equation}\label{ic:c11-c22-conjugation-cut}
  c_{11}(\lambda_+)
  =
  \overline{c_{22}(\lambda_-)},
  \qquad
  \lambda\in
  \dot\Gamma_1\setminus\{\pm A_2^{-1}\}.
 \end{equation}

 \item For $\lambda\in\mathbb C^+$,
 \begin{subequations}\label{ic:Psi-conjugation-columns}
 \begin{align}
  \Psi_1^{(1)}(\lambda)
  &=
  -\ii\overline{\Psi_1^{(2)}(\overline\lambda)},\\
  \Psi_2^{(2)}(\lambda)
  &=
  -\ii\overline{\Psi_2^{(1)}(\overline\lambda)}.
 \end{align}
 \end{subequations}

 \item  For $j=1,2$ and $\lambda\in\dot\Sigma_j$,
 \begin{equation}\label{ic:HtildePsi-conjugation}
  \bigl(H_j^{-1}\widetilde\Psi_j\bigr)(\lambda)
  =
  -\ii\overline{
  \bigl(H_j^{-1}\widetilde\Psi_j\bigr)(\lambda)}\sigma_1.
 \end{equation}
 \item For $\lambda\in\mathbb C^+$,
 \begin{subequations}\label{ic:HtildePsi-conjugation-columns}
 \begin{align}
  \bigl(H_1^{-1}\widetilde\Psi_1\bigr)^{(1)}(\lambda)
  &=
  -\ii\overline{
  \bigl(H_1^{-1}\widetilde\Psi_1\bigr)^{(2)}
  (\overline\lambda)},\\
  \bigl(H_2^{-1}\widetilde\Psi_2\bigr)^{(2)}(\lambda)
  &=
  -\ii\overline{
  \bigl(H_2^{-1}\widetilde\Psi_2\bigr)^{(1)}
  (\overline\lambda)}.
 \end{align}
 \end{subequations}

\end{enumerate}
\end{corollary}

\begin{proof}
Substituting \eqref{ic:Psi-conjugation} into the scattering relation
\eqref{ic:scattering-relation}, we obtain
\eqref{ic:c-conjugation-matrix}, or equivalently,
\eqref{ic:c-conjugation-entries}. The latter, together with
$\det c=1$, yields \eqref{ic:c-pseudo-unitarity}.

Next, both $c_{11}(\lambda)$ and
$\overline{c_{22}(\overline\lambda)}$ are analytic in
$\mathbb C^+$ and have the same boundary values on
$\dot\Sigma_1$ by \eqref{ic:c-conjugation-entries}.
Applying Lemma~\ref{lem:ic-boundary-uniqueness} to their
difference yields \eqref{ic:c11-c22-conjugation-analytic}.
Taking upper boundary values on the cut then gives
\eqref{ic:c11-c22-conjugation-cut}.

Taking the first column of \eqref{ic:Psi-conjugation} for $j=1$
and the second column for $j=2$ yields
\eqref{ic:Psi-conjugation-columns} on $\dot\Sigma_1$
and $\dot\Sigma_2$, respectively. Since both sides are analytic
in $\mathbb C^+$ and have the same boundary values on the
respective intervals, Lemma~\ref{lem:ic-boundary-uniqueness}
extends these identities to the entire upper half-plane.

Finally, combining \eqref{ic:Psi-tildePsi-relation},
\eqref{ic:G-symmetry-properties}, and the identity
$R_j(\lambda)=-\overline{R_j(\overline\lambda)}$,
we obtain \eqref{ic:HtildePsi-conjugation}.
The corresponding columnwise identities
\eqref{ic:HtildePsi-conjugation-columns} then follow
from Lemma~\ref{lem:ic-boundary-uniqueness}.
\end{proof}

\subsubsection*{Third symmetry:
$\lambda_\pm\longleftrightarrow(-\lambda)_\pm$}

\begin{proposition}\label{prop:ic-Psi-minus-upper-cut}
For the upper boundary values, we have
\begin{subequations}\label{ic:Psi-minus-upper-cut}
\begin{align}
 \Psi_1^{(1)}(\lambda_+)
 &=
 -\sigma_3\Psi_1^{(1)}((-\lambda)_+),
 &&\lambda\in\dot\Gamma_1\setminus\{0\},\\
 \Psi_2^{(2)}(\lambda_+)
 &=
 \sigma_3\Psi_2^{(2)}((-\lambda)_+),
 &&\lambda\in\dot\Gamma_2\setminus\{0\}.
\end{align}
\end{subequations}
\end{proposition}

\begin{proof}
The coefficient matrices in \eqref{eq:lax} are single-valued functions
of $\lambda$. Therefore, \eqref{ic:UV-minus-symmetry} remains valid
when $\lambda$ and $-\lambda$ are regarded as points on the upper side
of $\Gamma_j$. Consequently,
$ -\sigma_3\Psi_1^{(1)}((-\lambda)_+)$
and
 $\sigma_3\Psi_2^{(2)}((-\lambda)_+)$
satisfy the same Lax equations \eqref{eq:lax} as
$\Psi_1^{(1)}(\lambda_+)$ and
$\Psi_2^{(2)}(\lambda_+)$, respectively. By
\eqref{ic:R-cut-minus}, \eqref{ic:H-cut-minus-first}, and
\eqref{ic:H-cut-minus-second}, each pair has the same 
normalization. The result, therefore, follows from the uniqueness of the solution to the corresponding Volterra
integral equation.
\end{proof}

\begin{corollary}\label{cor:ic-c11-minus-upper-cut}
For $\lambda\in\dot\Gamma_2$,
\begin{equation}\label{ic:c11-minus-upper-cut}
 c_{11}(\lambda_+)
 =
 c_{11}((-\lambda)_+).
\end{equation}
\end{corollary}

\begin{proof}
Substituting \eqref{ic:Psi-minus-upper-cut} into the determinant
representation of $c_{11}$ in \eqref{ic:scattering-determinants}
proves the result. Indeed, the minus sign in the first column is
canceled by $\det\sigma_3=-1$.
\end{proof}

\begin{proposition}\label{prop:ic-Psi-minus-lower-cut}
For the lower boundary values,
\begin{subequations}\label{ic:Psi-minus-lower-cut}
\begin{align}
 \Psi_2^{(1)}(\lambda_-)
 &=
 \sigma_3\Psi_2^{(1)}((-\lambda)_-),
 &&\lambda\in\dot\Gamma_2\setminus\{0\},\\
 \Psi_1^{(2)}(\lambda_-)
 &=
 -\sigma_3\Psi_1^{(2)}((-\lambda)_-),
 &&\lambda\in\dot\Gamma_1\setminus\{0\}.
\end{align}
\end{subequations}
\end{proposition}

\begin{proof}
The proof is analogous to that of
Proposition~\ref{prop:ic-Psi-minus-upper-cut}. Since the coefficient matrices in
\eqref{eq:lax} are single-valued functions of $\lambda$,
\eqref{ic:UV-minus-symmetry} remains valid when $\lambda$ and
$-\lambda$ are regarded as points on the lower side of $\Gamma_j$.
Consequently,
$\sigma_3\Psi_2^{(1)}((-\lambda)_-)$ and
$-\sigma_3\Psi_1^{(2)}((-\lambda)_-)$ satisfy the same Lax equations
\eqref{eq:lax} as $\Psi_2^{(1)}(\lambda_-)$ and
$\Psi_1^{(2)}(\lambda_-)$, respectively. By
\eqref{ic:R-cut-minus}, \eqref{ic:H-cut-minus-first}, and
\eqref{ic:H-cut-minus-second}, each pair has the same 
normalization. The result therefore follows from the uniqueness of
the solution to the corresponding Volterra integral equation.
\end{proof}

\begin{corollary}\label{cor:ic-c22-minus-lower-cut}
For $\lambda\in\dot\Gamma_2$,
\begin{equation}\label{ic:c22-minus-lower-cut}
 c_{22}(\lambda_-)
 =
 c_{22}((-\lambda)_-).
\end{equation}
\end{corollary}

\begin{proof}
Similarly to Corollary \ref{cor:ic-c11-minus-upper-cut}, substituting \eqref{ic:Psi-minus-lower-cut} into the determinant
representation of $c_{22}$ in \eqref{ic:scattering-determinants}
proves the result. 
\end{proof}

\begin{remark}
   The corresponding symmetries for the columns of
$H_j^{-1}\widetilde\Psi_j$ are
\begin{subequations}\label{ic:HtildePsi-minus-cut}
\begin{align}
 \bigl(H_1^{-1}\widetilde\Psi_1\bigr)^{(1)}(\lambda_+)
 &=
 -\sigma_3
 \bigl(H_1^{-1}\widetilde\Psi_1\bigr)^{(1)}
 ((-\lambda)_+),
 &&\lambda\in\dot\Gamma_1\\
 \bigl(H_2^{-1}\widetilde\Psi_2\bigr)^{(2)}(\lambda_+)
 &=
 \sigma_3
 \bigl(H_2^{-1}\widetilde\Psi_2\bigr)^{(2)}
 ((-\lambda)_+),
 &&\lambda\in\dot\Gamma_2\\
 \bigl(H_2^{-1}\widetilde\Psi_2\bigr)^{(1)}(\lambda_-)
 &=
 \sigma_3
 \bigl(H_2^{-1}\widetilde\Psi_2\bigr)^{(1)}
 ((-\lambda)_-),
 &&\lambda\in\dot\Gamma_2\\
 \bigl(H_1^{-1}\widetilde\Psi_1\bigr)^{(2)}(\lambda_-)
 &=
 -\sigma_3
 \bigl(H_1^{-1}\widetilde\Psi_1\bigr)^{(2)}
 ((-\lambda)_-),
 &&\lambda\in\dot\Gamma_1.
\end{align}
\end{subequations}
These relations follow from
\eqref{ic:Psi-tildePsi-relation},
\eqref{ic:Psi-minus-upper-cut},
\eqref{ic:Psi-minus-lower-cut},
\eqref{ic:G-symmetry-properties}, and
\eqref{ic:R-cut-minus}. Unlike $\Psi_j$, the columns
of $H_j^{-1}\widetilde\Psi_j$ admit finite boundary values at
$\lambda=0$ by Remark~\ref{rem:ic-H-at-zero} and
Proposition~\ref{prop:inner-cut-column-analyticity}. Thus, the
relations extend to $\lambda=0$ by continuity.
\end{remark}

\subsubsection*{Fourth symmetry: $\lambda_+ \longleftrightarrow \lambda_+$}

\begin{proposition}\label{prop:ic-Psi-conjugation-upper-cut}
For the upper boundary values, we have 
\begin{equation}\label{ic:Psi-conjugation-upper-cut}
 \Psi_j^{(j)}(\lambda_+)
 =
 -\overline{\Psi_j^{(j)}(\lambda_+)},
 \qquad
 \lambda\in\dot\Gamma_j\setminus\{0\},
\end{equation}
for $j=1,2$.
\end{proposition}

\begin{proof}
The coefficient
matrices $\mathsf U(\lambda)$ and $\mathsf V(\lambda)$ in
\eqref{eq:lax} are single-valued functions of $\lambda$ and have real
entries for real $\lambda$. Therefore,
$-\overline{\Psi_j^{(j)}(\lambda_+)}$ satisfies the same Lax pair
\eqref{eq:lax} as $\Psi_j^{(j)}(\lambda_+)$.
Moreover,
\eqref{ic:H-cut-upper-conjugation} and \eqref{ic:R-cut-real} imply
\[
 -\overline{
 \bigl(H_j^{-1}(\lambda_+)\eul^{-R_j(\lambda_+)}\bigr)^{(j)}
 }
 =
 \bigl(H_j^{-1}(\lambda_+)\eul^{-R_j(\lambda_+)}\bigr)^{(j)}.
\]
Thus the two solutions have the same normalization, and uniqueness of the solution
to the corresponding Volterra integral equation proves the result.
\end{proof}

\begin{corollary}\label{cor:ic-upper-cut-conjugation}
The following relations hold.
\begin{enumerate}
 \item For $\lambda\in\dot\Gamma_2$,
 \begin{equation}\label{ic:c11-real-upper-cut}
  c_{11}(\lambda_+)
  =
  \overline{c_{11}(\lambda_+)}.
 \end{equation}

 \item For $j=1,2$,
 \begin{equation}\label{ic:HtildePsi-conjugation-upper-cut}
  \bigl(H_j^{-1}\widetilde\Psi_j\bigr)^{(j)}(\lambda_+)
  =
  -\overline{
  \bigl(H_j^{-1}\widetilde\Psi_j\bigr)^{(j)}(\lambda_+)},
  \qquad
  \lambda\in\dot\Gamma_j.
 \end{equation}
\end{enumerate}
\end{corollary}

\begin{proof}
Combining the determinant representation of $c_{11}$ in
\eqref{ic:scattering-determinants} with
\eqref{ic:Psi-conjugation-upper-cut} yields
\eqref{ic:c11-real-upper-cut}. Furthermore,
\eqref{ic:Psi-tildePsi-relation},
\eqref{ic:G-symmetry-properties},
\eqref{ic:R-cut-real}, and
\eqref{ic:Psi-conjugation-upper-cut} imply
\eqref{ic:HtildePsi-conjugation-upper-cut}.
\end{proof}

\begin{proposition}\label{prop:ic-Psi-conjugation-lower-cut}
For the lower boundary values,
\begin{subequations}\label{ic:Psi-conjugation-lower-cut}
\begin{align}
 \Psi_1^{(2)}(\lambda_-)
 &=
 \overline{\Psi_1^{(2)}(\lambda_-)},
 &&\lambda\in\dot\Gamma_1\setminus\{0\},\\
 \Psi_2^{(1)}(\lambda_-)
 &=
 \overline{\Psi_2^{(1)}(\lambda_-)},
 &&\lambda\in\dot\Gamma_2\setminus\{0\}.
\end{align}
\end{subequations}
\end{proposition}

\begin{proof}
The proof is analogous to that of
Proposition~\ref{prop:ic-Psi-conjugation-upper-cut}. In this case,
\eqref{ic:H-cut-lower-conjugation} and \eqref{ic:R-cut-real} show that
the lower-side normalization is real, and the result follows from
uniqueness of the solution to the corresponding Volterra integral equation.
\end{proof}

\begin{corollary}\label{cor:ic-lower-cut-conjugation}
The following relations hold.
\begin{enumerate}
 \item For $\lambda\in\dot\Gamma_2$,
 \begin{equation}\label{ic:c22-real-lower-cut}
  c_{22}(\lambda_-)
  =
  \overline{c_{22}(\lambda_-)}.
 \end{equation}

 \item For $\lambda\in\Gamma_j\setminus\{0\}$
 \begin{subequations}\label{ic:Psi-opposite-side-cut}
 \begin{align}
  \Psi_1^{(1)}(\lambda_+)
  &=
  -\ii\Psi_1^{(2)}(\lambda_-),
  &&\lambda\in\dot\Gamma_1\setminus\{0\},\\
  \Psi_2^{(2)}(\lambda_+)
  &=
  -\ii\Psi_2^{(1)}(\lambda_-),
  &&\lambda\in\dot\Gamma_2\setminus\{0\}.
 \end{align}
 \end{subequations}

 \item The scattering coefficients satisfy
 \begin{subequations}\label{ic:c-cut-relations}
 \begin{align}
  c_{11}(\lambda_+)
  &=
  c_{22}(\lambda_-),
  &&\lambda\in\dot\Gamma_2,
  \label{ic:c11-c22-cut}\\
  c_{11}(\lambda_+)
  &=
  -\ii c_{12}(\lambda_-),
  &&\lambda\in\dot\Sigma_0,
  \label{ic:c11-c12-band}\\
  c_{22}(\lambda_-)
  &=
  \ii c_{21}(\lambda_+),
  &&\lambda\in\dot\Sigma_0.
  \label{ic:c22-c21-band}
 \end{align}
 \end{subequations}

 \item  For $\lambda\in\Gamma_j$
 \begin{subequations}\label{ic:HtildePsi-conjugation-lower-cut}
 \begin{align}
  \bigl(H_1^{-1}\widetilde\Psi_1\bigr)^{(2)}(\lambda_-)
  &=
  \overline{
  \bigl(H_1^{-1}\widetilde\Psi_1\bigr)^{(2)}(\lambda_-)},
  &&\lambda\in\dot\Gamma_1,\\
  \bigl(H_2^{-1}\widetilde\Psi_2\bigr)^{(1)}(\lambda_-)
  &=
  \overline{
  \bigl(H_2^{-1}\widetilde\Psi_2\bigr)^{(1)}(\lambda_-)},
  &&\lambda\in\dot\Gamma_2.
 \end{align}
 \end{subequations}

 \item For $\lambda\in\Gamma_j$
 \begin{subequations}\label{ic:HtildePsi-opposite-side-cut}
 \begin{align}
  \bigl(H_1^{-1}\widetilde\Psi_1\bigr)^{(1)}(\lambda_+)
  &=
  -\ii
  \bigl(H_1^{-1}\widetilde\Psi_1\bigr)^{(2)}(\lambda_-),
  &&\lambda\in\dot\Gamma_1,\\
  \bigl(H_2^{-1}\widetilde\Psi_2\bigr)^{(2)}(\lambda_+)
  &=
  -\ii
  \bigl(H_2^{-1}\widetilde\Psi_2\bigr)^{(1)}(\lambda_-),
  &&\lambda\in\dot\Gamma_2.
 \end{align}
 \end{subequations}
\end{enumerate}
\end{corollary}

\begin{proof}
Relation \eqref{ic:c22-real-lower-cut} follows by combining the determinant representation
of $c_{22}$ in \eqref{ic:scattering-determinants} with
\eqref{ic:Psi-conjugation-lower-cut}. 

Next, passing to the boundary values in
\eqref{ic:Psi-conjugation-columns} and applying
\eqref{ic:Psi-conjugation-lower-cut} yields
\eqref{ic:Psi-opposite-side-cut}.

The determinant representations of $c_{11}$ and $c_{22}$ given in \eqref{ic:scattering-determinants}, together with
\eqref{ic:Psi-opposite-side-cut}, yield
\begin{align*}
 c_{11}(\lambda_+)
 &=
 \det\bigl(
 -\ii\Psi_1^{(2)}(\lambda_-),
 -\ii\Psi_2^{(1)}(\lambda_-)
 \bigr)=
 \det\bigl(
 \Psi_2^{(1)}(\lambda_-),
 \Psi_1^{(2)}(\lambda_-)
 \bigr)
 =
 c_{22}(\lambda_-),
\end{align*}
which proves \eqref{ic:c11-c22-cut}.

For $\lambda\in\dot\Sigma_0\subset\dot\Sigma_2$, the columns
$\Psi_2^{(1)}$ and $\Psi_2^{(2)}$ are evaluated at their natural
boundary values $\Psi_2^{(1)}(\lambda_-)$ and
$\Psi_2^{(2)}(\lambda_+)$, respectively. Accordingly, the continued
determinant representations of the off-diagonal scattering
coefficients take the form
\[
 c_{12}(\lambda_-)
 =
 \det\bigl(
 \Psi_1^{(2)}(\lambda_-),
 \Psi_2^{(2)}(\lambda_+)
 \bigr),
 \qquad
 c_{21}(\lambda_+)
 =
 \det\bigl(
 \Psi_2^{(1)}(\lambda_-),
 \Psi_1^{(1)}(\lambda_+)
 \bigr).
\]
Using \eqref{ic:Psi-opposite-side-cut} in the determinant
representations \eqref{ic:scattering-determinants}, we obtain
\begin{align*}
 c_{11}(\lambda_+)
 &=
 \det\bigl(
 \Psi_1^{(1)}(\lambda_+),
 \Psi_2^{(2)}(\lambda_+)
 \bigr)=
 \det\bigl(
 -\ii\Psi_1^{(2)}(\lambda_-),
 \Psi_2^{(2)}(\lambda_+)
 \bigr)
 =
 -\ii c_{12}(\lambda_-),\\
 c_{22}(\lambda_-)
 &=
 \det\bigl(
 \Psi_2^{(1)}(\lambda_-),
 \Psi_1^{(2)}(\lambda_-)
 \bigr)=
 \det\bigl(
 \Psi_2^{(1)}(\lambda_-),
 \ii\Psi_1^{(1)}(\lambda_+)
 \bigr)
 =
 \ii c_{21}(\lambda_+).
\end{align*}
Therefore, \eqref{ic:c11-c12-band} and
\eqref{ic:c22-c21-band} hold on $\dot\Sigma_0$.

Finally, \eqref{ic:Psi-tildePsi-relation},
\eqref{ic:G-symmetry-properties}, \eqref{ic:R-cut-real}, and
\eqref{ic:Psi-conjugation-lower-cut} yield
\eqref{ic:HtildePsi-conjugation-lower-cut}. Taking the boundary values
of \eqref{ic:HtildePsi-conjugation-columns} and applying
\eqref{ic:HtildePsi-conjugation-lower-cut} then gives
\eqref{ic:HtildePsi-opposite-side-cut}.
\end{proof}

The preceding symmetry relations also imply the following
nonvanishing property of the scattering coefficients.

\begin{proposition}
\label{prop:ic-no-zeros-continuous-spectrum}
The scattering coefficients satisfy
\begin{align}
 c_{11}(\lambda)&\neq0,
 &&\lambda\in\dot\Sigma_1,
 \label{ic:c11-nonzero-outer}\\
 c_{11}(\lambda_+)&\neq0,
 \qquad
 c_{21}(\lambda_+)\neq0,
 &&\lambda\in\dot\Sigma_0.
 \label{ic:c11-c21-nonzero-band}
\end{align}
\end{proposition}

\begin{proof}
For $\lambda\in\dot\Sigma_1$, Equation
\eqref{ic:c-pseudo-unitarity} gives
\[
 |c_{11}(\lambda)|^2
 =
 1+|c_{21}(\lambda)|^2>0,
\]
and hence $c_{11}(\lambda)\neq0$.

For $\lambda\in\dot\Sigma_0$, Equations
\eqref{ic:c11-c22-conjugation-cut} and
\eqref{ic:c22-c21-band} give
\[
 \overline{c_{11}(\lambda_+)}
 =
 c_{22}(\lambda_-)
 =
 \ii c_{21}(\lambda_+).
\]
Thus $c_{11}(\lambda_+)$ and $c_{21}(\lambda_+)$ vanish
simultaneously. If they vanished at some
$\lambda_0\in\dot\Sigma_0$, the first relation in
\eqref{ic:continued-column-scattering} would imply
 $\Psi_1^{(1)}(\lambda_{0+})=0$,
which is incompatible with the normalization. Therefore,
\[
 c_{11}(\lambda_+)\neq0,
 \qquad
 c_{21}(\lambda_+)\neq0,
 \qquad
 \lambda\in\dot\Sigma_0.
\]
\end{proof}

\subsection{Discrete spectrum and zeros of $c_{11}(\lambda)$}
\label{sec:ic-discrete-spectrum}

Proposition~\ref{prop:ic-no-zeros-continuous-spectrum} implies that
$c_{11}$ has no zeros on $\dot\Sigma_1$ and that
$c_{11}(\lambda_+)$ has no zeros on $\dot\Sigma_0$. To locate the
remaining possible zeros, we rewrite the spatial equation in
\eqref{ic:balanced-lax-pair} as a self-adjoint weighted Dirac problem.

Fix $t\in[0,T)$ and set
\[
 J:=
 \begin{pmatrix}
  0&-1\\
  1&0
 \end{pmatrix}.
\]
Multiplying the spatial equation for $\Psi^{\mathrm b}$ in
\eqref{ic:balanced-lax-pair} by $2J$ and using
\eqref{ic:balanced-U}, we obtain
\begin{equation}\label{ic:balanced-weighted-Dirac-system}
 2J\Psi_x^{\mathrm b}
 +(1+2r_x)\sigma_1\Psi^{\mathrm b}
 =
 \lambda\rho\Psi^{\mathrm b}.
\end{equation}
Accordingly, the equation for each column of $\Psi^{\mathrm b}$ can be
written as the spectral problem
\begin{equation}\label{ic:weighted-Dirac-operator}
 \mathcal D(t)f=\lambda f,
 \qquad
 \mathcal D(t)
 :=
 \frac1{\rho}
 \left(
  2J\frac{\dd}{\dd x}
  +(1+2r_x)\sigma_1
 \right).
\end{equation}
We consider $\mathcal D(t)$ in the weighted Hilbert space
\begin{equation}\label{ic:weighted-Hilbert-space}
 \mathcal H_t
 :=
 L^2(\mathbb R,\rho\dd x;\mathbb C^2),
 \qquad
 \langle f,g\rangle_{\mathcal H_t}
 :=
 \int_{\mathbb R}
 \rho(x,t)f(x)^\dagger g(x)\dd x,
\end{equation}
on the domain
\begin{equation}\label{ic:weighted-Dirac-domain}
 \operatorname{dom}\mathcal D(t)
 :=
 \left\{
 f\in\mathcal H_t\cap
 AC_{\mathrm{loc}}(\mathbb R;\mathbb C^2):
 \frac1{\rho}
 \bigl(
  2Jf_x+(1+2r_x)\sigma_1f
 \bigr)\in\mathcal H_t
 \right\}.
\end{equation}
Since $m,n>0$, for any $h=(h_1,h_2)^{\mathsf T}\in\mathbb C^2$, \eqref{ic:G-def} implies
\begin{equation}\label{ic:weighted-norm-balanced}
 \rho|Gh|^2
 =
 n|h_1|^2+m|h_2|^2.
\end{equation}

\begin{proposition}\label{prop:ic-self-adjoint-Dirac}
The operator $\mathcal D(t)$ defined by
\eqref{ic:weighted-Dirac-operator}--\eqref{ic:weighted-Dirac-domain}
is self-adjoint in $\mathcal H_t$, and
\begin{equation}\label{ic:essential-spectrum}
 \sigma_{\mathrm{ess}}\bigl(\mathcal D(t)\bigr)
 =
 \Sigma_2
\end{equation}
whereas
its discrete spectrum is real and
\begin{equation}\label{ic:discrete-spectrum-location}
 \sigma_{\mathrm{disc}}\bigl(\mathcal D(t)\bigr)
 \subset\dot\Gamma_2.
\end{equation}
\end{proposition}

\begin{proof}
Define
\[
 s=s(x,t):=\int_0^x\rho(\xi,t)\dd\xi.
\]
Since $\rho(x,t)\geq c_t>0$, we have $s_x=\rho>0$, and the definition
of $s$ gives $s(x,t)\to\pm\infty$ as $x\to\pm\infty$. Thus
$x\mapsto s(x,t)$ is a diffeomorphism of $\mathbb R$ onto itself.

For $g(s):=f(x(s,t))$, we have
\[
 \int_{\mathbb R}|g(s)|^2\dd s
 =
 \int_{\mathbb R}\rho(x,t)|f(x)|^2\dd x.
\]
Hence this change of variables is unitary from $\mathcal H_t$ onto
$L^2(\mathbb R;\mathbb C^2)$. Since $f_x=\rho g_s$, the equation
$\mathcal D(t)f=\lambda f$ becomes
\begin{equation}\label{ic:unweighted-Dirac-system}
 2Jg_s+q(s,t)\sigma_1g=\lambda g,
 \qquad
 q(s,t)
 :=
 \frac{1+2r_x(x(s,t),t)}
      {\rho(x(s,t),t)}.
\end{equation}
The function $q$ is real and bounded. The operator
$2J\frac{\dd}{\dd s}$ is self-adjoint on
$H^1(\mathbb R;\mathbb C^2)$, and adding the bounded Hermitian
matrix $q\sigma_1$ preserves self-adjointness. Since the change of
variables maps \eqref{ic:weighted-Dirac-domain} onto
$H^1(\mathbb R;\mathbb C^2)$, the operator $\mathcal D(t)$ is
self-adjoint in $\mathcal H_t$.

Moreover, $q(s,t)\to A_j^{-1}$ as
$s\to(-1)^j\infty$, $j=1,2$, and
\eqref{ic:standing-direct-assumptions} implies that
$q(\cdot,t)-A_1^{-1}\in L^1((-\infty,0))$ and
$q(\cdot,t)-A_2^{-1}\in L^1((0,\infty))$.

As $s\to(-1)^j\infty$, equation
\eqref{ic:unweighted-Dirac-system} approaches
\[
 2Jg_s+\frac1{A_j}\sigma_1g=\lambda g.
\]
Substitution of $g(s)=\eul^{\ii\xi s}g_0$ gives
$\lambda^2=4\xi^2+A_j^{-2}$. Thus the limiting equation has
oscillatory solutions precisely for $\lambda\in\Sigma_j$. Since the
coefficients of \eqref{ic:unweighted-Dirac-system} differ integrably
from their limits, its essential spectrum is determined by these two
limiting equations. Therefore,
\[
 \sigma_{\mathrm{ess}}\bigl(\mathcal D(t)\bigr)
 =
 \Sigma_1\cup\Sigma_2
 =
 \Sigma_2,
\]
where the last equality follows from $0<A_1<A_2$.

Finally, self-adjointness implies that the discrete spectrum of
$\mathcal D(t)$ is real. Any discrete eigenvalue must lie outside
$\Sigma_2$, and hence
\[
 \sigma_{\mathrm{disc}}\bigl(\mathcal D(t)\bigr)
 \subset\mathbb R\setminus\Sigma_2
 =
 \dot\Gamma_2.
\]
\end{proof}

\begin{corollary}\label{cor:ic-scattering-zeros}
For $\lambda\in\dot\Gamma_2\setminus\{0\}$,
\begin{equation}\label{ic:eigenvalue-zero-equivalence}
 \lambda\in\sigma_{\mathrm{disc}}\bigl(\mathcal D(t)\bigr)
 \quad\Longleftrightarrow\quad
 c_{11}(\lambda_+)=0.
\end{equation}
Moreover,
\begin{equation}\label{ic:no-nonreal-scattering-zeros}
 c_{11}(\lambda)\neq0,
 \qquad
 \lambda\in\mathbb C^+,
\end{equation}
and
\begin{equation}\label{ic:no-nonreal-scattering-zeros-lower}
 c_{22}(\lambda)\neq0,
 \qquad
 \lambda\in\mathbb C^-.
\end{equation}
At $\lambda=0$, the scattering
coefficients satisfy
\begin{equation}\label{ic:c11-at-zero}
 c_{11}(0_+)=1
\end{equation}
and
\begin{equation}\label{ic:c22-at-zero}
 c_{22}(0_-)=1.
\end{equation}
In particular, $\lambda=0$ is not an eigenvalue.
\end{corollary}

\begin{proof}
For $\lambda\in\dot\Gamma_2\setminus\{0\}$,
$\bigl(\Psi_1^{\mathrm b}\bigr)^{(1)}(x,t,\lambda_+)$ decays
exponentially as $x\to-\infty$, whereas
$\bigl(\Psi_2^{\mathrm b}\bigr)^{(2)}(x,t,\lambda_+)$ decays
exponentially as $x\to+\infty$. Since $\det G=1$,
\[
 \det\left(
  \bigl(\Psi_1^{\mathrm b}\bigr)^{(1)},
  \bigl(\Psi_2^{\mathrm b}\bigr)^{(2)}
 \right)
 =
 \det\left(
  \Psi_1^{(1)},\Psi_2^{(2)}
 \right)
 =
 c_{11}(\lambda_+).
\]
Thus $c_{11}(\lambda_+)=0$ precisely when the two decaying solutions
are linearly dependent. Their linear dependence implies the existence of a nonzero solution
that decays at both spatial ends and hence belongs to $\mathcal H_t$.
This solution is an eigenfunction of $\mathcal D(t)$. Conversely,
an eigenfunction in $\mathcal H_t$ must be proportional to the
decaying solution at each spatial end. This gives
\eqref{ic:eigenvalue-zero-equivalence}.

The same argument applies for $\lambda\in\mathbb C^+$. A zero of
$c_{11}(\lambda)$ would give an eigenfunction of the self-adjoint
operator $\mathcal D(t)$ with a nonreal eigenvalue. Hence
\eqref{ic:no-nonreal-scattering-zeros} holds, while
\eqref{ic:no-nonreal-scattering-zeros-lower} follows from
\eqref{ic:c11-c22-conjugation-analytic}.

Finally, the analytic continuations of $c_{11}$ and $c_{22}$ are
independent of $t$, so their values at $\lambda=0$ may be computed
from the spatial problem at $t=0$. By
Remark~\ref{rem:ic-H-at-zero},
\[
 H_j^{-1}(0_+)=\ii\sigma_1,
 \qquad
 H_j^{-1}(0_-)=-I,
\]
whereas
\[
 \mathsf U(x,0,0)=-\frac12\sigma_3.
\]
The normalizations of the relevant Jost columns therefore yield
\[
 \Psi_{1,+}^{(1)}(x,0,0)
 =
 \ii\eul^{x/2}
 \begin{pmatrix}0\\1\end{pmatrix},
 \qquad
 \Psi_{2,+}^{(2)}(x,0,0)
 =
 \ii\eul^{-x/2}
 \begin{pmatrix}1\\0\end{pmatrix},
\]
and
\[
 \Psi_{2,-}^{(1)}(x,0,0)
 =
 -\eul^{-x/2}
 \begin{pmatrix}1\\0\end{pmatrix},
 \qquad
 \Psi_{1,-}^{(2)}(x,0,0)
 =
 -\eul^{x/2}
 \begin{pmatrix}0\\1\end{pmatrix}.
\]
Consequently, the determinant representations of the diagonal
scattering coefficients imply
\[
 \begin{aligned}
 c_{11}(0_+)
 &=
 \det\bigl(
  \Psi_{1,+}^{(1)}(x,0,0),
  \Psi_{2,+}^{(2)}(x,0,0)
 \bigr)
 =1,\\
 c_{22}(0_-)
 &=
 \det\bigl(
  \Psi_{2,-}^{(1)}(x,0,0),
  \Psi_{1,-}^{(2)}(x,0,0)
 \bigr)
 =1.
 \end{aligned}
\]
 Furthermore, every vector solution of the spatial
equation at $\lambda=0$ has the form
\[
 h(x)=
 \begin{pmatrix}
  a\eul^{-x/2}\\
  b\eul^{x/2}
 \end{pmatrix},
 \qquad a,b\in\mathbb C.
\]
The corresponding solution of the balanced Dirac equation \eqref{ic:balanced-weighted-Dirac-system} is $Gh$.
Using \eqref{ic:weighted-norm-balanced}, we obtain
\[
 \|Gh\|_{\mathcal H_t}^2
 =
 \int_{\mathbb R}
 \left(
  n(x,t)|a|^2\eul^{-x}
  +
  m(x,t)|b|^2\eul^{x}
 \right)\dd x.
\]
Since $m$ and $n$ are bounded away from zero, this integral is finite
only if $a=b=0$. Hence $\lambda=0$ is not an eigenvalue.
\end{proof}

\medskip
\noindent
\textit{Assumption.}
We assume that $\mathcal D(t)$ has finitely many eigenvalues in
$\dot\Gamma_2$, or equivalently, that $c_{11}(\lambda_+)$ has
finitely many zeros there.

Denote the zeros of
$c_{11}(\lambda_+)$ in $\dot\Gamma_2$ by
$\{\lambda_k\}_{k=1}^{N}$. For each $\lambda_k$, there exists a
nonzero constant $\mathfrak b_k$ such that
\begin{equation}\label{ic:eigenfunction-linear-dependence}
 \Psi_1^{(1)}(x,t,\lambda_{k+})
 =
 \mathfrak b_k
 \Psi_2^{(2)}(x,t,\lambda_{k+}).
\end{equation}
Since both columns in
\eqref{ic:eigenfunction-linear-dependence} are purely imaginary by
\eqref{ic:Psi-conjugation-upper-cut},
\[
 \mathfrak b_k\in\mathbb R\setminus\{0\}.
\]

\begin{proposition}\label{prop:ic-simple-zeros}
Every zero of $c_{11}(\lambda_+)$ in $\dot\Gamma_2$ is simple.
\end{proposition}

\begin{proof}
Let $\lambda_k\in\dot\Gamma_2$ be a zero of
$c_{11}(\lambda_+)$ and set
\[
 \phi_k(x,t)
 :=
 \Psi_2^{(2)}(x,t,\lambda_{k+})
 =
 \begin{pmatrix}
  \phi_{k,1}(x,t)\\
  \phi_{k,2}(x,t)
 \end{pmatrix}.
\]
For $\lambda$ on the upper side of $\dot\Gamma_2$, introduce
\[
 y_1(x,\lambda)
 :=
 \Psi_1^{(1)}(x,t,\lambda_+),
 \qquad
 y_2(x,\lambda)
 :=
 \Psi_2^{(2)}(x,t,\lambda_+).
\]
Differentiating the spatial equation in \eqref{eq:lax} with respect to
$\lambda$, we obtain
\begin{equation}\label{ic:lambda-differentiated-spatial}
 \partial_x\bigl(\partial_\lambda y_j\bigr)
 =
 \mathsf U\,\partial_\lambda y_j
 +(\partial_\lambda\mathsf U)y_j,
 \qquad j=1,2,
\end{equation}
where
\begin{equation}\label{ic:U-lambda}
 \partial_\lambda\mathsf U
 =
 \frac12
 \begin{pmatrix}
  0&m\\
  -n&0
 \end{pmatrix}.
\end{equation}

At $\lambda=\lambda_k$, define
\[
 F(x)
 :=
 \det\bigl(\partial_\lambda y_1,y_2\bigr)(x,\lambda_k),
 \qquad
 K(x)
 :=
 \det\bigl(y_1,\partial_\lambda y_2\bigr)(x,\lambda_k).
\]
 By
\eqref{ic:eigenfunction-linear-dependence},
\[
 y_1(x,\lambda_k)=\mathfrak b_k\phi_k(x,t),
 \qquad
 y_2(x,\lambda_k)=\phi_k(x,t).
\]
Using \eqref{ic:lambda-differentiated-spatial},
$\operatorname{tr}\mathsf U=0$, and \eqref{ic:U-lambda}, we find
\begin{align*}
 F_x
 &=
 \mathfrak b_k
 \det\bigl((\partial_\lambda\mathsf U)\phi_k,\phi_k\bigr),\\
 K_x
 &=
 \mathfrak b_k
 \det\bigl(\phi_k,(\partial_\lambda\mathsf U)\phi_k\bigr).
\end{align*}
The components of $\phi_k$ are purely imaginary by
\eqref{ic:Psi-conjugation-upper-cut}. Therefore,
\begin{subequations}\label{ic:FK-derivatives}
\begin{align}
 F_x
 &=
 -\frac{\mathfrak b_k}{2}
 \left(
  n|\phi_{k,1}|^2+m|\phi_{k,2}|^2
 \right),
 \label{ic:F-derivative}\\
 K_x
 &=
 \frac{\mathfrak b_k}{2}
 \left(
  n|\phi_{k,1}|^2+m|\phi_{k,2}|^2
 \right).
 \label{ic:K-derivative}
\end{align}
\end{subequations}

The functions $y_1$ and $\partial_\lambda y_1$ decay as
$x\to-\infty$, while $y_2$ and $\partial_\lambda y_2$ decay as
$x\to+\infty$. Hence
\[
 F(-\infty)=0,
 \qquad
 K(+\infty)=0.
\]
Integrating \eqref{ic:FK-derivatives} gives
\begin{align*}
 F(x)
 &=
 -\frac{\mathfrak b_k}{2}
 \int_{-\infty}^{x}
 \left(
  n|\phi_{k,1}|^2+m|\phi_{k,2}|^2
 \right)\dd\xi,\\
 K(x)
 &=
 -\frac{\mathfrak b_k}{2}
 \int_x^{+\infty}
 \left(
  n|\phi_{k,1}|^2+m|\phi_{k,2}|^2
 \right)\dd\xi.
\end{align*}
Since
$c_{11}(\lambda_+)=\det(y_1,y_2)$, we obtain
\begin{equation}\label{ic:c11-derivative-norming}
 \left.
 \partial_\lambda c_{11}(\lambda_+)
 \right|_{\lambda=\lambda_k}
 =
 F(x)+K(x)
 =
 -\frac{\mathfrak b_k}{2}
 \int_{-\infty}^{+\infty}
 \left(
  n|\phi_{k,1}|^2+m|\phi_{k,2}|^2
 \right)\dd x.
\end{equation}
The integral is strictly positive and
$\mathfrak b_k\neq0$. Consequently,
\[
 \left.
 \partial_\lambda c_{11}(\lambda_+)
 \right|_{\lambda=\lambda_k}
 \neq0,
\]
and the zero $\lambda_k$ is simple.

Finally, setting $\phi_k^{\mathrm b}:=G\phi_k$ and using
\eqref{ic:weighted-norm-balanced}, we may equivalently write
\begin{equation}\label{ic:c11-derivative-balanced}
 \left.
 \partial_\lambda c_{11}(\lambda_+)
 \right|_{\lambda=\lambda_k}
 =
 -\frac{\mathfrak b_k}{2}
 \int_{-\infty}^{+\infty}
 \rho|\phi_k^{\mathrm b}|^2\dd x.
\end{equation}
\end{proof}

For the RH formulation, define the norming constants
\begin{equation}\label{ic:norming-constants-def}
 \varkappa_k
 :=
 \frac{\mathfrak b_k}{c'_{11,k}},
 \qquad
 \widetilde\varkappa_k
 :=
 \frac1{\mathfrak b_kc'_{11,k}},
\end{equation}
where $ c'_{11,k}
 :=
 \left.
 \partial_\lambda c_{11}(\lambda_+)
 \right|_{\lambda=\lambda_k}$. By
Proposition~\ref{prop:ic-simple-zeros},
\begin{subequations}\label{ic:norming-constants-integral}
\begin{align}
 \varkappa_k
 &=
 -2
 \left[
  \int_{-\infty}^{+\infty}
  \left(
   n|\phi_{k,1}|^2+m|\phi_{k,2}|^2
  \right)\dd x
 \right]^{-1}
 <0,
 \label{ic:right-norming-integral}\\
 \widetilde\varkappa_k
 &=
 -\frac{2}{\mathfrak b_k^2}
 \left[
  \int_{-\infty}^{+\infty}
  \left(
   n|\phi_{k,1}|^2+m|\phi_{k,2}|^2
  \right)\dd x
 \right]^{-1}
 <0.
 \label{ic:left-norming-integral}
\end{align}
\end{subequations}
In particular, both norming constants are real and nonzero. Since the
columns in \eqref{ic:eigenfunction-linear-dependence} satisfy the same
Lax pair, $\mathfrak b_k$ is independent of $x$ and $t$. Therefore,
$\varkappa_k$ and $\widetilde\varkappa_k$ are also independent of
$x$ and $t$.

By \eqref{ic:c11-minus-upper-cut}, the zeros occur in pairs:
\begin{equation}\label{ic:discrete-spectrum-parity}
 c_{11}(\lambda_{k+})=0
 \quad\Longrightarrow\quad
 c_{11}((-\lambda_k)_+)=0.
\end{equation}
Let $\lambda_{k^*}=-\lambda_k$. Substitution of
\eqref{ic:Psi-minus-upper-cut} into
\eqref{ic:eigenfunction-linear-dependence} and differentiation of
\eqref{ic:c11-minus-upper-cut} give
\begin{equation}\label{ic:norming-parity-input}
 \mathfrak b_{k^*}=-\mathfrak b_k,
 \qquad
 c'_{11,k^*}=-c'_{11,k}.
\end{equation}
Consequently,
\begin{equation}\label{ic:norming-parity}
 \varkappa_{k^*}
 =
 \varkappa_k,
 \qquad
 \widetilde\varkappa_{k^*}
 =
 \widetilde\varkappa_k.
\end{equation}

\subsection{Behavior at the branch points}
\label{sec:ic-branch-point-behavior}

At the branch points $\lambda=\pm A_j^{-1}$, the functions $l_j$
vanish $l_j\left(\pm A_j^{-1}\right)=0$, $j=1,2$. To describe the local behavior of the Jost solutions, introduce
the local parameters
\begin{equation}\label{ic:omega-branch-def}
 \omega_j^+(\lambda)
 :=
 \sqrt{\lambda-A_j^{-1}},
 \qquad
 \omega_j^-(\lambda)
 :=
 \sqrt{\lambda+A_j^{-1}}.
\end{equation}
The branch cuts of $\omega_j^+$ and $\omega_j^-$ are
$(-\infty,A_j^{-1}]$ and $[-A_j^{-1},+\infty)$, respectively.
Since $\lambda=0$ lies on both cuts, we specify the branches by
their boundary values
\[
 \omega_j^+(0_\pm)=\pm\ii A_j^{-1/2},
 \qquad
 \omega_j^-(0_\pm)=\mp A_j^{-1/2},
\]
where $0_+$ and $0_-$ denote approach from the upper and lower
half-planes, respectively. With these choices,
$\omega_j^-(-\lambda)=-\ii\omega_j^+(\lambda)$.
The expansions below are understood columnwise in the half-plane
of analyticity of each column.

\begin{proposition}\label{prop:ic-tildePsi-branch-expansion}
For $j=1,2$, there exist real-valued functions
$\chi_j(x,t)$, $a_j(x,t)$, and $b_j(x,t)$ such that
\begin{align}
 \widetilde\Psi_j(x,t,\lambda)
 &=
 \frac{\ii\chi_j(x,t)}{\omega_j^+(\lambda)}
 \begin{pmatrix}
  1&1\\
  -1&-1
 \end{pmatrix}
 +
 \begin{pmatrix}
  a_j(x,t)&b_j(x,t)\\
  b_j(x,t)&a_j(x,t)
 \end{pmatrix}
 +\ord\bigl(\omega_j^+(\lambda)\bigr),
 \qquad
 \lambda\to A_j^{-1},
 \label{ic:tildePsi-positive-branch}\\
 \widetilde\Psi_j(x,t,\lambda)
 &=
 \frac{\chi_j(x,t)}{\omega_j^-(\lambda)}
 \begin{pmatrix}
  -1&1\\
  -1&1
 \end{pmatrix}
 +
 \begin{pmatrix}
  a_j(x,t)&-b_j(x,t)\\
  -b_j(x,t)&a_j(x,t)
 \end{pmatrix}
 +\ord\bigl(\omega_j^-(\lambda)\bigr),
 \qquad
 \lambda\to-A_j^{-1}.
 \label{ic:tildePsi-negative-branch}
\end{align}
\end{proposition}

\begin{proof}
We first consider the behavior near $A_j^{-1}$. Introduce
\begin{equation}\label{ic:W-plus-def}
 W_j^+(\lambda)
 :=
 \begin{pmatrix}
  1&\dfrac{\ii}{\omega_j^+(\lambda)}\\[2mm]
  1&-\dfrac{\ii}{\omega_j^+(\lambda)}
 \end{pmatrix},
 \qquad
 X_j^+:=(W_j^+)^{-1}\widetilde\Psi_j.
\end{equation}
Then, in particular,
\[
 (W_j^+)^{-1}
 =
 \frac12
 \begin{pmatrix}
  1&1\\
  -\ii\omega_j^+&\ii\omega_j^+
 \end{pmatrix}.
\]
Substituting \eqref{ic:W-plus-def} into the Volterra equation
\eqref{ic:volterra-tildePsi} yields a Volterra equation for $X_j^+$:
\begin{equation}\label{ic:volterra-X-plus}
\begin{aligned}
 X_j^+(x,t,\lambda)
 &=
 \bigl(W_j^+(\lambda)\bigr)^{-1}+
 \int_{(-1)^j\infty}^{x}
 \bigl(W_j^+(\lambda)\bigr)^{-1}
 \eul^{\frac{\ii l_j(\lambda)}2
       \int_x^\xi\rho(\tau,t)\dd\tau\,\sigma_3}
 \widehat W_j(\xi,t,\lambda)W_j^+(\lambda)
 X_j^+(\xi,t,\lambda)
 \eul^{-\frac{\ii l_j(\lambda)}2
       \int_x^\xi\rho(\tau,t)\dd\tau\,\sigma_3}
 \dd\xi .
\end{aligned}
\end{equation}
A direct calculation shows that its kernel extends continuously to
$\lambda=A_j^{-1}$. The standard
successive-approximation argument therefore gives
\[
 X_j^+
 =
 X_{j,0}^+
 +\omega_j^+X_{j,1}^+
 +\ord\bigl((\omega_j^+)^2\bigr).
\]
Multiplying by $W_j^+$ and using the  symmetries established
in Section \ref{sec:ic-symmetries}, we obtain \eqref{ic:tildePsi-positive-branch}, with real-valued
coefficients $\chi_j$, $a_j$, and $b_j$.

Near $-A_j^{-1}$, introduce instead
\[
 W_j^-(\lambda)
 :=
 \begin{pmatrix}
  \dfrac{\ii}{\omega_j^-(\lambda)}&1\\[2mm]
  \dfrac{\ii}{\omega_j^-(\lambda)}&-1
 \end{pmatrix},
 \qquad
 X_j^-:=(W_j^-)^{-1}\widetilde\Psi_j.
\]
A similar argument as above yields
\[
 \widetilde\Psi_j
 =
 \frac{\beta_j}{\omega_j^-}
 \begin{pmatrix}
  -1&1\\
  -1&1
 \end{pmatrix}
 +
 \begin{pmatrix}
  \widehat a_j&\widehat b_j\\
  \widehat b_j&\widehat a_j
 \end{pmatrix}
 +\ord(\omega_j^-).
\]
The $\lambda\mapsto-\lambda$ symmetry of $\widetilde\Psi_j$ then
implies
\[
 \beta_j=\chi_j,
 \qquad
 \widehat a_j=a_j,
 \qquad
 \widehat b_j=-b_j,
\]
which proves \eqref{ic:tildePsi-negative-branch}.
\end{proof}

To describe the behavior of $H_j^{-1}$, set
\begin{equation}\label{ic:zeta-branch-def}
 \frac{1}{\zeta_j(\lambda)}
 :=
 \frac{1}{\sqrt{2}}
 \sqrt{\frac{1}{\ii A_j l_j(\lambda)}},
 \qquad
 \zeta_j(\lambda)
 :=
 \left(
 \frac{1}{\sqrt{2}}
 \sqrt{\frac{1}{\ii A_j l_j(\lambda)}}
 \right)^{-1},
 \qquad l_j(\lambda)\neq 0.
\end{equation}
Here the square root is taken on
$\mathbb C\setminus[0,\infty)$ and is normalized by
$\sqrt{-1}=\ii$, as in \eqref{ic:H-inverse}. Its boundary values are
\[
 \sqrt{r+\ii0}=\sqrt r,
 \qquad
 \sqrt{r-\ii0}=-\sqrt r,
 \qquad r>0.
\]
In particular,
\[
 \zeta_j(\lambda)^2=2\ii A_j l_j(\lambda).
\]

A direct expansion of
\eqref{ic:H-inverse} gives the following result.

\begin{proposition}\label{prop:ic-H-inverse-branch-expansion}
As $\lambda\to A_j^{-1}$,
\begin{equation}\label{ic:H-inverse-positive-branch}
 H_j^{-1}(\lambda)
 =
 \frac1{\zeta_j(\lambda)}
 \begin{pmatrix}
  1&1\\
  1&1
 \end{pmatrix}
 +
 \frac{\zeta_j(\lambda)}4
 \begin{pmatrix}
  1&-1\\
  -1&1
 \end{pmatrix}
 +\ord\bigl(\zeta_j(\lambda)^3\bigr).
\end{equation}
As $\lambda\to-A_j^{-1}$,
\begin{equation}\label{ic:H-inverse-negative-branch}
 H_j^{-1}(\lambda)
 =
 \frac1{\zeta_j(\lambda)}
 \begin{pmatrix}
  -1&1\\
  1&-1
 \end{pmatrix}
 -
 \frac{\zeta_j(\lambda)}4
 \begin{pmatrix}
  1&1\\
  1&1
 \end{pmatrix}
 +\ord\bigl(\zeta_j(\lambda)^3\bigr).
\end{equation}
\end{proposition}

\begin{corollary}\label{cor:ic-Jost-fourth-root}
For fixed $(x,t)$, the balanced and original Jost solutions have at
most fourth-root singularities at the branch points:
\begin{equation}\label{ic:Jost-fourth-root}
 \Psi_j^{\mathrm b}(x,t,\lambda),
 \ \Psi_j(x,t,\lambda)
 =
 \ord\bigl(
 |\lambda\mp A_j^{-1}|^{-1/4}
 \bigr),
 \qquad
 \lambda\to\pm A_j^{-1},
 \qquad j=1,2.
\end{equation}
\end{corollary}

\begin{proof}
Combining
\eqref{ic:tildePsi-positive-branch}--\eqref{ic:tildePsi-negative-branch}
with
\eqref{ic:H-inverse-positive-branch}--\eqref{ic:H-inverse-negative-branch},
and observing that the products of the leading matrices vanish, gives
\[
 H_j^{-1}\widetilde\Psi_j
 =
 \ord(\zeta_j^{-1}).
\]
Since
$|\zeta_j(\lambda)|\asymp
|\lambda\mp A_j^{-1}|^{1/4}$ and the remaining gauge factors are
bounded at the branch points, \eqref{ic:Jost-fourth-root} follows.
\end{proof}

The determinant representations
\eqref{ic:scattering-determinants} and the branch-point expansions above
yield the corresponding behavior of the scattering coefficients.

\begin{corollary}\label{cor:ic-c-branch-expansion}
There exist constants $\alpha_1\in\mathbb C$ and
$\gamma_2\in\mathbb R$ such that the following expansions hold.

As $\lambda\to A_1^{-1}$; $\lambda\in \dot\Sigma_1$,
\begin{subequations}\label{ic:c-positive-A1-outer}
\begin{align}
 c_{11}(\lambda)
 &=
 \frac{\alpha_1}{\zeta_1(\lambda)}
 +\ord\bigl(\zeta_1(\lambda)\bigr),\\
 c_{21}(\lambda)
 &=
 \frac{\ii\overline{\alpha_1}}{\zeta_1(\lambda)}
 +\ord\bigl(\zeta_1(\lambda)\bigr).
\end{align}
\end{subequations}
As $\lambda\to A_1^{-1}$; $\lambda\in \dot\Sigma_0$,
\begin{subequations}\label{ic:c-positive-A1-band}
\begin{align}
 c_{11}(\lambda_+)
 &=
 \frac{\alpha_1}{\zeta_1(\lambda_+)}
 +\ord\bigl(\zeta_1(\lambda_+)\bigr),\\
 c_{21}(\lambda_+)
 &=
 \frac{\ii\overline{\alpha_1}}{\zeta_1(\lambda_+)}
 +\ord\bigl(\zeta_1(\lambda_+)\bigr).
\end{align}
\end{subequations}

As $\lambda\to-A_1^{-1}$, $\lambda\in\dot\Sigma_1$,
\begin{subequations}\label{ic:c-negative-A1-outer}
\begin{align}
 c_{11}(\lambda)
 &=
 -\frac{\overline{\alpha_1}}{\zeta_1(\lambda)}
 +\ord\bigl(\zeta_1(\lambda)\bigr),\\
 c_{21}(\lambda)
 &=
 -\frac{\ii\alpha_1}{\zeta_1(\lambda)}
 +\ord\bigl(\zeta_1(\lambda)\bigr).
\end{align}
\end{subequations}
As $\lambda\to-A_1^{-1}$, $\lambda\in\dot\Sigma_0$,
\begin{subequations}\label{ic:c-negative-A1-band}
\begin{align}
 c_{11}(\lambda_+)
 &=
 -\frac{\overline{\alpha_1}}{\zeta_1(\lambda_+)}
 +\ord\bigl(\zeta_1(\lambda_+)\bigr),\\
 c_{21}(\lambda_+)
 &=
 -\frac{\ii\alpha_1}{\zeta_1(\lambda_+)}
 +\ord\bigl(\zeta_1(\lambda_+)\bigr).
\end{align}
\end{subequations}

As $\lambda\to A_2^{-1}$; $\lambda\in \dot\Sigma_0$,
\begin{subequations}\label{ic:c-positive-A2-band}
\begin{align}
 c_{11}(\lambda_+)
 &=
 \frac{\ii\gamma_2}{\zeta_2(\lambda_+)}
 +\ord\bigl(\zeta_2(\lambda_+)\bigr),\\
 c_{21}(\lambda_+)
 &=
 -\frac{\ii\gamma_2}{\zeta_2(\lambda_+)}
 +\ord\bigl(\zeta_2(\lambda_+)\bigr).
\end{align}
\end{subequations}
As $\lambda\to- A_2^{-1}$; $\lambda\in \dot\Sigma_0$,
\begin{subequations}\label{ic:c-negative-A2-band}
\begin{align}
 c_{11}(\lambda_+)
 &=
 \frac{\ii\gamma_2}{\zeta_2(\lambda_+)}
 +\ord\bigl(\zeta_2(\lambda_+)\bigr),\\
 c_{21}(\lambda_+)
 &=
 \frac{\ii\gamma_2}{\zeta_2(\lambda_+)}
 +\ord\bigl(\zeta_2(\lambda_+)\bigr).
\end{align}
\end{subequations}
 Moreover,
\begin{equation}\label{ic:alpha1-nonzero}
 \alpha_1\neq0.
\end{equation}

\end{corollary}

\begin{proof}
Combining
\eqref{ic:tildePsi-positive-branch}--\eqref{ic:tildePsi-negative-branch}
with
\eqref{ic:H-inverse-positive-branch}--\eqref{ic:H-inverse-negative-branch},
\eqref{ic:Psi-b-j-def}, and \eqref{ic:Psi-j-def}, we obtain 
\begin{align}
 \Psi_j(x,t,\lambda)
 &=
 \frac1{\zeta_j(\lambda)}
 \bigl(w_j,w_j\bigr)
 +\ord\bigl(\zeta_j(\lambda)\bigr),
 &&\lambda\to A_j^{-1},
 \label{ic:Psi-positive-branch-leading}\\
 \Psi_j(x,t,\lambda)
 &=
 \frac1{\zeta_j(\lambda)}
 \bigl(-\sigma_3w_j,\sigma_3w_j\bigr)
 +\ord\bigl(\zeta_j(\lambda)\bigr),
 &&\lambda\to-A_j^{-1},
 \label{ic:Psi-negative-branch-leading}
\end{align}
where $w_j=w_j(x,t)$ is some real
vectors.

At $\lambda=A_1^{-1}$, the columns of $\Psi_2$ are regular. Set
\[
 z_2:=\Psi_2^{(2)}(x,t,A_1^{-1}),
 \qquad
 \alpha_1:=\det(w_1,z_2).
\]
By the symmetry \eqref{ic:Psi-conjugation-columns},
\[
 \Psi_2^{(1)}(x,t,A_1^{-1})
 =
 -\ii\overline{z_2}.
\]
Substitution into \eqref{ic:scattering-determinants} gives
\eqref{ic:c-positive-A1-outer}. The same local expansions on the upper
side of the adjacent interval $\dot\Sigma_0$ give
\eqref{ic:c-positive-A1-band}. The formulas at $-A_1^{-1}$ follow from
\eqref{ic:Psi-negative-branch-leading} and the symmetry \eqref{ic:Psi-minus-columns} of the
Jost solutions.

At $\lambda=A_2^{-1}$, set
\[
 z_1:=\Psi_{1,+}^{(1)}(x,t,A_2^{-1}).
\]
The vector $z_1$ is purely imaginary, whereas $w_2$ is real. Therefore,
there exists $\gamma_2\in\mathbb R$ such that
\[
 \det(z_1,w_2)=\ii\gamma_2.
\]
The determinant representations then yield
\eqref{ic:c-positive-A2-band}; the  symmetry \eqref{ic:Psi-minus-columns} implies
\eqref{ic:c-negative-A2-band}. 
By Proposition~\ref{prop:scattering-matrix}, the coefficients
$\alpha_1$ and $\gamma_2$ are independent of $x$ and $t$.

Finally, if $\alpha_1=0$, then
\[
 c_{11}(\lambda),\ c_{21}(\lambda)
 =
 \ord\bigl(\zeta_1(\lambda)\bigr),
 \qquad
 \lambda\to A_1^{-1},\quad\lambda\in\dot\Sigma_1,
\]
which contradicts
\[
 |c_{11}(\lambda)|^2-|c_{21}(\lambda)|^2=1.
\]
Hence \eqref{ic:alpha1-nonzero} holds.
\end{proof}

\subsection{Behavior near $\lambda=0$}
\label{sec:ic-zero-behavior}

In this section, we study the behavior of the Jost solutions near
$\lambda=0$, which will later be used to derive the recovery formulas.
Since $0$ lies on both cuts $\Gamma_1$ and $\Gamma_2$, we consider the
expansions in $\mathbb C^+$ and $\mathbb C^-$ separately.

For the Camassa--Holm equation and related CH-type systems, the analysis
near $\lambda=0$ is often based on an additional gauge transformation;
see, for example, \cites{BS08,BS15,BS13}. As in the scalar mCH case \cite{BKS20}, no
further matrix gauge is needed here. It is enough to regroup the terms
in the transformed Lax pair \eqref{ic:transformed-lax}.

We first introduce
\begin{equation}\label{ic:D12-def}
 D_1(x,t)
 :=
 \int_{-\infty}^{x}\delta_1(\xi,t)\dd\xi,
 \qquad
 D_2(x,t)
 :=
 \int_x^{+\infty}\delta_2(\xi,t)\dd\xi.
\end{equation}
Using \eqref{ic:delta-def}, these functions can be written as
\begin{equation}\label{ic:D12-r-identities}
\begin{aligned}
 D_1(x,t)
 &=
 \frac1{2A_1}
 \int_{-\infty}^{x}
 \bigl(\rho(\xi,t)-A_1\bigr)\dd\xi-r(x,t),\\
 D_2(x,t)
 &=
 \frac1{2A_2}
 \int_x^{+\infty}
 \bigl(\rho(\xi,t)-A_2\bigr)\dd\xi+r(x,t).
\end{aligned}
\end{equation}

Define the modified phases
\begin{equation}\label{ic:f0-def}
 f_1^0
 :=
 f_1-\ii A_1l_1D_1,
 \qquad
 f_2^0
 :=
 f_2+\ii A_2l_2D_2.
\end{equation}
In view of \eqref{ic:D12-r-identities}, both phases have the common
representation
\begin{equation}\label{ic:f0-common-form}
 f_j^0(x,t,\lambda)
 =
 \ii A_jl_j(\lambda)
 \left[
  \frac{x}{2}+r(x,t)
  -t\left(
   \frac1{\lambda^2}+\frac{A_j^2}{2}
  \right)
 \right],
 \qquad j=1,2.
\end{equation}
Set
\begin{equation}\label{ic:tildePsi0-def}
 \widetilde\Psi_j^0
 :=
 \widehat\Psi_j\eul^{f_j^0\sigma_3}.
\end{equation}
By \eqref{ic:f-derivatives},
\eqref{ic:rho-r-conservation}, and \eqref{ic:delta-def},
\begin{align}
 f_{j,x}^0-f_{j,x}
 &=
 -\ii A_jl_j\delta_j,
 \label{ic:f0-f-x-difference}\\
 f_{j,t}^0-f_{j,t}
 &=
 \ii A_jl_j
 \left[
  Q\left(\delta_j+\frac12\right)
  -\frac{A_j^2}{2}
 \right].
 \label{ic:f0-f-t-difference}
\end{align}
Consequently, the Lax pair \eqref{ic:transformed-lax} takes the form
\begin{subequations}\label{ic:zero-adapted-lax}
\begin{align}
 \widetilde\Psi_{j,x}^0
 +[f_{j,x}^0\sigma_3,\widetilde\Psi_j^0]
 &=
 \widehat W_j^0\widetilde\Psi_j^0,
 \label{ic:zero-adapted-lax-x}\\
 \widetilde\Psi_{j,t}^0
 +[f_{j,t}^0\sigma_3,\widetilde\Psi_j^0]
 &=
 \widehat V_j^0\widetilde\Psi_j^0.
 \label{ic:zero-adapted-lax-t}
\end{align}
\end{subequations}
Here
\begin{align}
 \widehat W_j^0
 &:=
 \widehat W_j-\ii A_jl_j\delta_j\sigma_3 =
 \frac{\lambda\delta_j}{l_j}
 \bigl(\sigma_2-\ii A_j\lambda\sigma_3\bigr),
 \label{ic:W0-def}\\
 \widehat V_j^0
 &:=
 \widehat V_j
 +\ii A_jl_j
 \left[
  Q\left(\delta_j+\frac12\right)
  -\frac{A_j^2}{2}
 \right]\sigma_3.
 \label{ic:V0-def}
\end{align}
In particular,
\begin{equation}\label{ic:W0-at-zero}
 \widehat W_j^0(x,t,\lambda)=\ord(\lambda),
 \qquad \lambda\to0.
\end{equation}

\begin{remark}\label{rem:ic-zero-phase-scalar-reduction}
The term $r(x,t)$ in \eqref{ic:f0-common-form} arises from the diagonal
contribution $-r_x\sigma_3$ produced by the balancing gauge and constitutes 
 a crucial difference between  the reconstruction procedures in the two-component and scalar cases: the expansion at
$\lambda=0$ determines $x+2r$, rather than the physical variable
$x$ alone; see Section~\ref{sec:ic-two-component-reconstruction}. Under the
scalar reduction $u=v$, one has $r\equiv0$, and the modified phase
$f_j^0$ reduces to the corresponding phase for the scalar mCH equation \cite{BKS20}.
\end{remark}

The Jost solutions $\widetilde\Psi_j^0$ of
\eqref{ic:zero-adapted-lax} are normalized by
\[
\widetilde\Psi_j^0(x,t,\lambda) \to  I,
 \qquad x\to(-1)^j\infty,
\]
and satisfy the Volterra equations
\begin{equation}\label{ic:volterra-tildePsi0}
\begin{aligned}
 \widetilde\Psi_j^0(x,t,\lambda)
 =I+\int_{(-1)^j\infty}^{x}
 &\eul^{(f_j^0(\xi,t,\lambda)-f_j^0(x,t,\lambda))\sigma_3}
 \widehat W_j^0(\xi,t,\lambda)
 \widetilde\Psi_j^0(\xi,t,\lambda)
 \eul^{-(f_j^0(\xi,t,\lambda)-f_j^0(x,t,\lambda))\sigma_3}
 \dd\xi .
\end{aligned}
\end{equation}
As in \eqref{ic:volterra-tildePsi}, the equation is understood
columnwise. Compatibility implies that these solutions also satisfy
the $t$-equation in \eqref{ic:zero-adapted-lax}. The singular
time-dependent term in $f_j^0$ cancels in the phase difference appearing
in \eqref{ic:volterra-tildePsi0}.

\begin{proposition}\label{prop:ic-Psi-zero-expansion}
For each fixed $(x,t)$, the Jost solutions have the following expansions.
As $\lambda\to0$ in $\mathbb C^+$,
\begin{subequations}\label{ic:Psi-upper-zero-expansion}
\begin{align}
 \Psi_1^{(1)}(x,t,\lambda)
 &=
 \ii\eul^{-D_1(x,t)-r(x,t)-f_1(x,t,\lambda)}
 \left[
  \begin{pmatrix}0\\1\end{pmatrix}
  +\frac{\lambda}{2}
   \begin{pmatrix}p(x,t)\\0\end{pmatrix}
  +\ord(\lambda^2)
 \right],
 \label{ic:Psi11-upper-zero}\\
 \Psi_2^{(2)}(x,t,\lambda)
 &=
 \ii\eul^{-D_2(x,t)+r(x,t)+f_2(x,t,\lambda)}
 \left[
  \begin{pmatrix}1\\0\end{pmatrix}
  +\frac{\lambda}{2}
   \begin{pmatrix}0\\q(x,t)\end{pmatrix}
  +\ord(\lambda^2)
 \right].
 \label{ic:Psi22-upper-zero}
\end{align}
\end{subequations}
As $\lambda\to0$ in $\mathbb C^-$,
\begin{subequations}\label{ic:Psi-lower-zero-expansion}
\begin{align}
 \Psi_2^{(1)}(x,t,\lambda)
 &=
 -\eul^{-D_2(x,t)+r(x,t)-f_2(x,t,\lambda)}
 \left[
  \begin{pmatrix}1\\0\end{pmatrix}
  +\frac{\lambda}{2}
   \begin{pmatrix}0\\q(x,t)\end{pmatrix}
  +\ord(\lambda^2)
 \right],
 \label{ic:Psi21-lower-zero}\\
 \Psi_1^{(2)}(x,t,\lambda)
 &=
 -\eul^{-D_1(x,t)-r(x,t)+f_1(x,t,\lambda)}
 \left[
  \begin{pmatrix}0\\1\end{pmatrix}
  +\frac{\lambda}{2}
   \begin{pmatrix}p(x,t)\\0\end{pmatrix}
  +\ord(\lambda^2)
 \right].
 \label{ic:Psi12-lower-zero}
\end{align}
\end{subequations}
\end{proposition}

\begin{proof}
As $\lambda\to0$ in $\mathbb C^+$,
\begin{equation}\label{ic:zero-upper-coefficient-expansions}
 \ii A_jl_j(\lambda)
 =
 -1+\ord(\lambda^2),
 \qquad
 \widehat W_j^0
 =
 -\ii A_j\lambda\delta_j\sigma_2
 +\ord(\lambda^2).
\end{equation}
The first iteration of \eqref{ic:volterra-tildePsi0} therefore gives
\begin{align}
 \bigl(\widetilde\Psi_1^0\bigr)^{(1)}
 &=
 \begin{pmatrix}1\\0\end{pmatrix}
 +\lambda\eta_1(x,t)
  \begin{pmatrix}0\\1\end{pmatrix}
 +\ord(\lambda^2),
 \label{ic:tildePsi01-expansion}\\
 \bigl(\widetilde\Psi_2^0\bigr)^{(2)}
 &=
 \begin{pmatrix}0\\1\end{pmatrix}
 +\lambda\eta_2(x,t)
  \begin{pmatrix}1\\0\end{pmatrix}
 +\ord(\lambda^2),
 \label{ic:tildePsi02-expansion}
\end{align}
where
\begin{align}
 \eta_1(x,t)
 &=
 A_1\int_{-\infty}^{x}
 \eul^{-(x-\xi)-2r(x,t)+2r(\xi,t)}
 \delta_1(\xi,t)\dd\xi,
 \label{ic:eta1-def}\\
 \eta_2(x,t)
 &=
 A_2\int_x^{+\infty}
 \eul^{x-\xi+2r(x,t)-2r(\xi,t)}
 \delta_2(\xi,t)\dd\xi.
 \label{ic:eta2-def}
\end{align}
Combining \eqref{ic:tildePsi01-expansion} and
\eqref{ic:tildePsi02-expansion} with \eqref{H_at_0}, we obtain
\begin{align}
 \bigl(H_1^{-1}\widetilde\Psi_1^0\bigr)^{(1)}
 &=
 \ii
 \left[
  \begin{pmatrix}0\\1\end{pmatrix}
  +\lambda\left(\frac{A_1}{2}+\eta_1\right)
   \begin{pmatrix}1\\0\end{pmatrix}
  +\ord(\lambda^2)
 \right],
 \label{ic:HtildePsi0-first-column}\\
 \bigl(H_2^{-1}\widetilde\Psi_2^0\bigr)^{(2)}
 &=
 \ii
 \left[
  \begin{pmatrix}1\\0\end{pmatrix}
  +\lambda\left(\frac{A_2}{2}+\eta_2\right)
   \begin{pmatrix}0\\1\end{pmatrix}
  +\ord(\lambda^2)
 \right].
 \label{ic:HtildePsi0-second-column}
\end{align}

The definition of $\delta_j$ and the identities
\[
 p+p_x=\rho\eul^{2r},
 \qquad
 q-q_x=\rho\eul^{-2r}
\]
imply
\begin{equation}\label{ic:eta-pq-identities}
 \frac{A_1}{2}+\eta_1
 =
 \frac12\eul^{-2r}p,
 \qquad
 \frac{A_2}{2}+\eta_2
 =
 \frac12\eul^{2r}q.
\end{equation}
Indeed, both sides of the first identity satisfy the scalar equation
\[
 y_x+(1+2r_x)y=\frac{\rho}{2}
\]
and tend to $A_1/2$ as $x\to-\infty$. Likewise, both sides of the
second identity satisfy
\[
 y_x-(1+2r_x)y=-\frac{\rho}{2}
\]
and tend to $A_2/2$ as $x\to+\infty$. Thus
\eqref{ic:eta-pq-identities} follows by uniqueness.

The lower-half-plane expansions
\eqref{ic:Psi-lower-zero-expansion} follow from
\eqref{ic:Psi-upper-zero-expansion} by the  symmetries
\eqref{ic:Psi-conjugation-columns}, using
$ \overline{f_j(x,t,\overline\lambda)}
 =
 -f_j(x,t,\lambda).$
\end{proof}

\begin{remark}\label{rem:ic-Psi-zero-phases}
The exponential factors in
Proposition~\ref{prop:ic-Psi-zero-expansion} can be written without
$r$. Indeed, \eqref{ic:D12-r-identities} gives
\begin{align*}
 -D_1-r\pm f_1
 &=
 -\frac1{2A_1}
 \int_{-\infty}^{x}
 \bigl(\rho(\xi,t)-A_1\bigr)\dd\xi
 \pm f_1,\\
 -D_2+r\pm f_2
 &=
 -\frac1{2A_2}
 \int_x^{+\infty}
 \bigl(\rho(\xi,t)-A_2\bigr)\dd\xi
 \pm f_2.
\end{align*}
Thus, although $r$ is needed in the modified phase
\eqref{ic:f0-common-form}, its explicit contribution cancels from the
exponential factors in the expansions of the original Jost solutions.
\end{remark}

\section{Riemann--Hilbert problems}
\label{sec:ic-rh-problems}

We now formulate the right and left Riemann--Hilbert problems associated
with the scattering data constructed above.
Having both left and right RH problems 
facilitates the large-time asymptotic analysis, which will be addressed elsewhere.
Both RH matrices are built from the balanced Jost solutions $\Psi_j^{\mathrm b}$, $ j=1,2.$

\subsection{Pre-RH problems parametrized by $(x,t)$}
\label{sec:ic-pre-rh}

To formulate the RH problems with the standard normalization at infinity,
introduce the constant matrix
\begin{equation}\label{ic:E-def}
 \mathcal E:=
 -\frac1{\sqrt2}
 \begin{pmatrix}
  1&\ii\\
  \ii&1
 \end{pmatrix}.
\end{equation}
Then
\begin{equation}\label{ic:E-identities}
 \det\mathcal E=1,
 \qquad
 \mathcal E^{-1}=\overline{\mathcal E},
 \qquad
 \mathcal E=\sigma_2\mathcal E\sigma_3,
 \qquad
 \mathcal E=\ii\sigma_1\overline{\mathcal E}.
\end{equation}

Let $\{\lambda_k\}_{k=1}^{N}\subset\dot\Gamma_2$ denote the zeros of
$c_{11}(\lambda_+)$. They are simple and occur in symmetric pairs:
\[
 c_{11}(\lambda_{k+})=0
 \quad\Longleftrightarrow\quad
 c_{11}((-\lambda_k)_+)=0.
\]
We label the paired zero by $\lambda_{k^*}=-\lambda_k$. The corresponding
zeros of $c_{22}(\lambda_-)$ are located at the same points.

For a matrix-valued function with a simple one-sided pole, we write
\begin{align}
 \Res_{\lambda=\lambda_{k+}}F
 &:=
 \lim_{\substack{\lambda\to\lambda_k\\
                  \Im\lambda>0}}
 (\lambda-\lambda_k)F(\lambda),
 \label{ic:upper-side-residue}\\
 \Res_{\lambda=\lambda_{k-}}F
 &:=
 \lim_{\substack{\lambda\to\lambda_k\\
                  \Im\lambda<0}}
 (\lambda-\lambda_k)F(\lambda).
 \label{ic:lower-side-residue}
\end{align}

\subsubsection{Right pre-RH problem}

Define the right reflection coefficient by
\begin{equation}\label{ic:right-reflection-def}
 \mathfrak r(\lambda)
 :=
 \begin{cases}
  \dfrac{c_{21}(\lambda)}{c_{11}(\lambda)},
  &\lambda\in\dot\Sigma_1,\\[3mm]
  \dfrac{c_{21}(\lambda_+)}{c_{11}(\lambda_+)},
  &\lambda\in\dot\Sigma_0.
 \end{cases}
\end{equation}
It is continuous on each open component of
$\dot\Sigma_1\cup\dot\Sigma_0$ and satisfies
\begin{subequations}\label{ic:right-reflection-properties}
\begin{align}
 \mathfrak r(\lambda)
 &=
 \ord(\lambda^{-1}),
 &&\lambda\to\infty,
 \label{ic:right-reflection-infinity}\\
 \mathfrak r(\lambda)
 &=
 -\overline{\mathfrak r(-\lambda)},
 &&\lambda\in\dot\Sigma_1,
 \label{ic:right-reflection-outer-symmetry}\\
 \mathfrak r(\lambda)
 &=
 -\frac1{\mathfrak r(-\lambda)},
 \qquad
 |\mathfrak r(\lambda)|=1,
 &&\lambda\in\dot\Sigma_0.
 \label{ic:right-reflection-inner-symmetry}
\end{align}
\end{subequations}

\begin{remark}
If $\gamma_2\neq0$, the expansions at the inner branch points also
give
\begin{equation}\label{ic:right-reflection-inner-limits}
 \lim_{\substack{\lambda\to A_2^{-1}\\
                  \lambda\in\dot\Sigma_0}}
 \mathfrak r(\lambda)=-1,
 \qquad
 \lim_{\substack{\lambda\to-A_2^{-1}\\
                  \lambda\in\dot\Sigma_0}}
 \mathfrak r(\lambda)=1.
\end{equation}
\end{remark}

Define
\begin{equation}\label{ic:N-right-def}
 N(x,t,\lambda)
 :=
 \mathcal E
 \begin{cases}
 \left(
  \dfrac{(\Psi_1^{\mathrm b})^{(1)}(x,t,\lambda)}
        {c_{11}(\lambda)},
  (\Psi_2^{\mathrm b})^{(2)}(x,t,\lambda)
 \right)
 \eul^{f_2(x,t,\lambda)\sigma_3},
 &\lambda\in\mathbb C^+,\\[4mm]
 \left(
  (\Psi_2^{\mathrm b})^{(1)}(x,t,\lambda),
  \dfrac{(\Psi_1^{\mathrm b})^{(2)}(x,t,\lambda)}
        {c_{22}(\lambda)}
 \right)
 \eul^{f_2(x,t,\lambda)\sigma_3},
 &\lambda\in\mathbb C^-.
 \end{cases}
\end{equation}

\begin{proposition}\label{prop:ic-right-pre-rh}
The matrix $N$ has the following properties.

\begin{enumerate}[(i)]
\item It is analytic in $\mathbb C\setminus\mathbb R$ and has
one-sided simple-pole behavior at $\lambda_{k+}$ and
$\lambda_{k-}$, $k=1,\ldots,N$.

 \item Its boundary values satisfy 
 \begin{equation}\label{ic:N-right-jump}
  N_+(x,t,\lambda)
  =
  N_-(x,t,\lambda)G_{\mathrm R}(x,t,\lambda),
 \end{equation}
 where
 \begin{equation}\label{ic:G-right-def}
  G_{\mathrm R}(x,t,\lambda)
  =
  \begin{cases}
   \eul^{-f_2(x,t,\lambda)\sigma_3}
   G_{\mathrm R,0}(\lambda)
   \eul^{f_2(x,t,\lambda)\sigma_3},
   &\lambda\in\dot\Sigma_1\cup\dot\Sigma_0,\\[2mm]
   -\ii\sigma_1,
   &\lambda\in\dot\Gamma_2,
  \end{cases}
 \end{equation}
 with
 \begin{equation}\label{ic:G-right-zero-def}
  G_{\mathrm R,0}(\lambda)
  =
  \begin{cases}
   \begin{pmatrix}
    1-|\mathfrak r(\lambda)|^2&
    -\overline{\mathfrak r(\lambda)}\\
    \mathfrak r(\lambda)&1
   \end{pmatrix},
   &\lambda\in\dot\Sigma_1,\\[5mm]
   \begin{pmatrix}
    0&-\mathfrak r(\lambda)^{-1}\\
    \mathfrak r(\lambda)&1
   \end{pmatrix},
   &\lambda\in\dot\Sigma_0.
  \end{cases}
 \end{equation}

 \item The residue conditions are
\begin{subequations}\label{ic:N-right-residues}
\begin{align}
 \Res_{\lambda=\lambda_{k+}}N^{(1)}(x,t,\lambda)
 &=
 \varkappa_k
 \eul^{2f_2(x,t,\lambda_{k+})}
 N^{(2)}(x,t,\lambda_{k+}),
 \label{ic:N-right-residue-upper}\\
 \Res_{\lambda=\lambda_{k-}}N^{(2)}(x,t,\lambda)
 &=
 \varkappa_k
 \eul^{-2f_2(x,t,\lambda_{k-})}
 N^{(1)}(x,t,\lambda_{k-}),
 \label{ic:N-right-residue-lower}
\end{align}
\end{subequations}
for $k=1,\ldots,N$.

 \item As $\lambda\to\infty$,
 \begin{equation}\label{ic:N-right-normalization}
  N(x,t,\lambda)
  =
  I+\ord(\lambda^{-1}).
 \end{equation}

 \item At every
 $\lambda_0\in\{\pm A_1^{-1},\pm A_2^{-1}\}$,
 \begin{equation}\label{ic:N-right-branch-bound}
  N(x,t,\lambda)
  =
  \ord\bigl(|\lambda-\lambda_0|^{-1/4}\bigr).
 \end{equation}

 \item As $\lambda\to0$ nontangentially from either half-plane,
 \begin{equation}\label{ic:N-right-zero-bound}
  N(x,t,\lambda)=\ord(1).
 \end{equation}
\end{enumerate}
\end{proposition}

\begin{proof}
The analyticity of $N$ follows from
Proposition~\ref{prop:inner-cut-column-analyticity}. By
Proposition~\ref{prop:ic-simple-zeros}, the zeros of $c_{11}$ and,
by symmetry, those of $c_{22}$ are simple and therefore produce the
stated one-sided simple poles.

To establish the jump relation \eqref{ic:N-right-jump}, we first rewrite
the scattering relations in terms of the column matrices entering
\eqref{ic:N-right-def}. On $\dot\Sigma_1$, the
scattering relation \eqref{ic:scattering-relation}, together with
\eqref{ic:c-conjugation-entries} and
\eqref{ic:c-pseudo-unitarity}, yields
\begin{equation}\label{ic:right-column-jump-outer}
 \left(
  \frac{(\Psi_1^{\mathrm b})^{(1)}}{c_{11}},
  (\Psi_2^{\mathrm b})^{(2)}
 \right)
 =
 \left(
  (\Psi_2^{\mathrm b})^{(1)},
  \frac{(\Psi_1^{\mathrm b})^{(2)}}{c_{22}}
 \right)
 \begin{pmatrix}
  1-|\mathfrak r|^2&-\overline{\mathfrak r}\\
  \mathfrak r&1
 \end{pmatrix}.
\end{equation}
On $\dot\Sigma_0$, the continued column scattering relations
\eqref{ic:continued-column-scattering} and the coefficient relations
\eqref{ic:c-cut-relations} imply
\begin{equation}\label{ic:right-column-jump-band}
 \left(
  \frac{(\Psi_{1,+}^{\mathrm b})^{(1)}}{c_{11,+}},
  (\Psi_2^{\mathrm b})^{(2)}
 \right)
 =
 \left(
  (\Psi_2^{\mathrm b})^{(1)},
  \frac{(\Psi_{1,-}^{\mathrm b})^{(2)}}{c_{22,-}}
 \right)
 \begin{pmatrix}
  0&-\mathfrak r^{-1}\\
  \mathfrak r&1
 \end{pmatrix}.
\end{equation}
Substituting \eqref{ic:right-column-jump-outer} and
\eqref{ic:right-column-jump-band} into
\eqref{ic:N-right-def} and collecting the diagonal phase factors, the jump relations \eqref{ic:N-right-jump} on
$\dot\Sigma_1$ and $\dot\Sigma_0$, respectively.

On $\dot\Gamma_2\setminus\{0\}$,
\eqref{ic:Psi-opposite-side-cut} and
\eqref{ic:c11-c22-cut} imply
\[
 \left(
  \frac{(\Psi_{1,+}^{\mathrm b})^{(1)}}{c_{11,+}},
  (\Psi_{2,+}^{\mathrm b})^{(2)}
 \right)
 =
 \left(
  (\Psi_{2,-}^{\mathrm b})^{(1)},
  \frac{(\Psi_{1,-}^{\mathrm b})^{(2)}}{c_{22,-}}
 \right)(-\ii\sigma_1).
\]
Using \eqref{ic:N-right-def},
$f_{2,-}=-f_{2,+}$, and
$\sigma_1\sigma_3=-\sigma_3\sigma_1$, we obtain
$ N_+=N_-(-\ii\sigma_1)$.
By Proposition~\ref{prop:ic-Psi-zero-expansion}, this relation extends
to $\lambda=0$ by continuity. Together with the relations on
$\dot\Sigma_1$ and $\dot\Sigma_0$, this proves
\eqref{ic:N-right-jump}.

For the residue conditions, \eqref{ic:eigenfunction-linear-dependence},
together with $\Psi_j^{\mathrm b}=G\Psi_j$, yields
\begin{equation}\label{ic:balanced-eigenfunction-dependence-upper}
 (\Psi_1^{\mathrm b})^{(1)}(\lambda_{k+})
 =
 \mathfrak b_k
 (\Psi_2^{\mathrm b})^{(2)}(\lambda_{k+}).
\end{equation}
The relations \eqref{ic:Psi-opposite-side-cut} then imply
\begin{equation}\label{ic:balanced-eigenfunction-dependence-lower}
 (\Psi_1^{\mathrm b})^{(2)}(\lambda_{k-})
 =
 \mathfrak b_k
 (\Psi_2^{\mathrm b})^{(1)}(\lambda_{k-}).
\end{equation}
Combining \eqref{ic:balanced-eigenfunction-dependence-upper} and
\eqref{ic:balanced-eigenfunction-dependence-lower} with
\eqref{ic:norming-constants-def} yields the residue conditions
\eqref{ic:N-right-residues}.

By \eqref{ic:Psi-tildePsi-relation} and
\eqref{ic:tildePsi-large-lambda},
\begin{align*}
 \left(
  (\Psi_1^{\mathrm b})^{(1)}\eul^{f_1},
  (\Psi_2^{\mathrm b})^{(2)}\eul^{-f_2}
 \right)
 &=
\mathcal E^{-1}+\ord(\lambda^{-1}),
 &&\lambda\to\infty,\quad\lambda\in\mathbb C^+,\\
 \left(
  (\Psi_2^{\mathrm b})^{(1)}\eul^{f_2},
  (\Psi_1^{\mathrm b})^{(2)}\eul^{-f_1}
 \right)
 &=
\mathcal E^{-1}+\ord(\lambda^{-1}),
 &&\lambda\to\infty,\quad\lambda\in\mathbb C^-.
\end{align*}
Together with
\begin{align*}
 c_{11}(\lambda)
 &=
 \eul^{f_2-f_1}
 \bigl(1+\ord(\lambda^{-1})\bigr),
 &&\lambda\to\infty,\quad\lambda\in\mathbb C^+,\\
 c_{22}(\lambda)
 &=
 \eul^{f_1-f_2}
 \bigl(1+\ord(\lambda^{-1})\bigr),
 &&\lambda\to\infty,\quad\lambda\in\mathbb C^-,
\end{align*}
these expansions imply the normalization at infinity \eqref{ic:N-right-normalization}.

The branch-point estimates follow from the expansions obtained in
Section~\ref{sec:ic-branch-point-behavior}.

Finally, the column expansions
\eqref{ic:Psi-upper-zero-expansion} and
\eqref{ic:Psi-lower-zero-expansion}, together with the determinant
representations \eqref{ic:scattering-determinants}, yield
\begin{align*}
 c_{11}(\lambda)
 &=
 \eul^{-D_1-D_2+f_2-f_1}
 \bigl(1+\ord(\lambda^2)\bigr),
 &&\lambda\to0,\quad\lambda\in\mathbb C^+,\\
 c_{22}(\lambda)
 &=
 \eul^{-D_1-D_2+f_1-f_2}
 \bigl(1+\ord(\lambda^2)\bigr),
 &&\lambda\to0,\quad\lambda\in\mathbb C^-.
\end{align*}
Substituting the preceding asymptotics for $c_{11}$ and $c_{22}$,
together with the Jost-column expansions
\eqref{ic:Psi-upper-zero-expansion} and
\eqref{ic:Psi-lower-zero-expansion}, into
\eqref{ic:N-right-def}, and using
$\Psi_j^{\mathrm b}=G\Psi_j$, we obtain
\begin{align*}
 N_+(\lambda)
 &=
 \ii\mathcal E
 \begin{pmatrix}
  \dfrac{\lambda}{2}\eul^{D_2-2r}p&\eul^{-D_2}\\[1mm]
  \eul^{D_2}&\dfrac{\lambda}{2}\eul^{-D_2+2r}q
 \end{pmatrix}
 +\ord(\lambda^2),
 &&\lambda\to0,\quad\lambda\in\mathbb C^+,\\
 N_-(\lambda)
 &=
 -\mathcal E
 \begin{pmatrix}
  \eul^{-D_2}&\dfrac{\lambda}{2}\eul^{D_2-2r}p\\[1mm]
  \dfrac{\lambda}{2}\eul^{-D_2+2r}q&\eul^{D_2}
 \end{pmatrix}
 +\ord(\lambda^2),
 &&\lambda\to0,\quad\lambda\in\mathbb C^-.
\end{align*}
In particular, $N$ is bounded at $\lambda=0$, which proves
\eqref{ic:N-right-zero-bound}.
\end{proof}

\subsubsection{Left pre-RH problem}

Define the left reflection coefficient by
\begin{equation}\label{ic:left-reflection-def}
 \widetilde{\mathfrak r}(\lambda)
 :=
 \begin{cases}
  \dfrac{\overline{c_{21}(\lambda)}}{c_{11}(\lambda)},
  &\lambda\in\dot\Sigma_1,\\[3mm]
  \dfrac1{
   c_{21}(\lambda_+)c_{11}(\lambda_+)}
  =
  \dfrac{\ii}{|c_{11}(\lambda_+)|^2},
  &\lambda\in\dot\Sigma_0.
 \end{cases}
\end{equation}

The function $\widetilde{\mathfrak r}$ is continuous on each open
component of $\dot\Sigma_1\cup\dot\Sigma_0$ and satisfies
\begin{subequations}\label{ic:left-reflection-properties}
\begin{align}
 \widetilde{\mathfrak r}(\lambda)
 &=
 \ord(\lambda^{-1}),
 &&\lambda\to\infty,
 \label{ic:left-reflection-infinity}\\
 \widetilde{\mathfrak r}(\lambda)
 &=
 -\overline{\widetilde{\mathfrak r}(-\lambda)},
 &&\lambda\in\dot\Sigma_1,
 \label{ic:left-reflection-outer-symmetry}\\
 \widetilde{\mathfrak r}(\lambda)
 &=
 \widetilde{\mathfrak r}(-\lambda),
 &&\lambda\in\dot\Sigma_0.
 \label{ic:left-reflection-inner-symmetry}
\end{align}
\end{subequations}

\begin{remark}
The limiting behavior of the left reflection coefficient at the
endpoints of $\Sigma_0$ is given by
\begin{align}
 \lim_{\substack{\lambda\to A_1^{-1}\\
                  \lambda\in\dot\Sigma_1}}
 \widetilde{\mathfrak r}(\lambda)
 &=1,
 &
 \lim_{\substack{\lambda\to-A_1^{-1}\\
                  \lambda\in\dot\Sigma_1}}
 \widetilde{\mathfrak r}(\lambda)
 &=-1,
 \label{ic:left-reflection-A1-outer-limits}\\
 \lim_{\substack{\lambda\to\pm A_1^{-1}\\
                  \lambda\in\dot\Sigma_0}}
 \widetilde{\mathfrak r}(\lambda)
 &=0.
 \label{ic:left-reflection-A1-inner-limits}
\end{align}
Near $\pm A_2^{-1}$, its behavior depends on $\gamma_2$:
\begin{equation}\label{ic:left-reflection-A2-behavior}
 \widetilde{\mathfrak r}(\lambda)
 =
 \begin{cases}
  \ord\bigl(|\lambda\mp A_2^{-1}|^{1/2}\bigr),
  &\gamma_2\neq0,\\[1mm]
  \ord\bigl(|\lambda\mp A_2^{-1}|^{-1/2}\bigr),
  &\gamma_2=0,
 \end{cases}
 \qquad
 \lambda\to\pm A_2^{-1},\quad
 \lambda\in\dot\Sigma_0.
\end{equation}
\end{remark}

Define
\begin{equation}\label{ic:N-left-def}
 \widetilde N(x,t,\lambda)
 :=
 \mathcal E
 \begin{cases}
 \left(
  (\Psi_1^{\mathrm b})^{(1)}(x,t,\lambda),
  \dfrac{(\Psi_2^{\mathrm b})^{(2)}(x,t,\lambda)}
        {c_{11}(\lambda)}
 \right)
 \eul^{f_1(x,t,\lambda)\sigma_3},
 &\lambda\in\mathbb C^+,\\[4mm]
 \left(
  \dfrac{(\Psi_2^{\mathrm b})^{(1)}(x,t,\lambda)}
        {c_{22}(\lambda)},
  (\Psi_1^{\mathrm b})^{(2)}(x,t,\lambda)
 \right)
 \eul^{f_1(x,t,\lambda)\sigma_3},
 &\lambda\in\mathbb C^-.
 \end{cases}
\end{equation}

\begin{proposition}\label{prop:ic-left-pre-rh}
The matrix $\widetilde N$ has the following properties.

\begin{enumerate}[(i)]
 \item It is analytic in $\mathbb C\setminus\mathbb R$ and has
one-sided simple-pole behavior at $\lambda_{k+}$ and
$\lambda_{k-}$, $k=1,\ldots,N$.

 \item Its boundary values satisfy
 \begin{equation}\label{ic:N-left-jump}
  \widetilde N_+(x,t,\lambda)
  =
  \widetilde N_-(x,t,\lambda)
  G_{\mathrm L}(x,t,\lambda),
 \end{equation}
 where
 \begin{equation}\label{ic:G-left-def}
  G_{\mathrm L}(x,t,\lambda)
  =
  \begin{cases}
   \eul^{-f_1(x,t,\lambda)\sigma_3}
   \begin{pmatrix}
    1&-\widetilde{\mathfrak r}(\lambda)\\
    \overline{\widetilde{\mathfrak r}(\lambda)}
    &1-|\widetilde{\mathfrak r}(\lambda)|^2
   \end{pmatrix}
   \eul^{f_1(x,t,\lambda)\sigma_3},
   &\lambda\in\dot\Sigma_1,\\[6mm]
   -\ii
   \eul^{-f_1(x,t,\lambda_-)\sigma_3}
   \begin{pmatrix}
    0&1\\
    1&\widetilde{\mathfrak r}(\lambda)
   \end{pmatrix}
   \eul^{f_1(x,t,\lambda_+)\sigma_3},
   &\lambda\in\dot\Sigma_0,\\[6mm]
   -\ii\sigma_1,
   &\lambda\in\dot\Gamma_2.
  \end{cases}
 \end{equation}

 \item The residue conditions are
 \begin{subequations}\label{ic:N-left-residues}
\begin{align}
 \Res_{\lambda=\lambda_{k+}}
 \widetilde N^{(2)}(x,t,\lambda)
 &=
 \widetilde\varkappa_k
 \eul^{-2f_1(x,t,\lambda_{k+})}
 \widetilde N^{(1)}(x,t,\lambda_{k+}),
 \label{ic:N-left-residue-upper}\\
 \Res_{\lambda=\lambda_{k-}}
 \widetilde N^{(1)}(x,t,\lambda)
 &=
 \widetilde\varkappa_k
 \eul^{2f_1(x,t,\lambda_{k-})}
 \widetilde N^{(2)}(x,t,\lambda_{k-}),
 \label{ic:N-left-residue-lower}
\end{align}
\end{subequations}

 \item As $\lambda\to\infty$,
 \begin{equation}\label{ic:N-left-normalization}
  \widetilde N(x,t,\lambda)
  =
  I+\ord(\lambda^{-1}).
 \end{equation}

 \item At $\lambda_0=\pm A_1^{-1}$,
\begin{equation}\label{ic:N-left-A1-branch-bound}
 \widetilde N(x,t,\lambda)
 =
 \ord\bigl(|\lambda-\lambda_0|^{-1/4}\bigr).
\end{equation}
At $\lambda_0=\pm A_2^{-1}$, the behavior depends on $\gamma_2$.
If $\gamma_2\neq0$ (generic case), then
\begin{equation}\label{ic:N-left-A2-regular-bound}
 \widetilde N(x,t,\lambda)=\ord(1),
 \qquad
 \lambda\to\lambda_0.
\end{equation}
If $\gamma_2=0$, then, for each
$\lambda_0\in\{\pm A_2^{-1}\}$, there exist scalar functions
$\tau_{\lambda_0,\pm}(x,t)$ such that
\begin{subequations}\label{ic:N-left-degenerate-behavior}
\begin{align}
 \widetilde N_+^{(1)}(x,t,\lambda)
 &=\ord(1),&
 \widetilde N_+^{(2)}(x,t,\lambda)
 -\frac{\tau_{\lambda_0,+}(x,t)}
        {l_2(\lambda_+)}
  \widetilde N_+^{(1)}(x,t,\lambda)
 &=\ord(1),
 \label{ic:N-left-degenerate-upper}\\
 \widetilde N_-^{(2)}(x,t,\lambda)
 &=\ord(1),&
 \widetilde N_-^{(1)}(x,t,\lambda)
 -\frac{\tau_{\lambda_0,-}(x,t)}
        {l_2(\lambda_-)}
  \widetilde N_-^{(2)}(x,t,\lambda)
 &=\ord(1),
 \label{ic:N-left-degenerate-lower}
\end{align}
\end{subequations}
as $\lambda\to\lambda_0$ along $\dot\Sigma_0$. In particular,
\begin{equation}\label{ic:N-left-degenerate-bound}
 \widetilde N(x,t,\lambda)
 =
 \ord\bigl(|\lambda-\lambda_0|^{-1/2}\bigr).
\end{equation}

 \item As $\lambda\to0$ nontangentially from either half-plane,
 \begin{equation}\label{ic:N-left-zero-bound}
  \widetilde N(x,t,\lambda)=\ord(1).
 \end{equation}
\end{enumerate}
\end{proposition}

\begin{proof}
The proof is analogous to that of
Proposition~\ref{prop:ic-right-pre-rh}. The jump relation follows from
\eqref{ic:scattering-relation},
\eqref{ic:continued-column-scattering}, and
\eqref{ic:Psi-opposite-side-cut}, while
\eqref{ic:eigenfunction-linear-dependence} and
\eqref{ic:norming-constants-def} give the residue conditions. The
large-$\lambda$ and zero assertions follow from the same expansions
used there.
Finally, Corollaries~\ref{cor:ic-Jost-fourth-root} and
\ref{cor:ic-c-branch-expansion}, together with
\eqref{ic:N-left-def}, imply
\eqref{ic:N-left-A1-branch-bound}. At
$\lambda_0=\pm A_2^{-1}$, if $\gamma_2\neq0$, the relevant columns of
$\Psi_2^{\mathrm b}$ and the coefficients $c_{11}$ and $c_{22}$ have
the same fourth-root singularity. Their quotients in
\eqref{ic:N-left-def} are therefore bounded, which proves
\eqref{ic:N-left-A2-regular-bound}.

It remains to consider the non-generic case with $\gamma_2=0$. We first take
$\lambda_0=A_2^{-1}$ and write
$\zeta=\zeta_2(\lambda_+)$. Expanding the Jost solutions one term further as
$\lambda\to A_2^{-1}$, $\lambda\in \dot\Sigma_0$, we obtain
\[
 \Psi_{2,-}^{(1)}
 =
 \frac{w_2}{\zeta}
 +\zeta v_1+\ord(\zeta^3),
 \qquad
 \Psi_{2,+}^{(2)}
 =
 \frac{w_2}{\zeta}
 +\zeta v_2+\ord(\zeta^3),
\]
with some vectors $v_1=v_1(x,t)$ and $v_2=v_2(x,t)$. 
Since $\det\bigl(\Psi_{2,-}^{(1)},\Psi_{2,+}^{(2)}\bigr)=1$, it follows that $w_2\neq0$ and 
\[
 1
 =
\det\bigl(\Psi_{2,-}^{(1)},\Psi_{2,+}^{(2)}\bigr)
 =
 \det(w_2,v_2)+\det(v_1,w_2)+\ord(\zeta^2).
\]
Hence
\begin{equation}\label{ic:degenerate-determinant-balance}
 \det(w_2,v_2)+\det(v_1,w_2)=1.
\end{equation}
Set
\[
 z_1:=\Psi_{1,+}^{(1)}(A_2^{-1}).
\]
Since $A_2^{-1}$ is not a branch point of $\Psi_1$, the column
\[
 y(x):=\Psi_{1,+}^{(1)}(x,t,A_2^{-1})
\]
is a nontrivial solution of the spatial Lax equation. Its propagator is
invertible; hence $y(x_0)=0$ at any finite $x_0$ would imply
$y\equiv0$, contradicting the Jost normalization. Therefore
$z_1\neq0$.

Definition of $\gamma_2$ in
Corollary~\ref{cor:ic-c-branch-expansion} together with $\gamma_2=0$ imply that
 $\det(z_1,w_2)=0$. Thus
\[
 z_1=\vartheta w_2
\]
for some $\vartheta\neq0$. Moreover, since $A_2^{-1}$ is not a branch
point of $\Psi_1$,
\[
 \Psi_{1,+}^{(1)}(\lambda)
 =
 z_1+\ord(\lambda-A_2^{-1})
 =
 z_1+\ord(\zeta^4).
\]
The determinant representations
\eqref{ic:scattering-determinants} therefore yield
\[
 c_{11}(\lambda_+)
 =
 \kappa_{11}\zeta+\ord(\zeta^3),
 \qquad
 c_{21}(\lambda_+)
 =
 \kappa_{21}\zeta+\ord(\zeta^3),
\]
where
\[
 \kappa_{11}
 :=
 \vartheta\det(w_2,v_2),
 \qquad
 \kappa_{21}
 :=
 \vartheta\det(v_1,w_2).
\]
The symmetry relations
\eqref{ic:c11-c22-conjugation-cut} and
\eqref{ic:c22-c21-band} imply
\[
 \overline{c_{11}(\lambda_+)}
 =
 \ii c_{21}(\lambda_+).
\]
In particular, $|\kappa_{11}|=|\kappa_{21}|$. Together with
\eqref{ic:degenerate-determinant-balance}, this yields
\[
 \kappa_{11}\neq0,
 \qquad
 \kappa_{21}\neq0.
\]

Let
\[
 U:=\mathcal E G w_2,
 \qquad
 F_+:=f_1(x,t,\lambda_{0+}).
\]
Then $U\neq0$, and \eqref{ic:N-left-def} gives
\[
 \widetilde N_+^{(1)}
 =
 \vartheta\eul^{F_+}U+\ord(\zeta^4),
 \qquad
 \widetilde N_+^{(2)}
 =
 \frac{\eul^{-F_+}}
      {\kappa_{11}\zeta^2}U+\ord(1).
\]
Since
\[
 \zeta^2=2\ii A_2l_2(\lambda_+),
\]
we may define
\[
 \tau_{\lambda_0,+}(x,t)
 :=
 \frac{\eul^{-2F_+}}
      {2\ii A_2\vartheta\kappa_{11}}.
\]
This proves \eqref{ic:N-left-degenerate-upper}. The lower-side formula
\eqref{ic:N-left-degenerate-lower} follows from the 
symmetries \eqref{ic:Psi-conjugation-columns} and \eqref{ic:c11-c22-conjugation-cut}. The result at $\lambda_0=-A_2^{-1}$ follows from the symmetries\eqref{ic:Psi-minus-upper-cut},
\eqref{ic:Psi-minus-lower-cut},
\eqref{ic:c11-minus-upper-cut},
\eqref{ic:c22-minus-lower-cut}.

Finally,
\[
 l_2(\lambda_\pm)
 =
 \ord\bigl(|\lambda-\lambda_0|^{1/2}\bigr),
\]
which proves \eqref{ic:N-left-degenerate-bound}.
\end{proof}

Now we observe that the two pre-RH matrices differ by a diagonal factor.

\begin{proposition}\label{prop:ic-pre-rh-connection}
For $\lambda\in\mathbb C\setminus\mathbb R$,
\begin{equation}\label{ic:pre-rh-connection}
 N(x,t,\lambda)
 =
 \widetilde N(x,t,\lambda)
 P_{\mathrm{LR}}(x,t,\lambda),
\end{equation}
where
\begin{equation}\label{ic:pre-rh-connection-factor}
 P_{\mathrm{LR}}(x,t,\lambda)
 =
 \begin{cases}
 \begin{pmatrix}
 \displaystyle
 \frac{\eul^{f_2(x,t,\lambda)-f_1(x,t,\lambda)}}
      {c_{11}(\lambda)}
 &0\\[4mm]
 0&
 \displaystyle
 c_{11}(\lambda)
 \eul^{f_1(x,t,\lambda)-f_2(x,t,\lambda)}
 \end{pmatrix},
 &\lambda\in\mathbb C^+,\\[10mm]
 \begin{pmatrix}
 \displaystyle
 c_{22}(\lambda)
 \eul^{f_2(x,t,\lambda)-f_1(x,t,\lambda)}
 &0\\[4mm]
 0&
 \displaystyle
 \frac{\eul^{f_1(x,t,\lambda)-f_2(x,t,\lambda)}}
      {c_{22}(\lambda)}
 \end{pmatrix},
 &\lambda\in\mathbb C^-.
 \end{cases}
\end{equation}
Moreover,
\begin{equation}\label{ic:pre-rh-connection-properties}
 \det P_{\mathrm{LR}}(x,t,\lambda)=1,
 \qquad
 P_{\mathrm{LR}}(x,t,\lambda)
 =
 I+\ord(\lambda^{-1}),
 \qquad
 \lambda\to\infty
\end{equation}
in either half-plane.
\end{proposition}

\begin{proof}
It follows from \eqref{ic:N-right-def} and
\eqref{ic:N-left-def} that, for $\lambda\in\mathbb C^+$,
\[
 N^{(1)}
 =
 \frac{\eul^{f_2-f_1}}{c_{11}}\,
 \widetilde N^{(1)},
 \qquad
 N^{(2)}
 =
 c_{11}\eul^{f_1-f_2}\,
 \widetilde N^{(2)}.
\]
Similarly, for $\lambda\in\mathbb C^-$, the corresponding relations are
\[
 N^{(1)}
 =
 c_{22}\eul^{f_2-f_1}\,
 \widetilde N^{(1)},
 \qquad
 N^{(2)}
 =
 \frac{\eul^{f_1-f_2}}{c_{22}}\,
 \widetilde N^{(2)}.
\]
These column relations are precisely
\eqref{ic:pre-rh-connection}--\eqref{ic:pre-rh-connection-factor}.

A direct computation using
\eqref{ic:pre-rh-connection-factor} yields
$\det P_{\mathrm{LR}}(x,t,\lambda)=1$
. Furthermore,
\eqref{ic:scattering-determinants},
\eqref{ic:Psi-tildePsi-relation}, and
\eqref{ic:tildePsi-large-lambda} imply
\begin{align*}
 c_{11}(\lambda)\eul^{f_1-f_2}
 &=
 1+\ord(\lambda^{-1}),
 &&\lambda\to\infty,\quad\lambda\in\mathbb C^+,\\
 c_{22}(\lambda)\eul^{f_2-f_1}
 &=
 1+\ord(\lambda^{-1}),
 &&\lambda\to\infty,\quad\lambda\in\mathbb C^-.
\end{align*}
Consequently,
$ P_{\mathrm{LR}}(x,t,\lambda)
 =
 I+\ord(\lambda^{-1}),
 \qquad
 \lambda\to\infty.$
\end{proof}

\subsection{The right RH problem in the
$\boldsymbol{(y,t)}$ scale}
\label{sec:ic-right-rh}

The phase function $f_2$ contains $\rho$, so the resulting RH
problem would involve the solution that it is intended to reconstruct.
To obtain a closed formulation in terms of the scattering data, we
introduce the new spatial variable
\begin{equation}\label{ic:right-hodograph-def}
 y(x,t)
 :=
 x-\frac1{A_2}
 \int_x^{+\infty}
 \bigl(\rho(\xi,t)-A_2\bigr)\dd\xi
 -A_2^2t.
\end{equation}
The conservation law for $\rho$ in
\eqref{ic:rho-r-conservation} yields
\begin{equation}\label{ic:right-hodograph-derivatives}
 y_x=\frac{\rho}{A_2},
 \qquad
 y_t=-\frac{Q\rho}{A_2}=-Qy_x.
\end{equation}
In particular, $y_x>0$, and, for all fixed $t$, the map $x\mapsto y(x,t)$ is a strictly
increasing diffeomorphism of $\mathbb R$ onto $\mathbb R$.

Then
\begin{equation}\label{ic:right-phase-hodograph}
 f_2(x,t,\lambda)
 =
 \widehat f_2(y(x,t),t,\lambda),
\end{equation}
where
\begin{equation}\label{ic:right-phase-def}
 \widehat f_2(y,t,\lambda)
 :=
 \frac{\ii A_2l_2(\lambda)}2
 \left(
  y-\frac{2t}{\lambda^2}
 \right).
\end{equation}
Introduce $\widehat N$ by
\begin{equation}\label{ic:Nhat-right-def}
 N(x,t,\lambda)
 =
 \widehat N(y(x,t),t,\lambda).
\end{equation}

The \textbf{right RH problem} is to find a $2\times2$ matrix
$\widehat N(y,t,\lambda)$ with the following properties.

\begin{enumerate}[(i)]
 \item $\widehat N$ is analytic in
 $\mathbb C\setminus\mathbb R$,  has
one-sided simple-pole behavior at $\lambda_{k+}$ and
$\lambda_{k-}$, $k=1,\ldots,N$, and extends continuously from each half-plane to the real axis away from $\pm A_1^{-1}$, $\pm A_2^{-1}$, and $\lambda_{k}$, $k=1,\ldots,N$.

 \item Its boundary values satisfy
 \begin{equation}\label{ic:Nhat-right-jump}
  \widehat N_+(y,t,\lambda)
  =
  \widehat N_-(y,t,\lambda)
  \widehat G_{\mathrm R}(y,t,\lambda),
 \end{equation}
 where
 \begin{equation}\label{ic:Ghat-right-def}
  \widehat G_{\mathrm R}(y,t,\lambda)
  =
  \begin{cases}
   \eul^{-\widehat f_2(y,t,\lambda)\sigma_3}
   G_{\mathrm R,0}(\lambda)
   \eul^{\widehat f_2(y,t,\lambda)\sigma_3},
   &\lambda\in\dot\Sigma_1\cup\dot\Sigma_0,\\[2mm]
   -\ii\sigma_1,
   &\lambda\in\dot\Gamma_2,
  \end{cases}
 \end{equation}
 and $G_{\mathrm R,0}$ is defined by
 \eqref{ic:G-right-zero-def}.

 \item Its one-sided residues satisfy
\begin{subequations}\label{ic:Nhat-right-residues}
\begin{align}
 \Res_{\lambda=\lambda_{k+}}
 \widehat N^{(1)}(y,t,\lambda)
 &=
 \varkappa_k
 \eul^{2\widehat f_2(y,t,\lambda_{k+})}
 \widehat N^{(2)}(y,t,\lambda_{k+}),
 \label{ic:Nhat-right-residue-upper}\\
 \Res_{\lambda=\lambda_{k-}}
 \widehat N^{(2)}(y,t,\lambda)
 &=
 \varkappa_k
 \eul^{-2\widehat f_2(y,t,\lambda_{k-})}
 \widehat N^{(1)}(y,t,\lambda_{k-}).
 \label{ic:Nhat-right-residue-lower}
\end{align}
\end{subequations}

 \item As $\lambda\to\infty$,
 \begin{equation}\label{ic:Nhat-right-normalization}
  \widehat N(y,t,\lambda)
  =
  I+\ord(\lambda^{-1}).
 \end{equation}

 \item At every
 $\lambda_0\in\{\pm A_1^{-1},\pm A_2^{-1}\}$,
 \begin{equation}\label{ic:Nhat-right-branch-bound}
  \widehat N(y,t,\lambda)
  =
  \ord\bigl(|\lambda-\lambda_0|^{-1/4}\bigr).
 \end{equation}

\end{enumerate}

\begin{proposition}\label{prop:ic-right-rh-properties}
The right RH problem has at most one solution. If a solution exists,
then
\begin{equation}\label{ic:Nhat-right-determinant}
 \det\widehat N(y,t,\lambda)\equiv1.
\end{equation}
For scattering data arising from the direct problem, it satisfies
\begin{subequations}\label{ic:Nhat-right-symmetries}
\begin{align}
 \widehat N(\lambda)
 &=
 \sigma_2\widehat N(-\lambda)\sigma_2,
 &
 \widehat N(\lambda)
 &=
 \sigma_1
 \overline{\widehat N(\overline\lambda)}
 \sigma_1,
 &&\lambda\in\mathbb C^+,
 \label{ic:Nhat-right-off-contour-symmetries}\\
 \widehat N(\lambda_+)
 &=
 -\sigma_2
 \widehat N((-\lambda)_+)\sigma_3,
 &
 \widehat N(\lambda_+)
 &=
 -\ii\sigma_1
 \overline{\widehat N(\lambda_+)},
 &&\lambda\in
 \dot\Gamma_2\setminus\{\lambda_k\}_{k=1}^{N},
 \label{ic:Nhat-right-upper-symmetries}\\
 \widehat N(\lambda_-)
 &=
 \sigma_2
 \widehat N((-\lambda)_-)\sigma_3,
 &
 \widehat N(\lambda_-)
 &=
 \ii\sigma_1
 \overline{\widehat N(\lambda_-)},
 &&\lambda\in
 \dot\Gamma_2\setminus
 \{\lambda_k\}_{k=1}^{N}.
 \label{ic:Nhat-right-lower-symmetries}
\end{align}
\end{subequations}
\end{proposition}

\begin{proof}
The jump matrices have determinant one, and the residue conditions
make the possible singularities of $\det\widehat N$ at
$\lambda_{k\pm}$ removable. Hence $\det\widehat N$ has no jump and
extends across the regular parts of the contour. At each branch point
$\lambda_0$,
\[
 \det\widehat N(\lambda)
 =
 \ord\bigl(|\lambda-\lambda_0|^{-1/2}\bigr),
\]
so the remaining singularities are also removable. Thus
$\det\widehat N$ extends to an entire function. The normalization
\eqref{ic:Nhat-right-normalization} and Liouville's theorem yield
\[
 \det\widehat N\equiv1.
\]

If $\widehat N_1$ and $\widehat N_2$ are two solutions, then
$\widehat N_1\widehat N_2^{-1}$ has no jump or poles and extends
across the regular parts of the contour. Its singularities at the
branch points are removable by
\eqref{ic:Nhat-right-branch-bound}. Since
\[
 \widehat N_1\widehat N_2^{-1}
 =
 I+\ord(\lambda^{-1}),
 \qquad \lambda\to\infty,
\]
Liouville's theorem gives
$\widehat N_1\widehat N_2^{-1}\equiv I$, proving uniqueness.

For scattering data arising from the direct problem,
\eqref{ic:l-symmetries}, \eqref{ic:right-phase-def},
\eqref{ic:right-reflection-properties},
\eqref{ic:discrete-spectrum-parity},
\eqref{ic:norming-constants-integral}, and
\eqref{ic:norming-parity} imply that
\[
 \sigma_2\widehat N(-\lambda)\sigma_2,
 \qquad
 \sigma_1\overline{\widehat N(\overline\lambda)}\sigma_1
\]
satisfy the same RH problem as $\widehat N$. Uniqueness proves
\eqref{ic:Nhat-right-off-contour-symmetries}. Taking the corresponding
boundary values on $\dot\Gamma_2$ and using
\eqref{ic:Nhat-right-jump} gives
\eqref{ic:Nhat-right-upper-symmetries} and
\eqref{ic:Nhat-right-lower-symmetries}.
\end{proof}

\subsection{The left RH problem in the
$\boldsymbol{(\widetilde y,t)}$ scale}
\label{sec:ic-left-rh}

Similarly, for the left RH problem, introduce
\begin{equation}\label{ic:left-hodograph-def}
 \widetilde y(x,t)
 :=
 x+\frac1{A_1}
 \int_{-\infty}^{x}
 \bigl(\rho(\xi,t)-A_1\bigr)\dd\xi
 -A_1^2t.
\end{equation}
Then the conservation law for $\rho$ in
\eqref{ic:rho-r-conservation} yields
\begin{equation}\label{ic:left-hodograph-derivatives}
 \widetilde y_x=\frac{\rho}{A_1}>0,
 \qquad
 \widetilde y_t=-\frac{Q\rho}{A_1}
 =-Q\widetilde y_x.
\end{equation}
Hence, for all fixed $t$, $x\mapsto\widetilde y(x,t)$ is a strictly increasing
diffeomorphism of $\mathbb R$ onto $\mathbb R$.

Then
\begin{equation}\label{ic:left-phase-hodograph}
 f_1(x,t,\lambda)
 =
 \widehat f_1(\widetilde y(x,t),t,\lambda),
\end{equation}
where
\begin{equation}\label{ic:left-phase-def}
 \widehat f_1(\widetilde y,t,\lambda)
 :=
 \frac{\ii A_1l_1(\lambda)}2
 \left(
  \widetilde y-\frac{2t}{\lambda^2}
 \right).
\end{equation}
Introduce $\widehat{\widetilde N}$ by
\begin{equation}\label{ic:Nhat-left-def}
 \widetilde N(x,t,\lambda)
 =
 \widehat{\widetilde N}
 \bigl(\widetilde y(x,t),t,\lambda\bigr).
\end{equation}

The \textbf{left RH problem} is to find a $2\times2$ matrix
$\widehat{\widetilde N}(\widetilde y,t,\lambda)$ with the following
properties.

\begin{enumerate}[(i)]
 \item $\widehat{\widetilde N}$ is analytic in
 $\mathbb C\setminus\mathbb R$,  has
one-sided simple-pole behavior at $\lambda_{k+}$ and
$\lambda_{k-}$, $k=1,\ldots,N$, and extends continuously from each half-plane to the real axis away from $\pm A_1^{-1}$, $\pm A_2^{-1}$, and $\lambda_{k}$, $k=1,\ldots,N$.

 \item Its boundary values satisfy
 \begin{equation}\label{ic:Nhat-left-jump}
  \widehat{\widetilde N}_+
  (\widetilde y,t,\lambda)
  =
  \widehat{\widetilde N}_-
  (\widetilde y,t,\lambda)
  \widehat G_{\mathrm L}(\widetilde y,t,\lambda),
 \end{equation}
 where
 \begin{equation}\label{ic:Ghat-left-def}
  \widehat G_{\mathrm L}(\widetilde y,t,\lambda)
  =
  \begin{cases}
   \eul^{-\widehat f_1(\widetilde y,t,\lambda)\sigma_3}
   \begin{pmatrix}
    1&-\widetilde{\mathfrak r}(\lambda)\\
    \overline{\widetilde{\mathfrak r}(\lambda)}
    &1-|\widetilde{\mathfrak r}(\lambda)|^2
   \end{pmatrix}
   \eul^{\widehat f_1(\widetilde y,t,\lambda)\sigma_3},
   &\lambda\in\dot\Sigma_1,\\[6mm]
   -\ii
   \eul^{-\widehat f_1(
                 \widetilde y,t,\lambda_-)\sigma_3}
   \begin{pmatrix}
    0&1\\
    1&\widetilde{\mathfrak r}(\lambda)
   \end{pmatrix}
   \eul^{\widehat f_1(
                 \widetilde y,t,\lambda_+)\sigma_3},
   &\lambda\in\dot\Sigma_0,\\[6mm]
   -\ii\sigma_1,
   &\lambda\in\dot\Gamma_2.
  \end{cases}
 \end{equation}

 \item Its one-sided residues satisfy
 \begin{subequations}\label{ic:Nhat-left-residues}
\begin{align}
 \Res_{\lambda=\lambda_{k+}}
 \widehat{\widetilde N}^{(2)}(\widetilde y,t,\lambda)
 &=
 \widetilde\varkappa_k
 \eul^{-2\widehat f_1(\widetilde y,t,\lambda_{k+})}
 \widehat{\widetilde N}^{(1)}
 (\widetilde y,t,\lambda_{k+}),
 \label{ic:Nhat-left-residue-upper}\\
 \Res_{\lambda=\lambda_{k-}}
 \widehat{\widetilde N}^{(1)}(\widetilde y,t,\lambda)
 &=
 \widetilde\varkappa_k
 \eul^{2\widehat f_1(\widetilde y,t,\lambda_{k-})}
 \widehat{\widetilde N}^{(2)}
 (\widetilde y,t,\lambda_{k-}).
 \label{ic:Nhat-left-residue-lower}
\end{align}
\end{subequations}

 \item As $\lambda\to\infty$,
 \begin{equation}\label{ic:Nhat-left-normalization}
  \widehat{\widetilde N}
  (\widetilde y,t,\lambda)
  =
  I+\ord(\lambda^{-1}).
 \end{equation}

\item At $\lambda_0=\pm A_1^{-1}$,
\begin{equation}\label{ic:Nhat-left-A1-branch-bound}
 \widehat{\widetilde N}(\widetilde y,t,\lambda)
 =
 \ord\bigl(|\lambda-\lambda_0|^{-1/4}\bigr).
\end{equation}
At $\lambda_0=\pm A_2^{-1}$, the behavior depends on $\gamma_2$.
If $\gamma_2\neq0$, then
\begin{equation}\label{ic:Nhat-left-A2-regular-bound}
 \widehat{\widetilde N}(\widetilde y,t,\lambda)
 =
 \ord(1),
 \qquad
 \lambda\to\lambda_0.
\end{equation}
If $\gamma_2=0$, then, for each
$\lambda_0\in\{\pm A_2^{-1}\}$, there exist scalar functions
$\widehat\tau_{\lambda_0,\pm}(\widetilde y,t)$ such that
\begin{subequations}\label{ic:Nhat-left-degenerate-behavior}
\begin{align}
 \widehat{\widetilde N}_+^{(1)}
 (\widetilde y,t,\lambda)
 &=\ord(1),&
 \widehat{\widetilde N}_+^{(2)}
 (\widetilde y,t,\lambda)
 -
 \frac{\widehat\tau_{\lambda_0,+}(\widetilde y,t)}
      {l_2(\lambda_+)}
 \widehat{\widetilde N}_+^{(1)}
 (\widetilde y,t,\lambda)
 &=\ord(1),
 \label{ic:Nhat-left-degenerate-upper}\\
 \widehat{\widetilde N}_-^{(2)}
 (\widetilde y,t,\lambda)
 &=\ord(1),&
 \widehat{\widetilde N}_-^{(1)}
 (\widetilde y,t,\lambda)
 -
 \frac{\widehat\tau_{\lambda_0,-}(\widetilde y,t)}
      {l_2(\lambda_-)}
 \widehat{\widetilde N}_-^{(2)}
 (\widetilde y,t,\lambda)
 &=\ord(1),
 \label{ic:Nhat-left-degenerate-lower}
\end{align}
\end{subequations}
as $\lambda\to\lambda_0$ along $\dot\Sigma_0$. In particular,
\begin{equation}\label{ic:Nhat-left-degenerate-bound}
 \widehat{\widetilde N}(\widetilde y,t,\lambda)
 =
 \ord\bigl(|\lambda-\lambda_0|^{-1/2}\bigr).
\end{equation}

\end{enumerate}

\begin{proposition}\label{prop:ic-left-rh-properties}
The left RH problem has at most one solution. If a solution exists,
then
\begin{equation}\label{ic:Nhat-left-determinant}
 \det\widehat{\widetilde N}
 (\widetilde y,t,\lambda)
 \equiv1.
\end{equation}
For scattering data arising from the direct problem, it satisfies
\begin{subequations}\label{ic:Nhat-left-symmetries}
\begin{align}
 \widehat{\widetilde N}(\lambda)
 &=
 \sigma_2
 \widehat{\widetilde N}(-\lambda)\sigma_2,
 &
 \widehat{\widetilde N}(\lambda)
 &=
 \sigma_1
 \overline{
 \widehat{\widetilde N}(\overline\lambda)}
 \sigma_1,
 &&\lambda\in\mathbb C^+,
 \label{ic:Nhat-left-off-contour-symmetries}\\
 \widehat{\widetilde N}(\lambda_+)
 &=
 -\sigma_2
 \widehat{\widetilde N}((-\lambda)_+)\sigma_3,
 &
 \widehat{\widetilde N}(\lambda_+)
 &=
 -\ii\sigma_1
 \overline{
 \widehat{\widetilde N}(\lambda_+)},
 &&\lambda\in
 \dot\Gamma_2\setminus
 \{\lambda_k\}_{k=1}^{N},
 \label{ic:Nhat-left-upper-symmetries}\\
 \widehat{\widetilde N}(\lambda_-)
 &=
 \sigma_2
 \widehat{\widetilde N}((-\lambda)_-)\sigma_3,
 &
 \widehat{\widetilde N}(\lambda_-)
 &=
 \ii\sigma_1
 \overline{
 \widehat{\widetilde N}(\lambda_-)},
 &&\lambda\in
 \dot\Gamma_2\setminus
 \{\lambda_k\}_{k=1}^{N}.
 \label{ic:Nhat-left-lower-symmetries}
\end{align}
\end{subequations}
\end{proposition}

\begin{proof}
The argument used in the proof of
Proposition~\ref{prop:ic-right-rh-properties} applies directly when
$\gamma_2\neq0$. 

It remains to consider the case $\gamma_2=0$. By \eqref{ic:Nhat-left-degenerate-upper},
\[
 \widehat{\widetilde N}_+^{(2)}
 =
 \frac{\widehat\tau_{\lambda_0,+}}
      {l_2(\lambda_+)}
 \widehat{\widetilde N}_+^{(1)}
 +\ord(1),
 \qquad
 \widehat{\widetilde N}_+^{(1)}=\ord(1).
\]
Therefore,
\[
 \det\widehat{\widetilde N}_+
 =
 \frac{\widehat\tau_{\lambda_0,+}}
      {l_2(\lambda_+)}
 \det\left(
  \widehat{\widetilde N}_+^{(1)},
  \widehat{\widetilde N}_+^{(1)}
 \right)
 +\ord(1)
 =
 \ord(1).
\]
Similarly, \eqref{ic:Nhat-left-degenerate-lower} yields
\[
 \det\widehat{\widetilde N}_-
 =
 \frac{\widehat\tau_{\lambda_0,-}}
      {l_2(\lambda_-)}
 \det\left(
  \widehat{\widetilde N}_-^{(2)},
  \widehat{\widetilde N}_-^{(2)}
 \right)
 +\ord(1)
 =
 \ord(1).
\]
Thus the singularities of the determinant at
$\lambda_0=\pm A_2^{-1}$ are removable. 
Thus $\det\widehat{\widetilde N}$ is bounded at
$\lambda=\pm A_2^{-1}$. Together with the argument in the proof of
Proposition~\ref{prop:ic-right-rh-properties}, this shows that it
extends to an entire function. By
\eqref{ic:Nhat-left-normalization},
$ \det\widehat{\widetilde N}(\lambda) \to 1,
 \qquad \lambda\to\infty$,
and Liouville's theorem yields
\[
 \det\widehat{\widetilde N}\equiv1.
\]
To prove uniqueness, assume that
$\widehat{\widetilde N}_1$ and
$\widehat{\widetilde N}_2$ are two solutions and set
$ \mathcal X
 :=
 \widehat{\widetilde N}_1
 \widehat{\widetilde N}_2^{-1}$.
Proceeding as in the proof of
Proposition~\ref{prop:ic-right-rh-properties}, we conclude that $\mathcal X$ has no
jumps or poles; its possible singularities at
$\lambda=\pm A_1^{-1}$ are removable, and tends to $I$ as $\lambda\to\infty$. 
At $\lambda_0=\pm A_2^{-1}$,
\eqref{ic:Nhat-left-degenerate-upper} and
\eqref{ic:Nhat-left-degenerate-lower} yield the factorizations
\[
 \widehat{\widetilde N}_{q,\pm}
 =
 B_{q,\pm}
 \left(
  I+
  \frac{\widehat\tau_{q,\pm}}
       {l_2(\lambda_\pm)}E_\pm
 \right),
 \qquad q=1,2,
\]
where
\[
 E_+
 =
 \begin{pmatrix}0&1\\0&0\end{pmatrix},
 \qquad
 E_-
 =
 \begin{pmatrix}0&0\\1&0\end{pmatrix},
\]
and $B_{q,\pm}$ and $B_{q,\pm}^{-1}$ are bounded as
$\lambda\to\lambda_0$. Since $E_\pm^2=0$,
\[
 \mathcal X_\pm
 =
 B_{1,\pm}
 \left(
  I+
  \frac{\widehat\tau_{1,\pm}-\widehat\tau_{2,\pm}}
       {l_2(\lambda_\pm)}E_\pm
 \right)
 B_{2,\pm}^{-1}.
\]
Therefore,
\[
 \mathcal X(\lambda)
 =
 \ord\bigl(l_2(\lambda)^{-1}\bigr)
 =
 \ord\bigl(|\lambda-\lambda_0|^{-1/2}\bigr).
\]
Thus,
the singularity is removable, and Liouville's theorem implies
$ \mathcal X\equiv I$.

Finally, the symmetry relations follow from the same uniqueness argument used
in the proof of Proposition~\ref{prop:ic-right-rh-properties}.

\end{proof}

\subsection{Connection between the left and right
Riemann--Hilbert problems}
\label{sec:ic-left-right-connection}

The variables $y$ and $\widetilde y$ are  related. Indeed,
\[
 \widetilde y_x
 =
 \frac{A_2}{A_1}y_x,
 \qquad
 \widetilde y_t
 =
 \frac{A_2}{A_1}y_t.
\]
Hence
\begin{equation}\label{ic:two-hodograph-relation}
 \widetilde y(x,t)
 =
 \frac{A_2}{A_1}y(x,t)+C,
\end{equation}
where $C$ is independent of $x$ and $t$. Evaluating this identity at
$(x,t)=(0,0)$ yields
\begin{equation}\label{ic:two-hodograph-constant}
 C
 =
 \frac1{A_1}
 \left[
  \int_{-\infty}^{0}
  \bigl(\rho(\xi,0)-A_1\bigr)\dd\xi
  +
  \int_{0}^{+\infty}
  \bigl(\rho(\xi,0)-A_2\bigr)\dd\xi
 \right].
\end{equation}

In view of \eqref{ic:two-hodograph-relation},
\begin{equation}\label{ic:ytilde-of-y}
 \widetilde y(y)
 :=
 \frac{A_2}{A_1}y+C.
\end{equation}
We now establish a direct relation between
$\widehat N(y,t,\lambda)$ and
$\widehat{\widetilde N}(\widetilde y(y),t,\lambda)$, the solutions of
the right and left RH problems, respectively.

\begin{proposition}\label{prop:ic-left-right-connection}
The right RH problem at $(y,t)$ is solvable if and only if the left
RH problem at $(\widetilde y(y),t)$ is solvable. In this case, their
solutions satisfy
\begin{equation}\label{ic:left-right-connection}
 \widehat N(y,t,\lambda)
 =
 \widehat{\widetilde N}
 \bigl(\widetilde y(y),t,\lambda\bigr)
 \mathcal P_{\mathrm{LR}}(y,t,\lambda),
\end{equation}
where
\begin{equation}\label{ic:left-right-connection-factor}
 \mathcal P_{\mathrm{LR}}(y,t,\lambda)
 =
 \begin{cases}
 \begin{pmatrix}
 \displaystyle
 \frac1{
 c_{11}(\lambda)
 \eul^{
 \widehat f_1(\widetilde y(y),t,\lambda)
 -\widehat f_2(y,t,\lambda)}}&0\\[4mm]
 0&
 \displaystyle
 c_{11}(\lambda)
 \eul^{
 \widehat f_1(\widetilde y(y),t,\lambda)
 -\widehat f_2(y,t,\lambda)}
 \end{pmatrix},
 &\lambda\in\mathbb C^+,\\[9mm]
 \begin{pmatrix}
 \displaystyle
 c_{22}(\lambda)
 \eul^{
 \widehat f_2(y,t,\lambda)
 -\widehat f_1(\widetilde y(y),t,\lambda)}&0\\[4mm]
 0&
 \displaystyle
 \frac1{
 c_{22}(\lambda)
 \eul^{
 \widehat f_2(y,t,\lambda)
 -\widehat f_1(\widetilde y(y),t,\lambda)}}
 \end{pmatrix},
 &\lambda\in\mathbb C^-.
 \end{cases}
\end{equation}

\end{proposition}

\begin{proof}
A direct calculation using
\eqref{ic:left-right-connection-factor} shows that
 $\det\mathcal P_{\mathrm{LR}}\equiv1$.
Moreover,
\eqref{ic:scattering-determinants},
\eqref{ic:Psi-tildePsi-relation}, and
\eqref{ic:tildePsi-large-lambda} imply
\begin{align*}
 c_{11}(\lambda)
 \eul^{
  \widehat f_1(\widetilde y(y),t,\lambda)
  -\widehat f_2(y,t,\lambda)}
 &=
 1+\ord(\lambda^{-1}),
 &&\lambda\to\infty,\quad\lambda\in\mathbb C^+,\\
 c_{22}(\lambda)
 \eul^{
  \widehat f_2(y,t,\lambda)
  -\widehat f_1(\widetilde y(y),t,\lambda)}
 &=
 1+\ord(\lambda^{-1}),
 &&\lambda\to\infty,\quad\lambda\in\mathbb C^-.
\end{align*}
Consequently,
\begin{equation}\label{ic:left-right-factor-infinity}
 \mathcal P_{\mathrm{LR}}(y,t,\lambda)
 =
 I+\ord(\lambda^{-1}),
 \qquad
 \lambda\to\infty.
\end{equation}

Suppose that the left RH problem is solvable and define
\[
 \mathcal M(y,t,\lambda)
 :=
 \widehat{\widetilde N}
 \bigl(\widetilde y(y),t,\lambda\bigr)
 \mathcal P_{\mathrm{LR}}(y,t,\lambda).
\]
Substitution of \eqref{ic:Ghat-right-def},
\eqref{ic:Ghat-left-def}, and
\eqref{ic:left-right-connection-factor} yields
\[
 \mathcal P_{\mathrm{LR},-}^{-1}
 \widehat G_{\mathrm L}
 \mathcal P_{\mathrm{LR},+}
 =
 \widehat G_{\mathrm R}.
\]
Therefore,
\[
 \mathcal M_+
 =
 \widehat{\widetilde N}_-
 \widehat G_{\mathrm L}
 \mathcal P_{\mathrm{LR},+}
 =
 \widehat{\widetilde N}_-
 \mathcal P_{\mathrm{LR},-}
 \widehat G_{\mathrm R}
 =
 \mathcal M_-\widehat G_{\mathrm R}.
\]

At the discrete eigenvalues, the factors involving $c_{11}$ and
$c_{22}$ in \eqref{ic:left-right-connection-factor}, together with
\eqref{ic:norming-constants-def}, transform
\eqref{ic:Nhat-left-residues} into
\eqref{ic:Nhat-right-residues}. The branch-point expansions in
Corollary~\ref{cor:ic-c-branch-expansion} similarly transform the
endpoint conditions of the left problem into those of the right
problem. Finally,
\eqref{ic:Nhat-left-normalization} and
\eqref{ic:left-right-factor-infinity} give
\[
 \mathcal M(y,t,\lambda)
 =
 I+\ord(\lambda^{-1}),
 \qquad
 \lambda\to\infty.
\]
Thus $\mathcal M$ solves the right RH problem.

Conversely, if the right RH problem is solvable, the same calculations
show that
\[
 \widehat N(y,t,\lambda)
 \mathcal P_{\mathrm{LR}}(y,t,\lambda)^{-1}
\]
solves the left RH problem at $(\widetilde y(y),t)$. Hence the two
problems are solvable simultaneously, and uniqueness yields
\eqref{ic:left-right-connection}.
\end{proof}

\section{Recovery of the solution to the 2-mCH system}
\label{sec:ic-two-component-reconstruction}

\subsection{Recovery of the solution to the 2-mCH system from the right RH problem}\label{sec:ic-two-component-reconstruction-right}

In this section, we recover the original fields from the solution of
the right RH problem. The reconstruction is based on the expansion of
$\widehat N$ at $\lambda=0$.

\begin{proposition}\label{prop:ic-Nhat-zero-expansion}
As
$\lambda\to0$ in $\mathbb C^+$,
\begin{equation}\label{ic:Nhat-zero-expansion}
 -\ii\mathcal E^{-1}\widehat N_+(y,t,\lambda)
 =
 \begin{pmatrix}
  0&\widehat a_1(y,t)\\
  \widehat a_1(y,t)^{-1}&0
 \end{pmatrix}
 +
 \lambda
 \begin{pmatrix}
  \widehat a_2(y,t)&0\\
  0&\widehat a_3(y,t)
 \end{pmatrix}
 +\ord(\lambda^2),
\end{equation}
where \begin{subequations}\label{ic:right-zero-coefficients}
\begin{align}
 \widehat a_1(y,t)
 &=
 \eul^{-D_2(x(y,t),t)},
 \label{ic:right-zero-a1}\\
 \widehat a_2(y,t)
 &=
 \frac12
 \eul^{D_2(x(y,t),t)-2r(x(y,t),t)}
 p(x(y,t),t),
 \label{ic:right-zero-a2}\\
 \widehat a_3(y,t)
 &=
 \frac12
 \eul^{-D_2(x(y,t),t)+2r(x(y,t),t)}
 q(x(y,t),t),
 \label{ic:right-zero-a3}
\end{align}
\end{subequations}
where $D_2$ defined in \eqref{ic:D12-r-identities}.
In particular, all $\widehat a_j(y,t)$ are real-valued; moreover,
$\widehat a_1(y,t)>0$.
\end{proposition}

\begin{proof}
By definition,
$
 \widehat N(y,t,\lambda)
 =
 N(x(y,t),t,\lambda)$.
Combining the definition \eqref{ic:N-right-def} with the Jost-column
expansions \eqref{ic:Psi-upper-zero-expansion}, the expansion of
$c_{11}$ at $\lambda=0$, and the identity
$\Psi_j^{\mathrm b}=G\Psi_j$, we obtain
\[
 -\ii\mathcal E^{-1}\widehat N_+
 =
 \begin{pmatrix}
  \dfrac{\lambda}{2}\eul^{D_2-2r}p
  &
  \eul^{-D_2}
  \\[3mm]
  \eul^{D_2}
  &
  \dfrac{\lambda}{2}\eul^{-D_2+2r}q
 \end{pmatrix}
 +\ord(\lambda^2),
\]
which proves \eqref{ic:Nhat-zero-expansion} and
\eqref{ic:right-zero-coefficients}.
\end{proof}

\begin{corollary}\label{cor:ic-right-recoverable-combinations}
The following identities hold:
\begin{subequations}\label{ic:right-recoverable-combinations}
\begin{align}
 x(y,t)+2r(x(y,t),t)
 &=
 y+A_2^2t-2\log\widehat a_1(y,t),
 \label{ic:recover-x-plus-r}\\
 p(x(y,t),t)\eul^{-2r(x(y,t),t)}
 &=
 2\widehat a_1(y,t)\widehat a_2(y,t),
 \label{ic:recover-balanced-p}\\
 q(x(y,t),t)\eul^{2r(x(y,t),t)}
 &=
 2\widehat a_1(y,t)^{-1}\widehat a_3(y,t),
 \label{ic:recover-balanced-q}\\
 Q(x(y,t),t)
 &=
 4\widehat a_2(y,t)\widehat a_3(y,t).
 \label{ic:recover-Q}
\end{align}
\end{subequations}
\end{corollary}

\begin{proof}
Combining \eqref{ic:right-zero-a1} with
\eqref{ic:D12-r-identities} and
\eqref{ic:right-hodograph-def} yields
\[
 D_2(x(y,t),t)
 =
 \frac12
 \bigl(
  x(y,t)-y-A_2^2t
 \bigr)
 +r(x(y,t),t),
\]
which proves \eqref{ic:recover-x-plus-r}. Relations
\eqref{ic:recover-balanced-p} and
\eqref{ic:recover-balanced-q} follow directly from
\eqref{ic:right-zero-a2} and \eqref{ic:right-zero-a3}. Multiplying
these two identities and using $Q=pq$ yields
\eqref{ic:recover-Q}.
\end{proof}

Notice that identity \eqref{ic:recover-x-plus-r} allows determining the combination
$x+2r$. To recover $x$ and $r$ separately, consider the right 
 change of variables at $t=0$,
\begin{equation}\label{ic:initial-y-map}
 y_0(x)
 :=
 x-\frac1{A_2}
 \int_x^{+\infty}
 \bigl(\rho(\xi,0)-A_2\bigr)\dd\xi,
\end{equation}
and denote its inverse by $x_0:=y_0^{-1}$. The initial value of $r$ in this variable
is
\begin{equation}\label{ic:advected-r-profile}
 r^\sharp(y)
 :=
 r(x_0(y),0)
 =
 \frac14\log
 \frac{m_0(x_0(y))}{n_0(x_0(y))}.
\end{equation}
Remarkably, the value of  $ r(x(y,t),t)$ does not depend on $t$.

\begin{lemma}\label{lem:ic-r-recovery}
For all $(y,t)$,
\begin{equation}\label{ic:recovery-r}
 r(x(y,t),t)
 =
 r^\sharp(y).
\end{equation}
\end{lemma}

\begin{proof}
Differentiating the identity
$ x(y(x,t),t)=x$
with respect to $t$ at fixed $x$, we obtain
\[
 x_y(y(x,t),t)y_t(x,t)
 +
 x_t(y(x,t),t)
 =
 0.
\]
By \eqref{ic:right-hodograph-derivatives},
 $y_t=-Qy_x$,
while the inverse-function identity gives
 $x_y(y(x,t),t)y_x(x,t)=1$.
Therefore,
$ x_t(y(x,t),t)=Q(x,t)$,
or, equivalently,
\begin{equation}\label{ic:inverse-map-time-derivative}
 x_t(y,t)=Q(x(y,t),t).
\end{equation}

Using the transport equation
$r_t+Qr_x=0$ (see \eqref{ic:rho-r-conservation})
and \eqref{ic:inverse-map-time-derivative}, we obtain
\[
 \frac{\dd}{\dd t}r(x(y,t),t)
 =
 r_t(x(y,t),t)
 +
 x_t(y,t)r_x(x(y,t),t)
 =
 0.
\]
Hence \(r(x(y,t),t)\) is independent of \(t\). At \(t=0\),
\[
 r(x(y,0),0)
 =
 r(x_0(y),0)
 =
 r^\sharp(y),
\]
which proves \eqref{ic:recovery-r}.
\end{proof}

\begin{proposition}\label{prop:ic-recovery}
The solution of the 2-mCH system admits the following parametric
reconstruction. 
\begin{enumerate}
    \item 
    The inverse variable map is
\begin{equation}\label{ic:recovery-x}
 x(y,t)
 =
 y+A_2^2t
 -2\log\widehat a_1(y,t)
 -2r^\sharp(y).
\end{equation}
Equivalently,
\begin{equation}\label{ic:recovery-x-time}
 x(y,t)
 =
 x_0(y)
 +4\int_0^t
 \widehat a_2(y,\tau)\widehat a_3(y,\tau)\dd\tau.
\end{equation}
\item 
The functions $p$ and $q$ are given by
\begin{subequations}\label{ic:recovery-pq}
\begin{align}
 p(x(y,t),t)
 &=
 2\eul^{2r^\sharp(y)}
 \widehat a_1(y,t)\widehat a_2(y,t),\\
 q(x(y,t),t)
 &=
 2\eul^{-2r^\sharp(y)}
 \widehat a_1(y,t)^{-1}\widehat a_3(y,t),
\end{align}
\end{subequations}
while the momenta $m$ and $n$ are given by
\begin{subequations}\label{ic:recovery-mn}
\begin{align}
 m(x(y,t),t)
 &=
 \frac{A_2}{x_y(y,t)}
 \eul^{2r^\sharp(y)},\\
 n(x(y,t),t)
 &=
 \frac{A_2}{x_y(y,t)}
 \eul^{-2r^\sharp(y)}.
\end{align}
\end{subequations}

Finally,  $u$ and $v$ can be recovered by
\begin{subequations}\label{ic:recovery-uv}
\begin{align}
 u(x(y,t),t)
 &=
 \frac{A_2}{2}
 \int_{\mathbb R}
 \eul^{-\lvert x(y,t)-x(\eta,t)\rvert}
 \eul^{2r^\sharp(\eta)}
 \dd\eta,
 \label{ic:recovery-u}\\
 v(x(y,t),t)
 &=
 \frac{A_2}{2}
 \int_{\mathbb R}
 \eul^{-\lvert x(y,t)-x(\eta,t)\rvert}
 \eul^{-2r^\sharp(\eta)}
 \dd\eta.
 \label{ic:recovery-v}
\end{align}
\end{subequations}
\end{enumerate}

\end{proposition}

\begin{proof}

Combining \eqref{ic:recovery-r} with
\eqref{ic:recover-x-plus-r}, we obtain
\eqref{ic:recovery-x}. Equations
\eqref{ic:recover-Q} and $x_t=Q$ yield
\eqref{ic:recovery-x-time}, while
\eqref{ic:recover-balanced-p} and
\eqref{ic:recover-balanced-q} together with \eqref{ic:recovery-r} imply
\eqref{ic:recovery-pq}.

Since $y_x=\rho/A_2$, it follows that
\begin{equation}\label{rho_}
    \rho(x(y,t),t)=\frac{A_2}{x_y(y,t)}.
\end{equation}
Together with
$m=\rho\eul^{2r}$ and $n=\rho\eul^{-2r}$, this proves
\eqref{ic:recovery-mn}. 

Finally, since
\[
 m=u-u_{xx},
 \qquad
 n=v-v_{xx},
\]
the fields $u$ and $v$ are obtained by inverting
$1-\partial_x^2$. Its Green function on $\mathbb R$ is
$\frac12\eul^{-|x|}$, and hence
\begin{align*}
 u(x,t)
 &=
 \frac12\int_{\mathbb R}
 \eul^{-|x-z|}m(z,t)\dd z,\\
 v(x,t)
 &=
 \frac12\int_{\mathbb R}
 \eul^{-|x-z|}n(z,t)\dd z.
\end{align*}
Evaluating these identities at $x=x(y,t)$ and making the change of
variables
\[
 z=x(\eta,t),
 \qquad
 \dd z=x_\eta(\eta,t)\dd\eta,
\]
we obtain, using \eqref{ic:recovery-mn},
\begin{align*}
 m(x(\eta,t),t)x_\eta(\eta,t)
 &=
 A_2\eul^{2r^\sharp(\eta)},\qquad
 n(x(\eta,t),t)x_\eta(\eta,t)
 =
 A_2\eul^{-2r^\sharp(\eta)}.
\end{align*}
Therefore,
 \eqref{ic:recovery-uv} follows.
\end{proof}

\begin{remark}
Unlike the scalar case, the leading term in
\eqref{ic:Nhat-zero-expansion} determines $x+2r$ rather than $x$
alone. It follows that the reconstruction of the solution to the 2-mCH equation involves not only the solution of the  RH problem constructed in terms of spectral functions associated with initial data, but also incorporates the transported profile $r^\sharp$, which is determined directly from the initial data.
When $u=v$, one has
$r^\sharp\equiv0$, and \eqref{ic:recovery-x} reduces to the scalar mCH formula; see \cite{KST22}.
\end{remark}

\begin{remark}\label{rem:ic-recovery-at-infinity}
The expansion at $\lambda=0$ is not the only way to recover the inverse variable map. Assume that the solution of
the right RH problem has the expansion
\begin{equation}\label{ic:Nhat-infinity-first-order}
 \widehat N(y,t,\lambda)
 =
 I+\frac{\widehat N^{[1]}(y,t)}{\lambda}
 +\ord(\lambda^{-2}),
 \qquad
 \lambda\to\infty.
\end{equation}
Set
\[
 X(y,t):=x(y,t)+2r(x(y,t),t).
\]
Comparison of the terms of order $\lambda^{-1}$ in the spatial part of
the transformed Lax pair gives
\begin{equation}\label{ic:Xy-from-infinity}
 X_y(y,t)
 =
 2A_2\bigl(\widehat N^{[1]}(y,t)\bigr)_{12}.
\end{equation}
Since
\[
 r(x(y,t),t)=r^\sharp(y),
\]
we obtain
\begin{equation}\label{ic:xy-from-infinity}
 x_y(y,t)
 =
 2A_2\bigl(\widehat N^{[1]}(y,t)\bigr)_{12}
 -2(r^\sharp)'(y).
\end{equation}
Therefore,
\begin{equation}\label{ic:rho-from-infinity}
 \rho(x(y,t),t)
 =
 \frac{A_2}{
  2A_2\bigl(\widehat N^{[1]}(y,t)\bigr)_{12}
  -2(r^\sharp)'(y)}.
\end{equation}
This recovers $\rho$ from the large-$\lambda$ coefficient without
differentiating $\widehat a_1$ in the expansion at $\lambda=0$ (cf. \eqref{ic:recovery-x}, \eqref{rho_}).

The inverse variable map may equivalently be written in the spatial
integral form
\begin{equation}\label{ic:x-from-infinity}
\begin{aligned}
 x(y,t)
 =
 y+A_2^2t
 -\int_y^{+\infty}
 \Bigl[
  2A_2\bigl(\widehat N^{[1]}(\eta,t)\bigr)_{12}
  -2(r^\sharp)'(\eta)-1
 \Bigr]\dd\eta.
\end{aligned}
\end{equation}
\end{remark}

\subsection{Recovery of the solution to the 2-mCH system from the left
RH problem}
\label{sec:ic-two-component-reconstruction-left}

The solution may equivalently be recovered from the left RH problem. We begin with the expansion of
$\widehat{\widetilde N}$ at $\lambda=0$.

\begin{proposition}\label{prop:ic-Ntildehat-zero-expansion}
As
$\lambda\to0$ in $\mathbb C^+$,
\begin{equation}\label{ic:Ntildehat-zero-expansion}
 -\ii\mathcal E^{-1}
 \widehat{\widetilde N}_+(\widetilde y,t,\lambda)
 =
 \begin{pmatrix}
  0&\widehat b_1(\widetilde y,t)\\
  \widehat b_1(\widetilde y,t)^{-1}&0
 \end{pmatrix}
 +
 \lambda
 \begin{pmatrix}
  \widehat b_2(\widetilde y,t)&0\\
  0&\widehat b_3(\widetilde y,t)
 \end{pmatrix}
 +\ord(\lambda^2),
\end{equation}
where
\begin{subequations}\label{ic:left-zero-coefficients}
\begin{align}
 \widehat b_1(\widetilde y,t)
 &=
 \eul^{D_1(x(\widetilde y,t),t)},
 \label{ic:left-zero-b1}\\
 \widehat b_2(\widetilde y,t)
 &=
 \frac12
 \eul^{-D_1(x(\widetilde y,t),t)
       -2r(x(\widetilde y,t),t)}
 p(x(\widetilde y,t),t),
 \label{ic:left-zero-b2}\\
 \widehat b_3(\widetilde y,t)
 &=
 \frac12
 \eul^{D_1(x(\widetilde y,t),t)
       +2r(x(\widetilde y,t),t)}
 q(x(\widetilde y,t),t).
 \label{ic:left-zero-b3}
\end{align}
\end{subequations}
In particular, the coefficients are real-valued and
$\widehat b_1(\widetilde y,t)>0$.
\end{proposition}

\begin{proof}
By definition,
\[
 \widehat{\widetilde N}(\widetilde y,t,\lambda)
 =
 \widetilde N(x(\widetilde y,t),t,\lambda).
\]
Combining \eqref{ic:N-left-def} with the Jost-column expansions
\eqref{ic:Psi-upper-zero-expansion}, the expansion of $c_{11}$ at
$\lambda=0$, and $\Psi_j^{\mathrm b}=G\Psi_j$, we obtain
\[
 -\ii\mathcal E^{-1}\widehat{\widetilde N}_+
 =
 \begin{pmatrix}
  \dfrac{\lambda}{2}\eul^{-D_1-2r}p
  &
  \eul^{D_1}
  \\[3mm]
  \eul^{-D_1}
  &
  \dfrac{\lambda}{2}\eul^{D_1+2r}q
 \end{pmatrix}
 +\ord(\lambda^2).
\]
This proves \eqref{ic:Ntildehat-zero-expansion} and
\eqref{ic:left-zero-coefficients}.
\end{proof}

\begin{corollary}\label{cor:ic-left-recoverable-combinations}
The following identities hold:
\begin{subequations}\label{ic:left-recoverable-combinations}
\begin{align}
 x(\widetilde y,t)+2r(x(\widetilde y,t),t)
 &=
 \widetilde y+A_1^2t
 -2\log\widehat b_1(\widetilde y,t),
 \label{ic:left-recover-x-plus-r}\\
 p(x(\widetilde y,t),t)
 \eul^{-2r(x(\widetilde y,t),t)}
 &=
 2\widehat b_1(\widetilde y,t)
 \widehat b_2(\widetilde y,t),
 \label{ic:left-recover-balanced-p}\\
 q(x(\widetilde y,t),t)
 \eul^{2r(x(\widetilde y,t),t)}
 &=
 2\widehat b_1(\widetilde y,t)^{-1}
 \widehat b_3(\widetilde y,t),
 \label{ic:left-recover-balanced-q}\\
 Q(x(\widetilde y,t),t)
 &=
 4\widehat b_2(\widetilde y,t)
 \widehat b_3(\widetilde y,t).
 \label{ic:left-recover-Q}
\end{align}
\end{subequations}
\end{corollary}

\begin{proof}
Equations \eqref{ic:D12-r-identities} and
\eqref{ic:left-hodograph-def} imply
\[
 D_1(x(\widetilde y,t),t)
 =
 \frac12
 \bigl(
  \widetilde y-x(\widetilde y,t)+A_1^2t
 \bigr)
 -r(x(\widetilde y,t),t).
\]
Together with \eqref{ic:left-zero-coefficients}, this proves
\eqref{ic:left-recoverable-combinations}.
\end{proof}

To recover $x(y,t)$, introduce the left change of variable at $t=0$,
\begin{equation}\label{ic:initial-ytilde-map}
 \widetilde y_0(x)
 :=
 x+\frac1{A_1}
 \int_{-\infty}^{x}
 \bigl(\rho(\xi,0)-A_1\bigr)\dd\xi,
\end{equation}
and denote its inverse by
$\widetilde x_0:=\widetilde y_0^{-1}$. Define
\begin{equation}\label{ic:left-advected-r-profile}
 \widetilde r^\sharp(\widetilde y)
 :=
 r(\widetilde x_0(\widetilde y),0)
 =
 \frac14\log
 \frac{
  m_0(\widetilde x_0(\widetilde y))
 }{
  n_0(\widetilde x_0(\widetilde y))
 }.
\end{equation}

Applying the argument of Lemma~\ref{lem:ic-r-recovery} to the left
variable $\widetilde y$, we obtain the following result.

\begin{lemma}\label{lem:ic-left-r-recovery}
For all $(\widetilde y,t)$,
\begin{equation}\label{ic:left-recovery-r}
 r(x(\widetilde y,t),t)
 =
 \widetilde r^\sharp(\widetilde y).
\end{equation}
\end{lemma}

The argument used in Proposition~\ref{prop:ic-recovery} applies to the left RH problem. In view of
\eqref{ic:left-recoverable-combinations} and
\eqref{ic:left-recovery-r}, we obtain the following result.

\begin{proposition}\label{prop:ic-left-recovery}
The solution of the 2-mCH system admits the following parametric
reconstruction from the left RH problem. The inverse variable map is
\begin{equation}\label{ic:left-recovery-x}
 x(\widetilde y,t)
 =
 \widetilde y+A_1^2t
 -2\log\widehat b_1(\widetilde y,t)
 -2\widetilde r^\sharp(\widetilde y).
\end{equation}
Equivalently,
\begin{equation}\label{ic:left-recovery-x-time}
 x(\widetilde y,t)
 =
 \widetilde x_0(\widetilde y)
 +4\int_0^t
 \widehat b_2(\widetilde y,\tau)
 \widehat b_3(\widetilde y,\tau)\dd\tau.
\end{equation}

The functions $p$ and $q$ are
\begin{subequations}\label{ic:left-recovery-pq}
\begin{align}
 p(x(\widetilde y,t),t)
 &=
 2\eul^{2\widetilde r^\sharp(\widetilde y)}
 \widehat b_1(\widetilde y,t)
 \widehat b_2(\widetilde y,t),\\
 q(x(\widetilde y,t),t)
 &=
 2\eul^{-2\widetilde r^\sharp(\widetilde y)}
 \widehat b_1(\widetilde y,t)^{-1}
 \widehat b_3(\widetilde y,t),
\end{align}
\end{subequations}
while the momenta satisfy
\begin{subequations}\label{ic:left-recovery-mn}
\begin{align}
 m(x(\widetilde y,t),t)
 &=
 \frac{A_1}{x_{\widetilde y}(\widetilde y,t)}
 \eul^{2\widetilde r^\sharp(\widetilde y)},\\
 n(x(\widetilde y,t),t)
 &=
 \frac{A_1}{x_{\widetilde y}(\widetilde y,t)}
 \eul^{-2\widetilde r^\sharp(\widetilde y)}.
\end{align}
\end{subequations}
Finally, $u$ and $v$ can be recovered by
\begin{subequations}\label{ic:left-recovery-uv}
\begin{align}
 u(x(\widetilde y,t),t)
 &=
 \frac{A_1}{2}
 \int_{\mathbb R}
 \eul^{-\lvert
       x(\widetilde y,t)-x(\eta,t)
       \rvert}
 \eul^{2\widetilde r^\sharp(\eta)}
 \dd\eta,
 \label{ic:left-recovery-u}\\
 v(x(\widetilde y,t),t)
 &=
 \frac{A_1}{2}
 \int_{\mathbb R}
 \eul^{-\lvert
       x(\widetilde y,t)-x(\eta,t)
       \rvert}
 \eul^{-2\widetilde r^\sharp(\eta)}
 \dd\eta.
 \label{ic:left-recovery-v}
\end{align}
\end{subequations}
\end{proposition}

\begin{remark}\label{rem:ic-left-recovery-at-infinity}
Assume that the solution of
the right RH problem has the large-$\lambda$
expansion
\begin{equation}\label{ic:Ntildehat-infinity-first-order}
 \widehat{\widetilde N}(\widetilde y,t,\lambda)
 =
 I+\frac{
  \widehat{\widetilde N}^{[1]}(\widetilde y,t)
 }{\lambda}
 +\ord(\lambda^{-2}),
 \qquad
 \lambda\to\infty.
\end{equation}
The argument of Remark~\ref{rem:ic-recovery-at-infinity}, applied to
the left RH problem, gives
\begin{equation}\label{ic:left-Xy-from-infinity}
 \partial_{\widetilde y}
 \bigl(
  x(\widetilde y,t)
  +2r(x(\widetilde y,t),t)
 \bigr)
 =
 2A_1
 \bigl(
  \widehat{\widetilde N}^{[1]}(\widetilde y,t)
 \bigr)_{12}.
\end{equation}
Hence
\begin{equation}\label{ic:left-xy-from-infinity}
 x_{\widetilde y}(\widetilde y,t)
 =
 2A_1
 \bigl(
  \widehat{\widetilde N}^{[1]}(\widetilde y,t)
 \bigr)_{12}
 -2(\widetilde r^\sharp)'(\widetilde y),
\end{equation}
and
\begin{equation}\label{ic:left-rho-from-infinity}
 \rho(x(\widetilde y,t),t)
 =
 \frac{A_1}{
  2A_1
  \bigl(
   \widehat{\widetilde N}^{[1]}(\widetilde y,t)
  \bigr)_{12}
  -2(\widetilde r^\sharp)'(\widetilde y)}.
\end{equation}
Using
\[
 x(\widetilde y,t)-\widetilde y-A_1^2t
 \to 0,
 \qquad
 \widetilde y\to-\infty,
\]
the inverse variable map can also be written as
\begin{equation}\label{ic:left-x-from-infinity}
\begin{aligned}
 x(\widetilde y,t)
 =
 \widetilde y+A_1^2t
 +\int_{-\infty}^{\widetilde y}
 \Bigl[
  2A_1
  \bigl(
   \widehat{\widetilde N}^{[1]}(\eta,t)
  \bigr)_{12}
  -2(\widetilde r^\sharp)'(\eta)-1
 \Bigr]\dd\eta.
\end{aligned}
\end{equation}
\end{remark}

\section{Different background values for the two components}
\label{sec:different-component-backgrounds}

In this section, we discuss the modifications required when the two
components approach different constants at each spatial end. Suppose
that
\[
 (u(x,t),v(x,t))
 \to 
 \bigl(A_j^{(u)},A_j^{(v)}\bigr),
 \qquad
 x\to(-1)^j\infty,
 \qquad
 A_j^{(u)},A_j^{(v)}>0,
 \qquad j=1,2,
\]
and assume that the corresponding limits of the quantities entering
the Lax pair are
\[
 m,p \to  A_j^{(u)},
 \qquad
 n,q \to  A_j^{(v)},
 \qquad
 r_x \to 0,
 \qquad
 x\to(-1)^j\infty.
\]
Set
\begin{equation}\label{ic:general-background-parameters}
 \mathcal A_j
 :=
 \sqrt{A_j^{(u)}A_j^{(v)}},
 \qquad
 r_j
 :=
 \frac14
 \log\frac{A_j^{(u)}}{A_j^{(v)}}.
\end{equation}
Then
\[
 \rho \to \mathcal A_j,
 \qquad
  \nu  \to 
 \frac{A_j^{(u)}}{A_j^{(v)}},
 \qquad
 r \to  r_j,
 \qquad
 x\to(-1)^j\infty.
\]
Moreover,
\begin{equation}\label{ic:general-background-balanced-limits}
 p \nu ^{-1/2} \to \mathcal A_j,
 \qquad
 q \nu ^{1/2} \to \mathcal A_j,
 \qquad
 Q \to \mathcal A_j^2,
 \qquad
 x\to(-1)^j\infty.
\end{equation}
Assume throughout this section that
\begin{equation}\label{ic:general-background-ordering}
 0<\mathcal A_1<\mathcal A_2,
\end{equation}
or, equivalently,
\begin{equation}\label{ic:general-background-ordering-equiv}
 A_1^{(u)}A_1^{(v)}
 <
 A_2^{(u)}A_2^{(v)}.
\end{equation}

For $j=1,2$, set
\begin{equation}\label{ic:general-background-cuts}
 \Gamma_j^{\mathrm{bg}}
 :=
 \left[
  -\mathcal A_j^{-1},\mathcal A_j^{-1}
 \right]
\end{equation}
and define
\begin{equation}\label{ic:general-background-l-def}
 l_j^{\mathrm{bg}}(\lambda)
 :=
 \sqrt{\lambda^2-\mathcal A_j^{-2}}.
\end{equation}
The branch cut is taken along $\Gamma_j^{\mathrm{bg}}$, and the branch
is normalized by
\[
 l_j^{\mathrm{bg}}(\lambda)\sim\lambda,
 \qquad
 \lambda\to\infty.
\]
The corresponding branch points are
\begin{equation}\label{ic:general-background-branch-points}
 \lambda=\pm\mathcal A_j^{-1}
 =
 \pm\frac1{\sqrt{A_j^{(u)}A_j^{(v)}}}.
\end{equation}
Consequently, the location and ordering of the spectral cuts depend
only on the products
$A_j^{(u)}A_j^{(v)}$, $j=1,2$.

By \eqref{ic:general-background-parameters} and
\eqref{ic:general-background-balanced-limits}, the limiting
coefficients in \eqref{ic:balanced-U} and \eqref{ic:balanced-V} have
the same form as in the equal-background case, with $A_j$ replaced by
$\mathcal A_j$. Accordingly, the matrices $H_j$ and $H_j^{-1}$ are
defined by \eqref{ic:H-def} and \eqref{ic:H-inverse}, respectively,
with $A_j$ replaced by $\mathcal A_j$ and $l_j$ replaced by
$l_j^{\mathrm{bg}}$. The identities \eqref{ic:H-identities} and
\eqref{ic:H-diagonalization} remain valid after the same
replacements.

Define
\begin{equation}\label{ic:general-background-delta}
 \delta_j^{\mathrm{bg}}
 :=
 \frac{\rho-\mathcal A_j}{2\mathcal A_j}-r_x
\end{equation}
and
\begin{equation}\label{ic:general-background-phase}
 f_j^{\mathrm{bg}}(x,t,\lambda)
 :=
 \ii\mathcal A_jl_j^{\mathrm{bg}}(\lambda)
 \left[
  \frac1{2\mathcal A_j}
  \int_{(-1)^j\infty}^{x}
  \bigl(\rho(\xi,t)-\mathcal A_j\bigr)\dd\xi
  +\frac x2
  -t\left(
    \frac1{\lambda^2}+\frac{\mathcal A_j^2}{2}
   \right)
 \right].
\end{equation}
With these definitions, the calculation in
Proposition~\ref{prop:inner-cut-transformed-lax} yields the same
transformed Lax pair \eqref{ic:transformed-lax}, with
 $A_j,~ l_j,~\delta_j,~ f_j$ replaced by $\mathcal A_j,~ l_j^{\mathrm{bg}},~
 \delta_j^{\mathrm{bg}},~ f_j^{\mathrm{bg}}$,
respectively. The coefficient formulas \eqref{ic:W-def} and
\eqref{ic:hatV-def} retain the same form.

To carry the direct spectral and Riemann--Hilbert analysis through,
we impose the analogues of
\eqref{ic:standing-direct-assumptions}, with $A_j$ and $\delta_j$
replaced by $\mathcal A_j$ and $\delta_j^{\mathrm{bg}}$,
respectively.

Notice that the balancing gauge now generally has non-identity limits:
\begin{equation}\label{ic:general-background-gauge-limits}
 G(x,t)=\eul^{-r(x,t)\sigma_3}
 \to 
 G_j:=\eul^{-r_j\sigma_3},
 \qquad
 x\to(-1)^j\infty.
\end{equation}
Accordingly, the normalization of the original Jost solution $\Psi_j$
as $x\to(-1)^j\infty$ contains the additional factor $G_j^{-1}$.
Indeed, setting
\[
 R_j^{\mathrm{bg}}
 :=
 f_j^{\mathrm{bg}}\sigma_3,
\]
the normalized and balanced Jost solutions satisfy
\[
\widetilde\Psi_j(x,t,\lambda) \to  I,
 \qquad
 \Psi_j^{\mathrm b}(x,t,\lambda)
 \sim
 H_j^{-1}(\lambda)
 \eul^{-R_j^{\mathrm{bg}}(x,t,\lambda)},
 \qquad
 x\to(-1)^j\infty.
\]
Returning to the original Jost solutions through
$\Psi_j=G^{-1}\Psi_j^{\mathrm b}$ and using
\eqref{ic:general-background-gauge-limits}, we obtain
\begin{equation}\label{ic:general-background-original-Jost-normalization}
 \Psi_j(x,t,\lambda)
 \sim
 G_j^{-1}H_j^{-1}(\lambda)
 \eul^{-R_j^{\mathrm{bg}}(x,t,\lambda)},
 \qquad
 x\to(-1)^j\infty.
\end{equation}

The identity $\Psi_j^{\mathrm b}=G\Psi_j$ transforms the scattering
relation $\Psi_1=\Psi_2c$ into
$\Psi_1^{\mathrm b}=\Psi_2^{\mathrm b}c$. Thus the same scattering
matrix relates the original and balanced Jost solutions on the common
continuous spectrum.

Furthermore, the right and left RH problems retain the same structure,
with $A_j$, $l_j$, and $f_j$ replaced by
$\mathcal A_j$, $l_j^{\mathrm{bg}}$, and
$f_j^{\mathrm{bg}}$, respectively. Their jump and residue conditions
and their normalization at infinity remain unchanged in form. The
matrices $G_j$ do not enter the RH problems, since they are formulated in
terms of the balanced Jost solutions.

For the analysis near $\lambda=0$, define
\begin{equation}\label{ic:general-background-D-def}
 D_1^{\mathrm{bg}}(x,t)
 :=
 \int_{-\infty}^{x}
 \delta_1^{\mathrm{bg}}(\xi,t)\dd\xi,
 \qquad
 D_2^{\mathrm{bg}}(x,t)
 :=
 \int_x^{+\infty}
 \delta_2^{\mathrm{bg}}(\xi,t)\dd\xi.
\end{equation}
Using \eqref{ic:general-background-delta} and the limits
$r(x,t)\to r_j$ as $x\to(-1)^j\infty$, we obtain
\begin{subequations}\label{ic:general-background-D-identities}
\begin{align}
 D_1^{\mathrm{bg}}(x,t)
 &=
 \frac1{2\mathcal A_1}
 \int_{-\infty}^{x}
 \bigl(\rho(\xi,t)-\mathcal A_1\bigr)\dd\xi
 -\bigl(r(x,t)-r_1\bigr),
 \label{ic:general-background-D1}\\
 D_2^{\mathrm{bg}}(x,t)
 &=
 \frac1{2\mathcal A_2}
 \int_x^{+\infty}
 \bigl(\rho(\xi,t)-\mathcal A_2\bigr)\dd\xi
 +\bigl(r(x,t)-r_2\bigr).
 \label{ic:general-background-D2}
\end{align}
\end{subequations}
The analogues of the modified phases \eqref{ic:f0-def} are
\begin{equation}\label{ic:general-background-f0-def}
 f_1^{0,\mathrm{bg}}
 :=
 f_1^{\mathrm{bg}}
 -\ii\mathcal A_1l_1^{\mathrm{bg}}D_1^{\mathrm{bg}},
 \qquad
 f_2^{0,\mathrm{bg}}
 :=
 f_2^{\mathrm{bg}}
 +\ii\mathcal A_2l_2^{\mathrm{bg}}D_2^{\mathrm{bg}}.
\end{equation}
Substitution of
\eqref{ic:general-background-D-identities} into
\eqref{ic:general-background-f0-def} gives
\begin{equation}\label{ic:general-background-f0-common}
 f_j^{0,\mathrm{bg}}(x,t,\lambda)
 =
 \ii\mathcal A_jl_j^{\mathrm{bg}}(\lambda)
 \left[
  \frac{x}{2}+r(x,t)-r_j
  -t\left(
   \frac1{\lambda^2}+\frac{\mathcal A_j^2}{2}
  \right)
 \right],
 \qquad j=1,2.
\end{equation}
Thus the zero-adapted phase retains the form of
\eqref{ic:f0-common-form}, with $r$ replaced by $r-r_j$.

The corresponding new right and left spatial variables are
\begin{align}
 y(x,t)
 &:=
 x-\frac1{\mathcal A_2}
 \int_x^{+\infty}
 \bigl(\rho(\xi,t)-\mathcal A_2\bigr)\dd\xi
 -\mathcal A_2^2t,
 \label{ic:general-background-right-coordinate}\\
 \widetilde y(x,t)
 &:=
 x+\frac1{\mathcal A_1}
 \int_{-\infty}^{x}
 \bigl(\rho(\xi,t)-\mathcal A_1\bigr)\dd\xi
 -\mathcal A_1^2t.
 \label{ic:general-background-left-coordinate}
\end{align}
The leading coefficients in the expansions of the solutions to the  right and left RH
problems near  $\lambda=0$ then determine
\begin{align}
 y+\mathcal A_2^2t-2\log\widehat a_1(y,t)
 &=
 x(y,t)+2\bigl(r(x(y,t),t)-r_2\bigr),
 \label{ic:general-background-right-combination}\\
 \widetilde y+\mathcal A_1^2t
 -2\log\widehat b_1(\widetilde y,t)
 &=
 x(\widetilde y,t)
 +2\bigl(r(x(\widetilde y,t),t)-r_1\bigr).
 \label{ic:general-background-left-combination}
\end{align}
Let $r_{\mathrm R}^{\sharp}$ and $r_{\mathrm L}^{\sharp}$ denote the
initial profiles of $r$ in the right and left variables,
respectively, defined as in
\eqref{ic:advected-r-profile} and
\eqref{ic:left-advected-r-profile}. Since
$r_t+Qr_x=0$, the argument of Lemma~\ref{lem:ic-r-recovery} gives
\[
 r(x(y,t),t)=r_{\mathrm R}^{\sharp}(y),
 \qquad
 r(x(\widetilde y,t),t)
 =
 r_{\mathrm L}^{\sharp}(\widetilde y).
\]
Consequently, the inverse variable maps are recovered from
\begin{align}
 x(y,t)
 &=
 y+\mathcal A_2^2t
 -2\log\widehat a_1(y,t)
 -2\bigl(r_{\mathrm R}^{\sharp}(y)-r_2\bigr),
 \label{ic:general-background-right-recovery}\\
 x(\widetilde y,t)
 &=
 \widetilde y+\mathcal A_1^2t
 -2\log\widehat b_1(\widetilde y,t)
 -2\bigl(r_{\mathrm L}^{\sharp}(\widetilde y)-r_1\bigr).
 \label{ic:general-background-left-recovery}
\end{align}

\begin{remark}\label{rem:ic-common-background-ratio}
If
\begin{equation}\label{ic:common-background-ratio}
 \frac{A_1^{(u)}}{A_1^{(v)}}
 =
 \frac{A_2^{(u)}}{A_2^{(v)}},
\end{equation}
then the problem can be reduced to the equal-background setting by a
constant rescaling. In fact, \eqref{ic:common-background-ratio}
implies $r_1=r_2$, and there exists a constant $\gamma>0$ such that
\begin{equation*}
 \gamma^2
 =
 \frac{A_j^{(v)}}{A_j^{(u)}},
 \qquad j=1,2.
\end{equation*}
Setting
\begin{equation}\label{ic:constant-component-rescaling}
 \bar u=\gamma u,
 \qquad
 \bar v=\gamma^{-1}v,
\end{equation}
we obtain
\begin{equation*}
 (\bar u,\bar v)
 \to 
 (\mathcal A_j,\mathcal A_j),
 \qquad
 x\to(-1)^j\infty.
\end{equation*}
Moreover,
\begin{equation*}
 \bar m=\gamma m,
 \qquad
 \bar n=\gamma^{-1}n,
 \qquad
 (\bar u-\bar u_x)(\bar v+\bar v_x)
 =
 (u-u_x)(v+v_x),
\end{equation*}
so the rescaling \eqref{ic:constant-component-rescaling} leaves
\eqref{eq:2mch} invariant.

\end{remark}

\section{Concluding remarks}
\label{sec:concluding-remarks}

The analysis in this paper has been carried out under the assumption
$0<A_1<A_2$. The opposite ordering $0<A_2<A_1$ can be treated by
a corresponding reorganization of the spectral construction, as in
the scalar mCH case \cite{KST22}. In this setting,
\begin{equation*}
 \Sigma_2\subset\Sigma_1,
 \qquad
 \Gamma_1\subset\Gamma_2,
\end{equation*}
and the ``intermediate'' part of the spectrum becomes
\begin{equation*}
 \Sigma_0
 =
 \left[-A_2^{-1},-A_1^{-1}\right]
 \cup
 \left[A_1^{-1},A_2^{-1}\right].
\end{equation*}
Accordingly, the roles of the indices $1$ and $2$ in the analytic
continuation of the Jost columns, the scattering relations on
$\Sigma_0$, and the construction of the RH problems are interchanged.
The discrete spectrum is then contained in $\dot\Gamma_1$. The
remaining parts of the analysis follow in the same way.

The RH problems obtained in this paper also provide a starting point
for the study of the long-time behavior of solutions. After the
change of spatial variable, the dependence of the jump and residue
conditions on space and time is contained in the explicit phases
$\widehat f_2$ and $\widehat f_1$. This makes the resulting RH
problems suitable for the nonlinear steepest descent method. The
stationary points and branch points of these phases determine the
different asymptotic regions, while the discrete spectrum and the
reflection coefficients describe the soliton and continuous-spectrum
contributions, respectively. The asymptotics of $u$ and $v$ can then
be recovered parametrically from the corresponding RH expansions.
In the two-component problem, this final step must also incorporate
the transported profile $r^\sharp$, which distinguishes the
reconstruction from that of the scalar mCH equation.

\section{Acknowledgment}

IK acknowledges the support from the Austrian Science Fund (FWF), grant no.
10.55776/ESP691.
DS acknowledges the support from the National 
Research Foundation of Ukraine, grant no. 2025.07/0437.

\appendix

\section{Sign-preserving property of $m$ and $n$}\label{app:A}

The following result shows that the sign condition is preserved for as long
as a sufficiently regular solution exists (cf. \cite{KST22}).

\begin{proposition}\label{prop:positivity}
Let $(u,v)$ be a classical solution of \eqref{eq:2mch} on
$\R\times[0,T)$, $T\leq\infty$, and suppose that
\[
 Q=(u-u_x)(v+v_x)
 \in  C\bigl([0,T);C_b^1(\R)\bigr).
\]
If $m_0(x)>0$ and $n_0(x)>0$ for every $x\in\R$, then
\[
 m(x,t)>0,\qquad n(x,t)>0,
 \qquad (x,t)\in\R\times[0,T).
\]
\end{proposition}

\begin{proof}
Fix $T_0<T$. Since
$Q\in C([0,T);C_b^1(\R))$, we have
 $\sup_{0\leq t\leq T_0}
 \|Q(\cdot,t)\|_{C_b^1}<\infty$. Hence the characteristic problem
\begin{equation}\label{eq:characteristic-flow}
 \frac{d}{dt}\chi(t,\xi)=Q(\chi(t,\xi),t),
 \qquad \chi(0,\xi)=\xi,
\end{equation}
has a unique solution on $[0,T_0]$.

Differentiating \eqref{eq:characteristic-flow} with
respect to $\xi$, we obtain
\begin{equation}
    \label{eq:chi_xi}
     \frac{d}{dt}\chi_\xi(t,\xi)
 =Q_x(\chi(t,\xi),t)\chi_\xi(t,\xi),
 \qquad \chi_\xi(0,\xi)=1,
\end{equation}
and thus
\begin{equation}\label{eq:flow-jacobian}
    \chi_\xi(t,\xi)
=\exp\left(\int_0^t Q_x(\chi(s,\xi),s)\,ds\right)>0,
\end{equation}
It follows that, for each $t\in[0,T_0]$, the map
$\xi\mapsto\chi(t,\xi)$ is strictly increasing.  Moreover,
\[
 |\chi(t,\xi)-\xi|
 \leq\int_0^t\lVert Q(\cdot,s)\rVert_{L^\infty}\,ds
\]
which implies that $\chi(t,\xi)\to\pm\infty$ as $\xi\to\pm\infty$. Consequently, for
each $t\in[0,T_0]$, this map is a diffeomorphism from $\R$ onto $\R$.

Differentiating $m(\chi(t,\xi),t)$ and $n(\chi(t,\xi),t)$ with respect to $t$ using the chain rule, \eqref{eq:2mch}, and \eqref{eq:characteristic-flow}, we obtain that, along the characteristic flow,
\[
 \frac{d}{dt}m(\chi(t,\xi),t)
 =-Q_x(\chi(t,\xi),t)m(\chi(t,\xi),t),\qquad \frac{d}{dt}n(\chi(t,\xi),t)
 =-Q_x(\chi(t,\xi),t)n(\chi(t,\xi),t)
\]
Combining these identities with \eqref{eq:chi_xi} yields
\[
 \frac{d}{dt}\bigl[m(\chi(t,\xi),t)\chi_\xi(t,\xi)\bigr]=0,\qquad \frac{d}{dt}\bigl[n(\chi(t,\xi),t)\chi_\xi(t,\xi)\bigr]=0,
\]
and hence 
\begin{equation}\label{eq:mn-transport}
m(\chi(t,\xi),t)\chi_\xi(t,\xi)=m_0(\xi),\qquad n(\chi(t,\xi),t)\chi_\xi(t,\xi)=n_0(\xi).
\end{equation}
Equivalently,
\begin{align*}
 m(\chi(t,\xi),t)
 &=m_0(\xi)
   \exp\left(-\int_0^tQ_x(\chi(s,\xi),s)\,ds\right),\\
 n(\chi(t,\xi),t)
 &=n_0(\xi)
   \exp\left(-\int_0^tQ_x(\chi(s,\xi),s)\,ds\right).
\end{align*}
Since the exponential factor is strictly positive and the characteristic
flow is onto $\R$, the positivity of $m_0$ and $n_0$ implies \begin{equation*}
 m(x,t)>0,\qquad n(x,t)>0,
 \qquad x\in\R,\quad 0\leq t\leq T_0.
\end{equation*}
Since $T_0<T$ was arbitrary, the conclusion holds on
$\R\times[0,T)$.
\end{proof}

\begin{remark}
    The proof yields two additional observations.
\begin{enumerate}
    \item Since $m$ and $n$ are multiplied by the same factor along the
characteristic flow, their ratio is conserved:
\begin{equation}\label{eq:ratio-transport}
 \frac{m(\chi(t,\xi),t)}
 {n(\chi(t,\xi),t)}
 =
 \frac{m_0(\xi)}
 {n_0(\xi)}.
\end{equation}
Thus $m/n$ is constant along each characteristic.

\item
The proof of Proposition~\ref{prop:positivity} uses only the fact that $m$ and $n$ satisfy the conservation laws
\[
 m_t+(Qm)_x=0,
 \qquad
 n_t+(Qn)_x=0,
\]
and the stated regularity of
$Q$. No asymptotic boundary conditions on $u$ or $v$ are used.
In particular, it does not require any relation between the left and right
background values of $u$ and $v$. Hence the same argument applies when $u$ and $v$ have
distinct nonzero background values at spatial infinities.

\end{enumerate}
 
\end{remark}

\section{From a solution of the right RH problem
to a local solution of the 2-mCH system}
\label{sec:ic-rh-to-local-2mch}

We first express the 2-mCH system and its balanced Lax pair
in the new variables $(y,t)$. We then show that a
sufficiently regular solution of the right RH problem determines
a Lax pair whose compatibility condition yields the transformed
2-mCH system. Finally, after specifying the transported profile
$r^\sharp$, we reconstruct a local solution in the physical
variables.

\subsection{The balanced Lax pair in $(y,t)$ variables} We first reformulate the balanced Lax pair \eqref{ic:balanced-lax}
in the $(y,t)$ variables, where
$y=y(x,t)$ is defined by
\eqref{ic:right-hodograph-def}. We denote its inverse by
$x=x(y,t)$, and for a function $F(x,t)$, write
\[
 \widehat F(y,t):=F(x(y,t),t).
\]

The identities \eqref{ic:right-hodograph-derivatives} provide
\begin{equation}\label{ic:local-inverse-coordinate}
 x_y=\frac{A_2}{\widehat\rho},
 \qquad
 x_t=\widehat Q=\widehat p\,\widehat q.
\end{equation}

\begin{proposition}\label{prop:ic-local-balanced-lax}
The balanced eigenfunction $\widehat\Psi^{\mathrm b}(y,t,\lambda) :=
 \Psi^{\mathrm b}(x(y,t),t,\lambda)$
satisfies
\begin{subequations}\label{ic:local-balanced-lax}
\begin{align}
 \widehat\Psi^{\mathrm b}_y
 &=
 \left[
 \frac{A_2\lambda}{2}
 \begin{pmatrix}0&1\\-1&0\end{pmatrix}
 -
 \left(
 \frac{A_2}{2\widehat\rho}+\widehat r_y
 \right)\sigma_3
 \right]\widehat\Psi^{\mathrm b},
 \label{ic:local-balanced-lax-y}\\
 \widehat\Psi^{\mathrm b}_t
 &=
 \left[
 \frac1{\lambda^2}\sigma_3
 +
 \frac1\lambda
 \begin{pmatrix}
 0&-\widehat p\,\eul^{-2\widehat r}\\
 \widehat q\,\eul^{2\widehat r}&0
 \end{pmatrix}
 \right]\widehat\Psi^{\mathrm b}.
 \label{ic:local-balanced-lax-t}
\end{align}
\end{subequations}
\end{proposition}

\begin{proof}
Applying the chain rule to the balanced Lax pair
\eqref{ic:balanced-lax} gives the
coefficient matrices $x_yU^{\mathrm b}$ and
$V^{\mathrm b}+x_tU^{\mathrm b}$, evaluated at
$(x(y,t),t,\lambda)$.
Substituting \eqref{ic:local-inverse-coordinate} and using
$\widehat r_y=x_y\widehat{r_x}$ yields
\eqref{ic:local-balanced-lax}.
\end{proof}

\subsection{The 2-mCH system in $(y,t)$ variables}

We now rewrite the 2-mCH system in the variables $(y,t)$.
Since $x_t=\widehat Q$, the transport equation
$r_t+Qr_x=0$ becomes $\widehat r_t=0$.
The conservation law for $\rho$ and the definitions of
$p,q,m,n$ give the remaining evolution and spatial
relations stated below.

\begin{proposition}\label{prop:ic-local-2mch-y}
Suppose that $(u,v)$ is a classical solution of \eqref{eq:2mch}
with $m,n>0$. Define $\rho$, $r$, $p$ ,$q$, and $Q$ by
\eqref{ic:def-rhoRr}, \eqref{ic:def-pqQ}, and \eqref{ic:Q},
respectively, and introduce the new variables $(y,t)$
via \eqref{ic:right-hodograph-def}. Then
\begin{subequations}\label{ic:local-2mch-y}
\begin{align}
 \widehat\rho_t
 &=
 -\frac{\widehat\rho^2}{A_2}\widehat Q_y,
 \label{ic:local-rho-y}\\
 \widehat r_t
 &=0,
 \label{ic:local-r-y}\\
 \widehat p_y
 &=
 A_2\eul^{2\widehat r}
 -
 \frac{A_2}{\widehat\rho}\widehat p,
 \label{ic:local-p-y}\\
 \widehat q_y
 &=
 \frac{A_2}{\widehat\rho}\widehat q
 -
 A_2\eul^{-2\widehat r},
 \label{ic:local-q-y}\\
 \widehat u_y
 &=
 \frac{A_2}{\widehat\rho}
 (\widehat u-\widehat p),
 \label{ic:local-u-y}\\
 \widehat v_y
 &=
 \frac{A_2}{\widehat\rho}
 (\widehat q-\widehat v).
 \label{ic:local-v-y}
\end{align}
\end{subequations}
Here $\widehat Q=\widehat p\,\widehat q$.

Conversely, let $\widehat\rho,\widehat r,\widehat p,
\widehat q,\widehat u,\widehat v$ be real-valued $C^1$
functions, defined for $y\in\R$ and $0<t<T$, satisfying
\eqref{ic:local-2mch-y} and $\widehat\rho>0$.
There exists a $C^2$ function $x=x(y,t)$ on the same
domain, unique up to an additive constant, satisfying
\begin{equation}\label{ic:local-inverse-system}
 x_y=\frac{A_2}{\widehat\rho},
 \qquad
 x_t=\widehat Q,
\end{equation}
where $\widehat Q=\widehat p\,\widehat q$.
The map $(y,t)\mapsto(x(y,t),t)$ is a diffeomorphism
onto its image. 

Moreover, the functions $u$ and $v$ defined by
\[
 u(x(y,t),t)=\widehat u(y,t),
 \qquad
 v(x(y,t),t)=\widehat v(y,t)
\]
solve \eqref{eq:2mch} on this image, with positive momenta $m$ and $n$ satisfying
\[
 m(x(y,t),t)=\widehat\rho(y,t)\eul^{2\widehat r(y,t)},
 \qquad
 n(x(y,t),t)=\widehat\rho(y,t)\eul^{-2\widehat r(y,t)}.
\]
\end{proposition}

\begin{proof}
By the chain rule and \eqref{ic:local-inverse-coordinate},
for any smooth function $F$ we have
\[
 \begin{aligned}
 \widehat F_t(y,t)
 &=
 F_t(x(y,t),t)
 +
 Q(x(y,t),t)\,F_x(x(y,t),t),\qquad
 \widehat F_y(y,t)
 =
 \frac{A_2}{\widehat\rho(y,t)}
 F_x(x(y,t),t).
 \end{aligned}
\]
In particular, the conservation laws in
Lemma \ref{lem:rho-r-conservation} become
\[
 \widehat\rho_t
 =
 -\widehat\rho\,\widehat{Q_x},
 \qquad
 \widehat r_t=0.
\]
Together with
$\widehat Q_y
 =
 \frac{A_2}{\widehat\rho}\widehat{Q_x}$,
these identities yield \eqref{ic:local-rho-y} and \eqref{ic:local-r-y}.

Next, \eqref{eq:identities} yields
\[
 \widehat p_y
 =
 \frac{A_2}{\widehat\rho}
 (\widehat m-\widehat p),
 \qquad
 \widehat q_y
 =
 \frac{A_2}{\widehat\rho}
 (\widehat q-\widehat n).
\]
Using
\[
 \widehat m=\widehat\rho\,\eul^{2\widehat r},
 \qquad
 \widehat n=\widehat\rho\,\eul^{-2\widehat r},
\]
as follows from \eqref{ic:def-rhoRr}, we obtain
\eqref{ic:local-p-y} and \eqref{ic:local-q-y}.
 Finally, \eqref{ic:local-u-y} and \eqref{ic:local-v-y}
follow directly from \eqref{ic:def-pqQ} by the same
change of variables.

To construct the inverse change of variables, we first rewrite
\eqref{ic:local-rho-y} as
\begin{equation}\label{ic:local-inverse-compatibility}
 \partial_t\left(\frac{A_2}{\widehat\rho}\right)
 =
 \widehat Q_y,
\end{equation}
which is the compatibility condition for
\eqref{ic:local-inverse-system}. For arbitrary $y_*\in\R$, $t_*\in(0,T)$, and $x_*\in\R$,
the function
\begin{equation}\label{ic:local-inverse-formula}
 x(y,t)
 =
 x_*
 +
 \int_{y_*}^{y}
 \frac{A_2}{\widehat\rho(\eta,t)}\,\dd\eta
 +
 \int_{t_*}^{t}
 \widehat Q(y_*,\tau)\,\dd\tau
\end{equation}
is well defined for all $y\in\R$ and $0<t<T$.
Differentiating \eqref{ic:local-inverse-formula} and using
\eqref{ic:local-inverse-compatibility}, we obtain
\[
 x_y=\frac{A_2}{\widehat\rho},
 \qquad
 x_t
 =
 \int_{y_*}^{y}\widehat Q_y(\eta,t)\,\dd\eta
 +
 \widehat Q(y_*,t)
 =
 \widehat Q(y,t).
\]
Thus \eqref{ic:local-inverse-formula} defines a $C^2$
solution of \eqref{ic:local-inverse-system}, uniquely
determined by $x(y_*,t_*)=x_*$. Any other solution differs
from it by a constant.

For each fixed $t$, positivity of $x_y$ makes
$x(\cdot,t)$ strictly increasing. Consequently, the map
$(y,t)\mapsto(x(y,t),t)$ is injective. The inverse function theorem
therefore makes it a $C^2$ diffeomorphism onto its open image.

Returning to the variables $(x,t)$, we define
 $F(x,t):=\widehat F(y(x,t),t).$
By the chain rule and \eqref{ic:local-inverse-system},
the last four equations of \eqref{ic:local-2mch-y}
take the form
\[
 \begin{aligned}
 p_x&=\rho\eul^{2r}-p,\qquad
 &
 q_x&=q-\rho\eul^{-2r},\\
 u_x&=u-p,
 &
 v_x&=q-v.
 \end{aligned}
\]
It follows that
\[
 m:=\rho\eul^{2r}=p+p_x=u-u_{xx},
 \qquad
 n:=\rho\eul^{-2r}=q-q_x=v-v_{xx},
\]
while
\[
 Q=pq=(u-u_x)(v+v_x).
\]
Thus \eqref{ic:def-rhoRr}, \eqref{ic:def-pqQ},
and \eqref{ic:Q} hold.

Differentiating $\widehat m=\widehat\rho\,\eul^{2\widehat r}$
and $\widehat n=\widehat\rho\,\eul^{-2\widehat r}$ with
respect to $t$, and using \eqref{ic:local-rho-y} and
\eqref{ic:local-r-y}, we obtain
\[
 \widehat m_t
 =
 -\frac{\widehat\rho}{A_2}
 \widehat Q_y\,\widehat m,
 \qquad
 \widehat n_t
 =
 -\frac{\widehat\rho}{A_2}
 \widehat Q_y\,\widehat n.
\]
The chain rule and \eqref{ic:local-inverse-system} imply
\[
 \begin{aligned}
 \widehat{m_t+(Qm)_x}
 &=
 \widehat m_t
 +
 \frac{\widehat\rho}{A_2}\widehat Q_y\,\widehat m
 =
 0,\\
 \widehat{n_t+(Qn)_x}
 &=
 \widehat n_t
 +
 \frac{\widehat\rho}{A_2}\widehat Q_y\,\widehat n
 =
 0.
 \end{aligned}
\]
Together with the identities $m=u-u_{xx}$ and
$n=v-v_{xx}$ established above, this proves that
$(u,v)$ solves \eqref{eq:2mch}, with positive momenta
$m=\rho\eul^{2r}$ and $n=\rho\eul^{-2r}$.
\end{proof}

\begin{remark}
   Notice that \eqref{ic:local-rho-y} is equivalent to
\begin{equation}\label{ic:local-rho-inverse-y}
 (\widehat\rho^{-1})_t
 =
 \frac1{A_2}\widehat Q_y.
\end{equation}
\end{remark}

\subsection{From the RH problem to the Lax pair}

We now start from a solution $\widehat N(y,t,\lambda)$
of the right RH problem, without assuming that it arises
from a solution of the 2-mCH system.
In particular, $\widehat N$ satisfies the jump condition
\eqref{ic:Nhat-right-jump}, the one-sided residue conditions
\eqref{ic:Nhat-right-residues}, and the normalization
\eqref{ic:Nhat-right-normalization}.

We further assume that $\mathfrak r$ satisfies
\eqref{ic:right-reflection-outer-symmetry} and
\eqref{ic:right-reflection-inner-symmetry},  and that the
discrete points occur in symmetric pairs with real norming
constants satisfying \eqref{ic:norming-parity}. Thus the
uniqueness of the right RH problem implies the symmetry
relations \eqref{ic:Nhat-right-symmetries} for $\widehat N$.

For the argument below, we impose the following regularity assumptions on the right RH
solution $\widehat N$. All estimates below are locally uniform
in $(y,t)$.
\begin{enumerate}
\item
For $\lambda\notin\R$, $\widehat N$ is $C^1$ in $y$ and $t$.
Its boundary values and one-sided expansions at the discrete
spectral points may be differentiated once in $y$ or $t$.

\item
Near $\lambda=\infty$,
\begin{equation}\label{ic:local-N-infinity}
 \widehat N(y,t,\lambda)
 =
 I+\frac{\widehat N^{[1]}(y,t)}{\lambda}
 +\ord(\lambda^{-2}),
\end{equation}
where $\widehat N^{[1]}$ is $C^1$ in $(y,t)$ and is the
same in both half-planes. Furthermore,
\[
 \begin{aligned}
 \widehat N_y(y,t,\lambda)
 &=
 \frac{\partial_y\widehat N^{[1]}(y,t)}{\lambda}
 +\ord(\lambda^{-2}),\\
 \widehat N_t(y,t,\lambda)
 &=
 \frac{\partial_t\widehat N^{[1]}(y,t)}{\lambda}
 +\ord(\lambda^{-2}),
 \end{aligned}
 \qquad \lambda\to\infty.
\]

\item
Near $\lambda=0$, $\widehat N$, $\widehat N_y$, and
$\widehat N_t$ are bounded in the upper and lower half-planes.

\item
At every branch point
$\lambda_0\in\{\pm A_1^{-1},\pm A_2^{-1}\}$,
\[
 \widehat N,\ \widehat N_y,\ \widehat N_t
 =
 o\bigl(|\lambda-\lambda_0|^{-1/2}\bigr).
\]
\end{enumerate}

We next establish the expansion of $\widehat N$ at
$\lambda=0$. Its coefficients will be used to reconstruct
a local solution of the 2-mCH system in $(y,t)$ variables.

For fixed $(y,t)$, let $D$ be a sufficiently small disk centered at zero.
Define $K$ on $D\setminus\R$ by
\[
 K(\lambda)=
 \begin{cases}
  -\ii\mathcal E^{-1}\widehat N(y,t,\lambda), & \Im\lambda>0,\\
  -\mathcal E^{-1}\widehat N(y,t,\lambda)\sigma_1, & \Im\lambda<0.
 \end{cases}
\]
By \eqref{ic:Nhat-right-jump}, the two boundary values agree on
$(D\cap\R)\setminus\{0\}$, so $K$ extends holomorphically across
this interval. Since $K$ is bounded near zero, it extends uniquely
to a holomorphic function on all of $D$.

The symmetry relations
\eqref{ic:Nhat-right-symmetries} imply
\[
 K(\lambda)=-\sigma_3K(-\lambda)\sigma_3,
 \qquad
 K(\lambda)=\overline{K(\overline\lambda)}.
\]
The diagonal entries of $K$ are therefore odd in $\lambda$,
and the off-diagonal entries are even, with real coefficients. In particular, $K(0)$ has zero diagonal,
and $\det K(0)=-1$ gives
$K_{12}(0)K_{21}(0)=1$.
Writing these two entries as $\widehat a_1$ and
$\widehat a_1^{-1}$, and using
$K=-\ii\mathcal E^{-1}\widehat N$ in the upper half-disk,
we obtain
\begin{equation}\label{ic:local-N-zero}
 -\ii\mathcal E^{-1}\widehat N_+(y,t,\lambda)
 =
 \begin{pmatrix}
  0&\widehat a_1(y,t)\\
  \widehat a_1(y,t)^{-1}&0
 \end{pmatrix}
 +
 \lambda
 \begin{pmatrix}
  \widehat a_2(y,t)&0\\
  0&\widehat a_3(y,t)
 \end{pmatrix}
 +
 \ord(\lambda^2),
\end{equation}
where $\widehat a_j$ are real-valued and
$\widehat a_1\ne0$.

We now derive the Lax pair from $\widehat N$. For this perpose,
define
\begin{equation}\label{ic:local-Phi}
 \widehat\Phi
 :=
 \mathcal E^{-1}\widehat N
 \eul^{-\widehat f_2\sigma_3}.
\end{equation}

\begin{proposition}\label{prop:ic-local-generated-y}
The function $\widehat\Phi$ satisfies
\begin{equation}\label{ic:local-generated-y}
 \widehat\Phi_y
 =
 \doublehat U\,\widehat\Phi,
 \qquad
 \doublehat U
 =
 \frac{A_2\lambda}{2}
 \begin{pmatrix}0&1\\-1&0\end{pmatrix}
 -
 \frac{A_2\alpha}{2}\sigma_3,
\end{equation}
where
\begin{equation}\label{ic:local-alpha}
 \alpha
 :=
 2\bigl(\widehat N^{[1]}\bigr)_{12},
\end{equation}
and $\widehat N^{[1]}$ is the coefficient in
\eqref{ic:local-N-infinity}.
Moreover, the following relation holds
\begin{equation}\label{ic:local-alpha-a1}
 A_2\alpha
 =
 1-2\frac{\widehat a_{1,y}}{\widehat a_1},
\end{equation}
where $\widehat a_1$ is defined in
\eqref{ic:local-N-zero}.
\end{proposition}

\begin{proof}
The jump matrices of $\widehat\Phi$ are independent of $y$ and $t$. In particular, on the background cut this follows from $\widehat f_{2,-}=-\widehat f_{2,+}$
and $\sigma_1\sigma_3=-\sigma_3\sigma_1$. Hence
$L:=\widehat\Phi_y\widehat\Phi^{-1}$ has no jump on the contour.

In the upper half-disk about each $\lambda_k$, the residue condition gives
\[
 \widehat\Phi
 =F_k\begin{pmatrix}
  1 & 0\\
  \dfrac{\varkappa_k}{\lambda-\lambda_k} & 1
 \end{pmatrix},
\]
where $F_k$ is regular and invertible up to $\lambda_k$.
Since the triangular factor is independent of $y$,
$L=(F_k)_yF_k^{-1}$ is bounded from above. The same argument applied
to the second column gives boundedness from below. Thus $L$ extends
holomorphically through $\lambda_k$.

Near each branch point $\lambda_0=\pm\frac{1}{A_j}$, $j=1,2$,  the local estimates for $\widehat N$
and $\widehat N_y$ give $L=o\bigl(|\lambda-\lambda_0|^{-1}\bigr)$,
since $\det\widehat N=1$ and $\widehat f_{2,y}$ is bounded.
Hence the isolated singularity of $L$ at $\lambda_0$ is removable.

At $\lambda=\infty$, the symmetries imply that 
$\widehat N^{[1]}_{21}=\widehat N^{[1]}_{12}=\alpha/2$.
Using $\widehat f_{2,y}=\ii A_2\lambda/2+\ord(\lambda^{-1})$, we obtain
\[
 L=\frac{A_2\lambda}{2}\begin{pmatrix}0&1\\-1&0\end{pmatrix}
   -\frac{A_2\alpha}{2}\sigma_3+\ord(\lambda^{-1}).
\]

At $\lambda=0$,
 $\widehat f_{2,y}
 =
 -\frac12+\ord(\lambda^2)$,
and \eqref{ic:local-N-zero} implies that $L$ is bounded. Therefore,
\[
 L
 -
 \frac{A_2\lambda}{2}
 \begin{pmatrix}0&1\\-1&0\end{pmatrix}
 +
 \frac{A_2\alpha}{2}\sigma_3
\]
is entire and vanishes at infinity.
Liouville's theorem proves
\eqref{ic:local-generated-y}.

Finally, evaluation at $0_+$ gives
\[
 L(0_+)
 =
 \begin{pmatrix}
 \dfrac{\widehat a_{1,y}}{\widehat a_1}-\dfrac12&0\\[2mm]
 0&\dfrac12-\dfrac{\widehat a_{1,y}}{\widehat a_1}
 \end{pmatrix},
\]
which yields \eqref{ic:local-alpha-a1}.
\end{proof}

\begin{proposition}\label{prop:ic-local-generated-t}
The function $\widehat\Phi$ satisfies
\begin{equation}\label{ic:local-generated-t}
 \widehat\Phi_t
 =
 \doublehat V\,\widehat\Phi,
 \qquad
 \doublehat V
 =
 \frac1{\lambda^2}\sigma_3
 +
 \frac1\lambda
 \begin{pmatrix}0&\beta\\\gamma&0\end{pmatrix},
\end{equation}
where
\begin{equation}\label{ic:local-beta-gamma}
 \beta=-2\widehat a_1\widehat a_2,
 \qquad
 \gamma=2\widehat a_1^{-1}\widehat a_3
\end{equation}
with $\widehat a_j$, $j=1,2,3$ is defined in
\eqref{ic:local-N-zero}.
\end{proposition}

\begin{proof} Set
\[
 B:=\widehat\Phi_t\widehat\Phi^{-1}
 =\mathcal E^{-1}
  \bigl(\widehat N_t-\widehat f_{2,t}\widehat N\sigma_3\bigr)
  \widehat N^{-1}\mathcal E.
\]
Since the jump matrices and one-sided pole factors of $\widehat\Phi$
are independent of $t$, $B$ has no jump on the open arcs and extends
holomorphically through the discrete spectral points.

At each branch point $\lambda_0=\pm A_j^{-1}$, $j=1,2$, the phase
derivative $\widehat f_{2,t}$ is bounded. The same estimates as for $L$
imply $B=o\bigl(|\lambda-\lambda_0|^{-1}\bigr)$, and thus the singularity
is removable. At $\lambda=\infty$, \eqref{ic:local-N-infinity} and
$\widehat f_{2,t}=\ord(\lambda^{-1})$ yield $B=\ord(\lambda^{-1})$.

Near $\lambda=0$, the phase and its $t$-derivative satisfy
\begin{equation}\label{ic:local-phase-time}
\begin{aligned}
 \widehat f_2 &= t\lambda^{-2}-\tfrac12(y+A_2^2t)+\ord(\lambda^2),
 &\widehat f_{2,t} &= \lambda^{-2}-\tfrac12 A_2^2+\ord(\lambda^2)
 &&(\Im\lambda>0),\\
 \widehat f_2 &= -t\lambda^{-2}+\tfrac12(y+A_2^2t)+\ord(\lambda^2),
 &\widehat f_{2,t} &= -\lambda^{-2}+\tfrac12 A_2^2+\ord(\lambda^2)
 &&(\Im\lambda<0).
\end{aligned}
\end{equation}
On the cut near zero, $\widehat N_-=\ii\widehat N_+\sigma_1$.
Using \eqref{ic:local-N-zero} and $\sigma_1\sigma_3\sigma_1=-\sigma_3$,
we obtain
\[
 \mathcal E^{-1}\widehat N\sigma_3\widehat N^{-1}\mathcal E
 =\mp\begin{pmatrix}
       1 & \lambda\beta\\
       \lambda\gamma & -1
     \end{pmatrix}
   +\ord(\lambda^2),
 \qquad \pm\Im\lambda>0,
\]
where $\beta=-2\widehat a_1\widehat a_2$ and
$\gamma=2\widehat a_3/\widehat a_1$. Since
$\widehat N_t\widehat N^{-1}$ is bounded on both sides, substitution
of \eqref{ic:local-phase-time} gives the same expansion in both half-disks,
\[
 B=\frac{\sigma_3}{\lambda^2}
   +\frac{1}{\lambda}\begin{pmatrix}0&\beta\\ \gamma&0\end{pmatrix}
   +\ord(1).
\]
Subtracting this principal part leaves an entire matrix that vanishes
at infinity. Therefore, Liouville's theorem yields \eqref{ic:local-generated-t}.
\end{proof}

The compatibility condition for the Lax pair of $\widehat\Phi$ is
\[
 \doublehat U_t-\doublehat V_y+[\doublehat U,\doublehat V]=0.
\]
Equating the constant diagonal term and the off-diagonal terms of
order $\lambda^{-1}$ to zero yields the following system for
$\alpha$, $\beta$, and $\gamma$.

\begin{proposition}\label{prop:ic-local-compatibility}
The coefficients $\alpha,\beta,\gamma$ satisfy
\begin{subequations}\label{ic:local-coefficient-system}
\begin{align}
 \alpha_t&=\beta+\gamma,
 \label{ic:local-coefficient-alpha}\\
 \beta_y&=-A_2(\alpha\beta+1),
 \label{ic:local-coefficient-beta}\\
 \gamma_y&=A_2(\alpha\gamma-1).
 \label{ic:local-coefficient-gamma}
\end{align}
\end{subequations}
\end{proposition}

Unlike in the scalar case, we also need $\widehat r$ to recover
$\widehat\rho$, $\widehat p$, and $\widehat q$. Comparison with
\eqref{ic:local-balanced-lax} gives 
\begin{equation}\label{ic:local-coefficient-identification}
 \alpha
 =
 \frac1{\widehat\rho}
 +
 \frac{2}{A_2}\widehat r_y,
 \qquad
 \beta=-\widehat p\,\eul^{-2\widehat r},
 \qquad
 \gamma=\widehat q\,\eul^{2\widehat r}.
\end{equation}

Since $\widehat r_t=0$, it suffices to prescribe $\widehat r$ as a
function of $y$. To match \eqref{ic:rhoR-limits}, choose a real-valued
$r^\sharp\in C^3(\R)$ such that 
\[
 r^\sharp(y)\to0,\qquad (r^\sharp)'(y)\to0
 \qquad\text{as }y\to\pm\infty.
\]
Using \eqref{ic:local-coefficient-identification}, define
\begin{subequations}\label{ic:local-balanced-recovery}
\begin{align}
 \widehat r(y,t)
 &:=
 r^\sharp(y),\\
 \widehat\rho(y,t)^{-1}
 &:=
 \alpha(y,t)
 -
 \frac2{A_2}(r^\sharp)'(y),\\
 \widehat p(y,t)
 &:=
 -\beta(y,t)\eul^{2r^\sharp(y)},\\
 \widehat q(y,t)
 &:=
 \gamma(y,t)\eul^{-2r^\sharp(y)}.
\end{align}
\end{subequations}
Set
\[
 \widehat Q=\widehat p\,\widehat q=-\beta\gamma.
\]

\begin{remark}
For the Cauchy problem, the initial momenta $m_0$ and $n_0$ determine
$r^\sharp$. With $y_0(x)=y(x,0)$, \eqref{ic:advected-r-profile} gives
\[
 r^\sharp(y_0(x))=\frac14\log\frac{m_0(x)}{n_0(x)}.
\]
\end{remark}

\begin{corollary}\label{cor:ic-local-generated-y-system}
Whenever $\widehat\rho^{-1}>0$, the functions defined by
\eqref{ic:local-balanced-recovery} satisfy
\eqref{ic:local-rho-y}--\eqref{ic:local-q-y}.
\end{corollary}

\begin{proof}
Since $r^\sharp$ is independent of time, \eqref{ic:local-coefficient-alpha} yields
\[
 (\widehat\rho^{-1})_t
 =
 \alpha_t
 =
 \beta+\gamma.
\]
Differentiating $\widehat Q=-\beta\gamma$ in $y$ and using
\eqref{ic:local-coefficient-beta} and \eqref{ic:local-coefficient-gamma}, we obtain
\[
 \begin{aligned}
 \widehat Q_y
 =
 -\beta_y\gamma-\beta\gamma_y=
 A_2(\beta+\gamma).
 \end{aligned}
\]
Hence
\[
 (\widehat\rho^{-1})_t
 =
 \frac1{A_2}\widehat Q_y.
\]
Furthermore, differentiating $\widehat p=-\beta\eul^{2r^\sharp}$ with respect to $y$ and using \eqref{ic:local-coefficient-beta}
implies
\[
 \begin{aligned}
 \widehat p_y
 &=
 -\eul^{2r^\sharp}
 \bigl(\beta_y+2(r^\sharp)'\beta\bigr)=
 A_2\eul^{2r^\sharp}
 -
 \frac{A_2}{\widehat\rho}\widehat p.
 \end{aligned}
\]
Similarly,
\[
 \widehat q_y
 =
 \frac{A_2}{\widehat\rho}\widehat q
 -
 A_2\eul^{-2r^\sharp}.
\]
\end{proof}

Using $\widehat a_1$ from \eqref{ic:local-N-zero}, we define the
inverse coordinate $x(y,t)$ by
\begin{equation}\label{ic:local-x-definition}
 x(y,t)
 =
 y+A_2^2t
 -
 2\log\widehat a_1(y,t)
 -
 2r^\sharp(y).
\end{equation}

\begin{proposition}\label{prop:ic-local-coordinate-identities}
The function $x(y,t)$ defined by \eqref{ic:local-x-definition} satisfies
\begin{equation}\label{ic:local-coordinate-identities}
 x_y=\frac{A_2}{\widehat\rho},
 \qquad
 x_t=\widehat Q.
\end{equation}
No spatial endpoint normalization is needed for
these local identities.
\end{proposition}

\begin{proof}
Equation \eqref{ic:local-alpha-a1} gives
\[
 x_y
 =
 A_2\alpha-2(r^\sharp)'
 =
 \frac{A_2}{\widehat\rho}.
\]

Now, consider the time derivative. At fixed $y$, \eqref{ic:local-x-definition} gives
$x_t=A_2^2-2\widehat a_{1,t}/\widehat a_1$.
To determine $\widehat a_{1,t}$, write the second column of
\eqref{ic:local-generated-t} in terms of $\widehat N$:
\[
 \partial_t\bigl(\mathcal E^{-1}\widehat N^{(2)}\bigr)
 =\bigl(\doublehat V-\widehat f_{2,t}I\bigr)
   \mathcal E^{-1}\widehat N^{(2)}.
\]
By \eqref{ic:local-phase-time},
$\lambda^{-2}-\widehat f_{2,t}=\tfrac12 A_2^2+\ord(\lambda^2)$
as $\lambda\to0$ in the upper half-plane. Substituting
\eqref{ic:local-N-zero} and comparing the constant terms in the
first component gives
\[
 \widehat a_{1,t}=\tfrac12 A_2^2\widehat a_1+\beta\widehat a_3.
\]
Hence
\[
 x_t=-2\beta\frac{\widehat a_3}{\widehat a_1}
     =-\beta\gamma=\widehat Q.
\]
\end{proof}

We now return to the original variables $(x,t)$ and show how to recover
a local solution $(u,v)$ of \eqref{eq:2mch} from $\widehat p$,
$\widehat q$, and the map $y\mapsto x(y,t)$.

\begin{theorem}\label{thm:ic-local-rh-to-2mch}
Suppose that the preceding RH regularity assumptions
hold and that
\begin{equation}\label{ic:local-positivity}
 \alpha(y_0,t_0)
 -
 \frac2{A_2}(r^\sharp)'(y_0)>0
\end{equation}
at an interior point.

Then there is a small rectangle $J_0\times I_0$
on which $(y,t)\mapsto(x(y,t),t)$ is invertible.
For a fixed $y_*\in J_0$ and arbitrary real
$c_u,c_v\in C^2(I_0)$, define
\begin{equation}\label{ic:local-uv-reconstruction}
 \begin{aligned}
 \widehat u(y,t)
 &=
 \eul^{x(y,t)}
 \left[
 c_u(t)
 -
 \int_{y_*}^{y}
 \eul^{-x(\eta,t)}
 \widehat p(\eta,t)x_y(\eta,t)\,\dd\eta
 \right],\\
 \widehat v(y,t)
 &=
 \eul^{-x(y,t)}
 \left[
 c_v(t)
 +
 \int_{y_*}^{y}
 \eul^{x(\eta,t)}
 \widehat q(\eta,t)x_y(\eta,t)\,\dd\eta
 \right].
 \end{aligned}
\end{equation}
Set
\[
 \widehat m=\widehat\rho\,\eul^{2r^\sharp},
 \qquad
 \widehat n=\widehat\rho\,\eul^{-2r^\sharp}.
\]
Then
\[
 u(x(y,t),t)=\widehat u(y,t),
 \qquad
 v(x(y,t),t)=\widehat v(y,t)
\]
form a classical solution of \eqref{eq:2mch}
on the physical image of $J_0\times I_0$,
with positive momenta $m,n$.
\end{theorem}

\begin{proof}
By \eqref{ic:local-coordinate-identities} and \eqref{ic:local-positivity},
$x_y(y_0,t_0)>0$. The inverse function theorem then implies that
$(y,t)\mapsto(x(y,t),t)$ is locally invertible at $(y_0,t_0)$.

Differentiating
\eqref{ic:local-uv-reconstruction}, we obtain
\[
 \widehat u_y=x_y(\widehat u-\widehat p),
 \qquad
 \widehat v_y=x_y(\widehat q-\widehat v).
\]
Therefore, after returning to physical variables,
\[
 p=u-u_x,
 \qquad
 q=v+v_x.
\]
Equations~\eqref{ic:local-p-y} and~\eqref{ic:local-q-y} for
$\widehat p$ and $\widehat q$ imply
\[
 \begin{aligned}
 p+p_x
 &=
 \widehat p+\frac{\widehat p_y}{x_y}
 =
 \widehat\rho\,\eul^{2r^\sharp}
 =
 m,\\
 q-q_x
 &=
 \widehat q-\frac{\widehat q_y}{x_y}
 =
 \widehat\rho\,\eul^{-2r^\sharp}
 =
 n.
 \end{aligned}
\]
Thus, $m=u-u_{xx}$ and $n=v-v_{xx}$.

By \eqref{ic:local-coordinate-identities},
\[
 x_t=\widehat Q
 =
 \widehat p\,\widehat q
 =
 Q(x(y,t),t),
\]
so
\[
 x_{yt}=Q_xx_y.
\]
Moreover,
\[
 \widehat m x_y=A_2\eul^{2r^\sharp(y)},
 \qquad
 \widehat n x_y=A_2\eul^{-2r^\sharp(y)}
\]
are independent of time. Consequently,
\[
 \begin{aligned}
 0
 &=
 \partial_t(\widehat m x_y)=
 x_y\bigl(m_t+Qm_x+Q_xm\bigr),
 \end{aligned}
\]
with the physical derivatives evaluated at
$(x(y,t),t)$. This proves
$m_t+(Qm)_x=0$.
The same argument gives $n_t+(Qn)_x=0$.
Positivity follows from $\widehat\rho>0$.
\end{proof}

\begin{remark}
The functions $c_u,c_v$ are local integration
constants. They do not enter the RH problem:
changing them adds $c(t)\eul^x$ to $u$ and
$d(t)\eul^{-x}$ to $v$, leaving $p,q,m,n$ unchanged.
Spatial boundedness or boundary conditions can
fix these freedoms.
\end{remark}

\section{Reflectionless solutions on a constant background}
\label{app:rh-local-reflectionless}

\subsection{The RH problem with one discrete pair}

When $A_1<A_2$, \eqref{ic:right-reflection-inner-symmetry} implies
that $|\mathfrak r(\lambda)|=1$ for $A_2^{-1}<|\lambda|<A_1^{-1}$.
Thus $\mathfrak r$ can vanish on the entire continuous spectrum only
when the backgrounds agree. We denote their common value by
$A=A_1=A_2>0$, so that $u,v\to A$ as $x\to\pm\infty$. The two cuts
then coincide with $\Gamma=[-A^{-1},A^{-1}]$, and the reflectionless
condition reduces to $\mathfrak r(\lambda)=0$ for $|\lambda|>A^{-1}$.

Let $l(\lambda)=\sqrt{\lambda^2-A^{-2}}$ denote the branch with cut
$\Gamma$ and normalization $l(\lambda)\sim\lambda$ at infinity.
The phase \eqref{ic:right-phase-def} becomes
\[
 f(y,t,\lambda)
 =\frac{\ii A l(\lambda)}{2}\left(y-\frac{2t}{\lambda^2}\right),
\]
and the background matrix \eqref{ic:H-inverse} takes the form
\begin{equation}\label{app:rl-background}
 H^{-1}(\lambda)
 =\frac1{\sqrt2}\sqrt{\frac1{\ii A l(\lambda)}-1}
 \begin{pmatrix}
  \dfrac{A\lambda}{1-\ii A l(\lambda)}&1\\[2mm]
  1&\dfrac{A\lambda}{1-\ii A l(\lambda)}
 \end{pmatrix}.
\end{equation}
The outer square root has its cut on $[0,\infty)$ and satisfies
$\sqrt{-1}=\ii$. With this choice, $H^{-1}\to\mathcal E^{-1}$ at
infinity, and its boundary values satisfy
\[
 (H^{-1})_+=(H^{-1})_-(-\ii\sigma_1)
 \qquad (|\lambda|<A^{-1}).
\]
Near zero, the upper boundary value has the expansion
\[
 (H^{-1})_+(\lambda)
 =\ii\sigma_1+\frac{\ii A\lambda}{2}I+\ord(\lambda^2).
\]

We consider the case in which the discrete spectrum consists of a
single pair $\pm\mu$, with $0<\mu<A^{-1}$. By \eqref{ic:norming-parity}
and \eqref{ic:right-norming-integral}, the two points share a norming
constant $\varkappa<0$. The jump condition \eqref{ic:Nhat-right-jump}
reduces to
\[
 \widehat N_+(\lambda)=
 \begin{cases}
  \widehat N_-(\lambda), & |\lambda|>A^{-1},\\
  \widehat N_-(\lambda)(-\ii\sigma_1),
    & |\lambda|<A^{-1},\quad \lambda\ne\pm\mu.
 \end{cases}
\]

To write the residue conditions, we introduce
\begin{equation}\label{app:rl-explicit-parameters}
\begin{gathered}
 \delta=\sqrt{1-A^2\mu^2},\qquad
 z_0=-\frac{2\delta^2}{A\varkappa},\qquad
 z=z_0\eul^{\delta(y-2t/\mu^2)},\\
 d_\pm=z+1\pm\delta,\qquad
 F=d_+d_-=z^2+2z+A^2\mu^2.
\end{gathered}
\end{equation}
Since $0<\delta<1$ and $\varkappa<0$, the quantities $z$, $d_+$,
and $d_-$ are positive. At either discrete point $\xi\in\{-\mu,\mu\}$,
we have $l(\xi_\pm)=\pm\ii\delta/A$, so
\[
 2f(y,t,\xi_+)=-2f(y,t,\xi_-)
 =-\delta\left(y-\frac{2t}{\mu^2}\right).
\]
The two residue conditions in \eqref{ic:Nhat-right-residues}
therefore have the same coefficient $s$ and read
\begin{equation}\label{app:rl-explicit-residues}
 \begin{aligned}
  \Res_{\lambda=\xi_+}\widehat N^{(1)}
   &=s\,\widehat N^{(2)}(\xi_+),\\
  \Res_{\lambda=\xi_-}\widehat N^{(2)}
   &=s\,\widehat N^{(1)}(\xi_-),
 \end{aligned}
 \qquad
 s=\varkappa\eul^{-\delta(y-2t/\mu^2)}
   =-\frac{2\delta^2}{Az}.
\end{equation}

\subsection{Solution of the residue conditions}

Since $\widehat N$ and $H^{-1}$ have the same jump on $\Gamma$,
the matrix $D=\mathcal E^{-1}\widehat N H$ has no jump. At each
endpoint $\lambda_0=\pm A^{-1}$, the quarter-power bounds for
$\widehat N$ and $H$ imply
$D=\ord(|\lambda-\lambda_0|^{-1/2})$, so the singularity is
removable. Moreover, $D\to I$ at infinity. Its only remaining
singularities are the simple poles at $\pm\mu$, and hence
\begin{equation}\label{app:rl-dressing}
 \widehat N=\mathcal E D H^{-1},\qquad
 D(\lambda)=I+\frac{T_+}{\lambda-\mu}
              +\frac{T_-}{\lambda+\mu}.
\end{equation}

We determine $T_+$ and $T_-$ from the upper residue conditions at
$\mu$ and $-\mu$, respectively. Substituting \eqref{app:rl-dressing}
into \eqref{app:rl-explicit-residues} and multiplying on the left
by $\mathcal E^{-1}$ gives
\[
 \Res_{\lambda=\xi_+}(D H^{-1})^{(1)}
 =s\lim_{\lambda\to\xi_+}(D H^{-1})^{(2)}(\lambda),
 \qquad \xi=\pm\mu.
\]
The scalar prefactor of $H^{-1}$ in \eqref{app:rl-background} is
analytic and nonzero at each $\xi_+$. Its value at the pole multiplies
both sides and cancels. The upper residue conditions can therefore be
computed from the two columns of
\[
 D(\lambda)
 \left(\sigma_1+
 \frac{A\lambda}{1+\sqrt{1-A^2\lambda^2}}\,I\right),
\]
where the square root is positive for real $|\lambda|<A^{-1}$.

Write $T_\pm^{(1)}$ and $T_\pm^{(2)}$ for the columns of $T_\pm$.
At each pole, the second column of the product above is regular.
Its residue must therefore vanish:
\[
 T_\pm\binom{1}{\pm A\mu/(1+\delta)}=0.
\]
It follows that
\[
 T_\pm^{(1)}=\mp\frac{A\mu}{1+\delta}\,T_\pm^{(2)},
\]
so it suffices to find $T_+^{(2)}$ and $T_-^{(2)}$.

At $(\pm\mu)_+$, the residue of the first column is
\[
 T_\pm\binom{\pm A\mu/(1+\delta)}{1}
 =\frac{2\delta}{1+\delta}\,T_\pm^{(2)},
\]
where we used $A^2\mu^2=1-\delta^2$. At the same point, the value
of the regular second column is
\[
 \binom{1}{\pm A\mu/(1+\delta)}
 +\frac{A}{1+\delta}\,T_\mp^{(2)}
 +\frac{A}{\delta(1+\delta)}\,T_\pm^{(2)}.
\]
The middle term comes from evaluating the opposite-pole term of $D$
at $\lambda=\pm\mu$. The last term comes from multiplying
$T_\pm/(\lambda\mp\mu)$ by the linear term of the second column,
since
\[
 \left.
 \frac{\dd}{\dd\lambda}
 \frac{A\lambda}{1+\sqrt{1-A^2\lambda^2}}
 \right|_{\lambda=\pm\mu}
 =\frac{A}{\delta(1+\delta)}.
\]

Substituting
$s=-2\delta^2/(Az)$ and simplifying gives
\[
 \begin{aligned}
 (z+1)T_+^{(2)}+\delta T_-^{(2)}
  &=-\delta\binom{(1+\delta)/A}{\mu},\\
 \delta T_+^{(2)}+(z+1)T_-^{(2)}
  &=-\delta\binom{(1+\delta)/A}{-\mu}.
 \end{aligned}
\]
Adding and subtracting these equations yields
\[
 T_\pm^{(2)}
 =-\delta\binom{(1+\delta)/(A d_+)}{\pm\mu/d_-},
 \qquad d_\pm=z+1\pm\delta.
\]
The first columns now follow from the relation above. Hence
\begin{equation}\label{app:rl-explicit-T}
 T_+=\delta
 \begin{pmatrix}
  \mu/d_+&-(1+\delta)/(A d_+)\\
  (1-\delta)/(A d_-)&-\mu/d_-
 \end{pmatrix},
 \qquad
 T_-=-\sigma_3T_+\sigma_3.
\end{equation}

The jump on $\Gamma$ interchanges the columns, so the lower residue
conditions follow from the upper ones. The determinant has no pole
at either point because the residue of the singular column is
proportional to the value of the regular column. Thus
$\det\widehat N=\det D=1$ by normalization at infinity. Finally,
$D$ is regular at $\pm A^{-1}$, so $\widehat N$ inherits the
quarter-power bound of $H^{-1}$ at both endpoints.

Combining the two fractions in \eqref{app:rl-dressing}, we obtain
the one-pair reflectionless RH solution explicitly:
\begin{equation}\label{app:rl-explicit-N}
 \widehat N(y,t,\lambda)
 =
 \mathcal E
 \left[
 I+\frac{2\delta}{\lambda^2-\mu^2}
 \begin{pmatrix}
  \dfrac{\mu^2}{d_+}
  &-\dfrac{(1+\delta)\lambda}{A d_+}\\[3mm]
  \dfrac{(1-\delta)\lambda}{A d_-}
  &-\dfrac{\mu^2}{d_-}
 \end{pmatrix}
 \right]H^{-1}(\lambda),
\end{equation}
where $H^{-1}$ is given by \eqref{app:rl-background} and
$d_\pm=1\pm\delta+z_0\eul^{\delta(y-2t/\mu^2)}$.

\subsection{Reconstruction of the fields}

We expand the RH solution \eqref{app:rl-explicit-N} at $\lambda=0$
to determine the coefficients in \eqref{ic:local-N-zero}.
Recall that $d_\pm=z+1\pm\delta$, with
$\delta=\sqrt{1-A^2\mu^2}$. From \eqref{app:rl-dressing} and
\eqref{app:rl-explicit-T},
\[
 D(0)=\begin{pmatrix}d_-/d_+&0\\0&d_+/d_-\end{pmatrix},
 \qquad
 D'(0)=-\frac{T_++T_-}{\mu^2}.
\]
As $\lambda\to0$ in the upper half-plane,
$H^{-1}(\lambda)=\ii\sigma_1+\frac{\ii A\lambda}{2}I+\ord(\lambda^2)$.
Since $\mathcal E^{-1}\widehat N=D H^{-1}$, it follows that
\[
 -\ii\mathcal E^{-1}\widehat N
 =D(0)\sigma_1
  +\lambda\left(D'(0)\sigma_1+\frac A2D(0)\right)
  +\ord(\lambda^2).
\]
Comparing the upper boundary value with \eqref{ic:local-N-zero} yields
\begin{equation}\label{app:rl-explicit-aj}
 \widehat a_1=\frac{d_-}{d_+},\qquad
 \widehat a_2=\frac{A}{2d_+}
      \left(z+\frac{(1+\delta)^2}{1-\delta}\right),\qquad
 \widehat a_3=\frac{A}{2d_-}
      \left(z+\frac{(1-\delta)^2}{1+\delta}\right).
\end{equation}

Substituting $\widehat a_1=d_-/d_+$ into \eqref{ic:local-x-definition}
gives the physical coordinate
\begin{equation}\label{app:rl-explicit-x}
 x(y,t)=y+A^2t+2\log\frac{d_+}{d_-}-2r^\sharp(y).
\end{equation}
Since $z_y=\delta z$ and $z_t=-2\delta z/\mu^2$, differentiation gives
\begin{equation}\label{app:rl-explicit-xy}
 x_y=1-\frac{4\delta^2z}{d_+d_-}-2(r^\sharp)'(y)
\end{equation}
and
\begin{equation}\label{app:rl-explicit-Q}
 x_t=A^2+\frac{8\delta^2z}{\mu^2d_+d_-}
    =4\widehat a_2\widehat a_3.
\end{equation}
By \eqref{ic:local-beta-gamma} and \eqref{ic:local-balanced-recovery},
\begin{equation}\label{app:rl-physical-pq}
 \widehat p=2\eul^{2r^\sharp}\widehat a_1\widehat a_2,\qquad
 \widehat q=2\eul^{-2r^\sharp}\frac{\widehat a_3}{\widehat a_1}.
\end{equation}
Theorem~\ref{thm:ic-local-rh-to-2mch} then gives a local solution
whenever $x_y>0$.

To obtain a solution on the whole line, we estimate $x_y$ from below.
Since $d_+d_-=z^2+2z+A^2\mu^2$ and $z>0$,
\[
 \frac{z}{d_+d_-}
 =\frac{1}{z+2+A^2\mu^2/z}
 \le \frac{1}{2(1+A\mu)},
\]
with equality at $z=A\mu$. Using $\delta^2=1-A^2\mu^2$ in
\eqref{app:rl-explicit-xy}, we obtain
\[
 x_y(y,t)\ge 2A\mu-1-2\|(r^\sharp)'\|_\infty.
\]
Thus the conditions
\begin{equation}\label{app:rl-smooth-range}
 \frac12<A\mu<1,\qquad
 r^\sharp\in\mathcal S(\R;\R),\qquad
 2\|(r^\sharp)'\|_\infty<2A\mu-1
\end{equation}
ensure that $x_y$ is uniformly positive. For each fixed $t$,
\eqref{app:rl-explicit-x} also gives $x(y,t)=y+A^2t+\ord(1)$ as
$y\to\pm\infty$. Hence $y\mapsto x(y,t)$ is invertible on $\R$.

The momenta are
\[
 \widehat m=\frac{A\eul^{2r^\sharp}}{x_y},\qquad
 \widehat n=\frac{A\eul^{-2r^\sharp}}{x_y}.
\]
To recover the bounded fields, we integrate
$(\eul^{-x}u)_x=-\eul^{-x}p$ from $x$ to $+\infty$ and
$(\eul^xv)_x=\eul^xq$ from $-\infty$ to $x$. In the variable $y$,
the resulting formulas are
\begin{equation}\label{app:rl-bounded-fields}
\begin{aligned}
 u(x(y,t),t)
 &=\int_y^\infty
   \eul^{x(y,t)-x(\eta,t)}
   \widehat p(\eta,t)x_\eta(\eta,t)\,\dd\eta,\\
 v(x(y,t),t)
 &=\int_{-\infty}^y
   \eul^{x(\eta,t)-x(y,t)}
   \widehat q(\eta,t)x_\eta(\eta,t)\,\dd\eta.
\end{aligned}
\end{equation}
As $y\to\pm\infty$, we have $r^\sharp\to0$, $x_y\to1$, and
$\widehat p,\widehat q\to A$. Thus both fields and both momenta
tend to $A$ at each spatial end.

\subsection{Dependence on the profile $r^\sharp$ and the scalar soliton}

By \eqref{app:rl-explicit-x}, the discrete data determine
$x(y,t)+2r^\sharp(y)$. Specifying the profile $r^\sharp$ then determines $x$
and the momenta $m$ and $n$. To clarify why different $r^\sharp$ can correspond to
the same scattering data, recall that $\widehat\rho\,x_y=A$ and
$r_x(x(y,t),t)x_y=(r^\sharp)'(y)$. Hence \eqref{app:rl-explicit-xy}
gives
\begin{equation}\label{app:rl-profile-cancellation}
 \frac{1+2r_x(x(y,t),t)}{\rho(x(y,t),t)}
 =\frac{x_y+2(r^\sharp)'(y)}{A}
 =\frac1A\left(1-\frac{4\delta^2z}{d_+d_-}\right),
 \qquad \rho\,\dd x=A\,\dd y.
\end{equation}
Thus the coefficients of the balanced spatial equation
\eqref{ic:local-balanced-lax-y} are independent of $r^\sharp$.
Since $r^\sharp(y)\to0$ as $y\to\pm\infty$, the limits of
$x(y,t)-y$ at both ends are also independent of the profile.
The normalized Jost solutions, therefore, agree in the $y$ variable.
They have the same scattering coefficients and, by
\eqref{ic:c11-derivative-balanced} and $\rho\,\dd x=A\,\dd y$,
the same norming constants. 

On the other hand, for prescribed initial momenta
$m_0$ and $n_0$,
$r^\sharp$ is fixed by \eqref{ic:advected-r-profile}.
Thus, different initial data $m_0$ and $n_0$ may produce the same spectral picture,
with the same continuous and discrete parts.

\begin{figure}[ht]
 \centering
 \includegraphics[width=\linewidth]{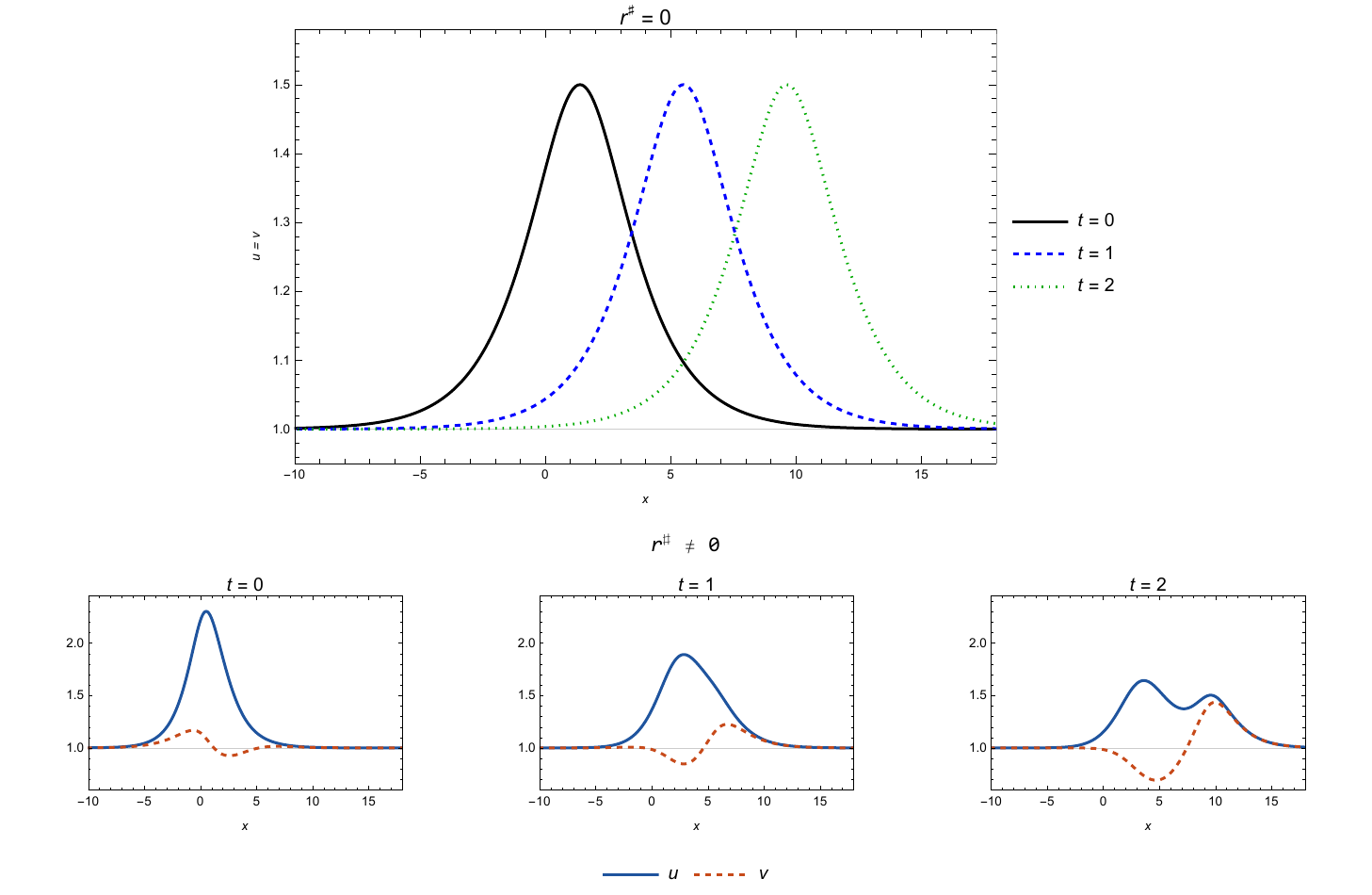}
 \caption{One-pair reflectionless solutions for $A=1$, $\mu=4/5$,
 and $\varkappa=-9/10$. The upper panel shows the scalar soliton
 $u=v$ for $r^\sharp=0$ at $t=0,1,2$. The lower panels show
 $u$ (solid) and $v$ (dashed) for
 $r^\sharp(y)=\tfrac14\eul^{-(y/3)^2}$ at $t=0,1,2$,
 from left to right.}
 \label{fig:rl-profile-comparison}
\end{figure}

Consider first $r^\sharp\equiv0$. Then $m=n$, so $u-v$ is a bounded
solution of $(u-v)_{xx}=u-v$ on $\R$. Hence $u=v$. Since
$p=u-u_x$ and $q=u+u_x$, \eqref{app:rl-physical-pq} gives
\[
 \widehat u=\widehat v
 =\frac{\widehat p+\widehat q}{2}
 =\widehat a_1\widehat a_2+\frac{\widehat a_3}{\widehat a_1}.
\]
Substituting \eqref{app:rl-explicit-aj}, we obtain
\begin{equation}\label{app:rl-scalar-wave}
 u(x(y,t),t)=v(x(y,t),t)
 =A\left[
 1+\frac{4\delta^2z\bigl(z^2+2A^2\mu^2z+A^2\mu^2\bigr)}
          {A^2\mu^2(d_+d_-)^2}
 \right].
\end{equation}
The right-hand side depends only on $y-2t/\mu^2$. For
$c=A^2+2/\mu^2$, \eqref{app:rl-explicit-x} gives
\[
 x(y,t)-ct=y-\frac{2t}{\mu^2}+2\log\frac{d_+}{d_-}.
\]
This too is a function of $y-2t/\mu^2$, and it is strictly increasing
by \eqref{app:rl-smooth-range}. Thus the common field in
\eqref{app:rl-scalar-wave} is a travelling soliton of speed
$c=A^2+2/\mu^2$ in the physical coordinate $x$.

\begin{proposition}\label{app:rl-nonscalar}
Under \eqref{app:rl-smooth-range}, every nonzero profile $r^\sharp$
produces a reflectionless solution with $u\not\equiv v$.
The two components cannot travel with a common speed.
\end{proposition}

\begin{proof}
Since $\widehat m/\widehat n=\eul^{4r^\sharp}$, a nonzero profile
implies $m\not\equiv n$ and hence $u\not\equiv v$.

Suppose that both fields depend only on $x-ct$. Integrating the
conservation laws \eqref{eq:2mch} in this variable and using
$m,n\to A$ and $Q\to A^2$ at infinity, we obtain
\[
 (Q-c)m=(Q-c)n=A(A^2-c).
\]
If $c\ne A^2$, the right-hand side is nonzero, so $Q-c$ never
vanishes and $m=n$. If $c=A^2$, positivity of $m$ implies
$Q=A^2$ everywhere. This contradicts \eqref{app:rl-explicit-Q},
which gives
\[
 Q(x(y,t),t)=x_t=A^2+\frac{8\delta^2z}{\mu^2d_+d_-}>A^2.
\]
Both cases are therefore excluded.
\end{proof}

Figure~\ref{fig:rl-profile-comparison} shows how the choice of
$r^\sharp$ changes the solution for the same continuous and discrete data.

\bibliography{references}

@article{BKS20,
  author = {Boutet de Monvel, Anne and Karpenko, Iryna and Shepelsky, Dmitry},
  title = {{A Riemann-Hilbert approach to the modified Camassa--Holm equation with nonzero boundary conditions}},
  journal = {J. Math. Phys.},
  volume = {61},
  year = {2020},
  number = {3},
  pages = {031504, 24},
}

@article{BS06,
  author = {Boutet de Monvel, Anne and Shepelsky, Dmitry},
  title = {{Riemann--Hilbert approach for the Camassa--Holm equation on the line}},
  journal = {C. R. Math. Acad. Sci. Paris},
  volume = {343},
  year = {2006},
  number = {10},
  pages = {627--632},
}

@article{BS08,
  author = {Boutet de Monvel, Anne and Shepelsky, Dmitry},
  title = {{Riemann-Hilbert problem in the inverse scattering for the Camassa-Holm equation on the line}},
  booktitle = {{Probability, geometry and integrable systems}},
  series = {Math. Sci. Res. Inst. Publ.},
  volume = {55},
  publisher = {Cambridge Univ. Press},
  address = {Cambridge},
  year = {2008},
  pages = {53--75},
}

@article{BS13,
  author = {Boutet de Monvel, Anne and Shepelsky, Dmitry},
  title = {{A Riemann--Hilbert approach for the Degasperis--Procesi equation}},
  journal = {Nonlinearity},
  volume = {26},
  year = {2013},
  number = {7},
  pages = {2081--2107},
}

@article{BS15,
  author = {Boutet de Monvel, Anne and Shepelsky, Dmitry},
  title = {{The Ostrovsky--Vakhnenko equation by a Riemann--Hilbert approach}},
  journal = {J. Phys. A},
  volume = {48},
  year = {2015},
  number = {3},
  pages = {035204, 34},
}

@article{CH93,
  author = {Camassa, Roberto and Holm, Darryl D.},
  title = {{An integrable shallow water equation with peaked solitons}},
  journal = {Phys. Rev. Lett.},
  volume = {71},
  year = {1993},
  number = {11},
  pages = {1661--1664},
}

@article{Karpenko2022,
  author  = {Karpenko, I.},
  title   = {{Long-time asymptotics for the modified Camassa--Holm equation with nonzero boundary conditions}},
  journal = {Journal of Mathematical Physics, Analysis, Geometry},
  volume  = {18},
  number  = {2},
  pages   = {224--252},
  year    = {2022},
  doi     = {10.15407/mag18.02.224}
}

@article{KST22,
  author = {Karpenko, Iryna and Shepelsky, Dmitry and Teschl, Gerald},
  title = {{A Riemann--Hilbert approach to the modified Camassa--Holm equation with step-like boundary conditions}},
  journal = {Monatshefte f\"ur Mathematik},
  volume = {201},
  year = {2023},
  pages = {127–172},
}

@article{OR96,
  author = {Olver, P. J. and Rosenau, P.},
  title = {{Tri-hamiltonian duality between solitons and solitary-wave solutions having compact support}},
  journal = {Phys. Rev. E},
  volume = {53},
  year = {1996},
  number = {2},
  pages = {1900},
}

@article{SQQ11,
    author = {J. Song and C. Qu and Z. Qiao}, 
    title = {A new integrable two-component system with cubic nonlinearity},
    journal = {J. Math. Phys.},
    year = 2011,
    volume = {52, 1 },
    pages = {013503}
}

@book{Ahlfors79,
  author    = {Ahlfors, Lars V.},
  title     = {Complex Analysis},
  edition   = {3},
  publisher = {McGraw--Hill},
  address   = {New York},
  year      = {1979}
}

@article{YanQiaoZhang2015,
  author  = {Yan, K. and Qiao, Z. and Zhang, Y.},
  title   = {{Blow-up phenomena for an integrable two-component Camassa--Holm system with cubic nonlinearity and peakon solutions}},
  journal = {Journal of Differential Equations},
  volume  = {259},
  pages   = {6644--6671},
  year    = {2015}
}

@article{Fokas95,
  author  = {Fokas, Athanassios S.},
  title   = {On a class of physically important integrable equations},
  journal = {Physica D},
  volume  = {87},
  number  = {1--4},
  pages   = {145--150},
  year    = {1995},
  doi     = {10.1016/0167-2789(95)00133-O}
}

@article{Fuchssteiner96,
  author  = {Fuchssteiner, Benno},
  title   = {Some tricks from the symmetry-toolbox for nonlinear
             equations: Generalizations of the {Camassa--Holm}
             equation},
  journal = {Physica D},
  volume  = {95},
  number  = {3--4},
  pages   = {229--243},
  year    = {1996},
  doi     = {10.1016/0167-2789(96)00048-6}
}

@article{Qiao06,
  author  = {Qiao, Zhijun},
  title   = {A new integrable equation with cuspons and
             {W/M}-shape-peaks solitons},
  journal = {Journal of Mathematical Physics},
  volume  = {47},
  number  = {11},
  pages   = {112701},
  year    = {2006},
  doi     = {10.1063/1.2365758}
}

@article{CHS16,
  author  = {Chang, Xiang-Ke and Hu, Xing-Biao and Szmigielski, Jacek},
  title   = {Multipeakons of a two-component modified
             {Camassa--Holm} equation and the relation with the
             finite {Kac--van Moerbeke} lattice},
  journal = {Advances in Mathematics},
  volume  = {299},
  pages   = {1--35},
  year    = {2016},
  doi     = {10.1016/j.aim.2016.05.004}
}

@article{Manakov74,
  author  = {Manakov, S. V.},
  title   = {On the theory of two-dimensional stationary
             self-focusing of electromagnetic waves},
  journal = {Soviet Physics JETP},
  volume  = {38},
  number  = {2},
  pages   = {248--253},
  year    = {1974}
}

@article{PAB06,
  author  = {Prinari, Barbara and Ablowitz, Mark J.
             and Biondini, Gino},
  title   = {Inverse scattering transform for the vector nonlinear
             {Schr{\"o}dinger} equation with nonvanishing boundary
             conditions},
  journal = {Journal of Mathematical Physics},
  volume  = {47},
  number  = {6},
  pages   = {063508},
  year    = {2006},
  doi     = {10.1063/1.2209169}
}

@article{PBT11,
  author  = {Prinari, Barbara and Biondini, Gino
             and Trubatch, A. David},
  title   = {Inverse scattering transform for the multi-component
             nonlinear {Schr{\"o}dinger} equation with nonzero
             boundary conditions},
  journal = {Studies in Applied Mathematics},
  volume  = {126},
  number  = {3},
  pages   = {245--302},
  year    = {2011},
  doi     = {10.1111/j.1467-9590.2010.00504.x}
}

@article{PDLH18,
  author  = {Prinari, Barbara and Demontis, Francesco and Li, Sitai
             and Horikis, Theodoros P.},
  title   = {Inverse scattering transform and soliton solutions for
             square matrix nonlinear {Schr{\"o}dinger} equations
             with non-zero boundary conditions},
  journal = {Physica D},
  volume  = {368},
  pages   = {22--49},
  year    = {2018},
  doi     = {10.1016/j.physd.2017.12.007}
}

\end{document}